\pdfoutput=1
\RequirePackage{ifpdf}
\ifpdf 
\documentclass[pdftex]{sigma}
\else
\documentclass{sigma}
\fi

\numberwithin{equation}{section}

\newtheorem{Theorem}{Theorem}[section]
\newtheorem{Corollary}[Theorem]{Corollary}
\newtheorem{Lemma}[Theorem]{Lemma}
 { \theoremstyle{definition}
\newtheorem{Remark}[Theorem]{Remark} }

\newcommand{\C}{{\mathbb C}}

\newcommand{\Z}{{\mathbb Z}}

\newcommand{\T}{{\mathbb T}}

\newcommand{\al}{\alpha}
\newcommand{\be}{\beta}

\newcommand{\Ga}{\Gamma}
\newcommand{\La}{\Lambda}
\newcommand{\la}{\lambda}

\newcommand{\de}{\delta}

\newcommand{\Om}{\Omega}
\newcommand{\ze}{\zeta}

\newcommand{\di}{\displaystyle}

\def\bigO{{\cal O}}
\usepackage{tikz}
\usepackage{pict2e}
\usepackage{todonotes}
\usetikzlibrary{decorations.markings}
\usetikzlibrary{calc,patterns,angles,quotes}
\usetikzlibrary{arrows,calc,chains, positioning, shapes.geometric,shapes.symbols,decorations.markings,arrows.meta}

\tikzset{middlearrow/.style={
decoration={markings,
mark= at position 0.6 with {\arrow{#1}} ,
},
postaction={decorate}
}
}

\tikzset{->-/.style={decoration={
markings,
mark=at position #1 with {\arrow{latex}}},postaction={decorate}}}

\tikzset{-<-/.style={decoration={
markings,
mark=at position #1 with {\arrowreversed{latex}}},postaction={decorate}}}

\begin{document}

\allowdisplaybreaks

\newcommand{\arXivNumber}{1909.00963}

\renewcommand{\PaperNumber}{100}

\FirstPageHeading

\ShortArticleName{A Riemann--Hilbert Approach to Asymptotic Analysis of Toeplitz+Hankel Determinants}

\ArticleName{A Riemann--Hilbert Approach to Asymptotic Analysis\\ of Toeplitz+Hankel Determinants}

\Author{Roozbeh GHARAKHLOO~$^\dag$ and Alexander ITS~$^{\ddag\S}$}

\AuthorNameForHeading{R.~Gharakhloo and A.~Its}

\Address{$^\dag$~Department of Mathematics, Colorado State University,\\
\hphantom{$^\dag$}~1874 Campus Delivery, Fort Collins, CO 80523-1874, USA}
\EmailD{\href{mailto: roozbeh.gharakhloo@colostate.edu}{roozbeh.gharakhloo@colostate.edu}}
\URLaddressD{\url{https://www.math.colostate.edu/~gharakhl}}

\Address{$^\ddag$~Department of Mathematical Sciences, Indiana University-Purdue University Indianapolis,\\
\hphantom{$^\ddag$}~402 N.~Blackford St., Indianapolis, IN 46202, USA}
\EmailD{\href{mailto:aits@iupui.edu}{aits@iupui.edu}}

\Address{$^\S$~St.~Petersburg State University, Universitetskaya emb.~7/9, 199034, St.~Petersburg, Russia}

\ArticleDates{Received November 03, 2019, in final form August 26, 2020; Published online October 06, 2020}

\Abstract{In this paper we will formulate $4\times4$ Riemann--Hilbert problems for Toep\-litz+Han\-kel determinants and the associated system of orthogonal polynomials, when the Hankel symbol is supported on the unit circle and also when it is supported on an interval~$[a,b]$, $0<a<b<1$. The distinguishing feature of this work is that in the formulation of~the Riemann--Hilbert problem no specific relationship is assumed between the Toeplitz and Hankel symbols. We will develop nonlinear steepest descent methods for analysing these problems in the case where the symbols are smooth (i.e., in the absence of Fisher--Hartwig singularities) and admit an analytic continuation in a neighborhood of the unit circle (if~the symbol's support is the unit circle). We will finally introduce a model problem and will present its solution requiring certain conditions on the ratio of Hankel and Toeplitz symbols. This in turn will allow us to find the asymptotics of the norms $h_n$ of the corresponding orthogonal polynomials and, in fact, the large $n$ asymptotics of the polynomials themselves. We will explain how this solvable case is related to the recent operator-theoretic approach in~[Basor~E., Ehrhardt~T., in Large Truncated {T}oeplitz Matrices, {T}oeplitz Operators, and Related Topics, \textit{Oper. Theory Adv. Appl.}, Vol.~259, Birkh\"auser/Springer, Cham, 2017, 125--154, arXiv:1603.00506] to Toeplitz+Hankel determinants. At the end we will discuss the prospects of future work and outline several technical, as well as conceptual, issues which we are going to~address next within the $4\times 4$ Riemann--Hilbert framework introduced in this paper.}


\Classification{15B05; 30E15; 35Q15}

\section{Introduction and preliminaries}

The $n \times n$ Toeplitz and Hankel matrices associated respectively to the symbols $\phi$ and $w$, supported on the unit circle $\mathbb{T}$ are respectively defined as
\begin{gather*}
T_{n}[\phi;r]:=\{\phi_{j-k+r}\}, \qquad
j,k=0,\dots, n-1, \qquad
\phi_k=\int_{\T} z^{-k}\phi(z) \frac{{\rm d}z}{2\pi {\rm i} z},
\end{gather*}
and
\begin{gather}\label{Hankel}
H_{n}[w;s]:=\{w_{j+k+s}\}, \qquad
j,k=0,\dots, n-1, \qquad
w_k=\int_{\T} z^{-k}w(z) \frac{{\rm d}z}{2\pi {\rm i} z},
\end{gather}
for fixed \textit{offset} values $r,s \in \Z$. If the Hankel symbol $w$ is supported on a subset $I$ of the real line, then $w_k$ in \eqref{Hankel} are instead given by
\begin{gather}\label{moments of w}
w_k=\int_{I} x^{k}w(x)\, {\rm d}x.
\end{gather}

The Toeplitz and Hankel determinants characterize important objects particularly in random matrix theory, statistical mechanics, theory of orthogonal polynomials, theory of Fredhom determinants, etc. For more on the history
of the development of the theory of Toeplitz and Hankel determinants and their numerous applications we refer the reader to the review articles~\cite{DIK1} and~\cite{Krasovsky1}. We also refer to the monographs~\cite{Bottcher-Silbermann} and~\cite{Bottcher-Silbermann1} as the main sources for general facts concerning the theory of Toeplitz matrices and operators.
The asymptotic results concerning the Hankel determinants and their applications~-- both recent and classical,
are featured in the papers \cite{Charlier,CharlierGharakhloo,DIK1,ItsKrasovsky,Krasovsky} and in the references therein.

There has been a growing interest in the asymptotics of Toeplitz+Hankel determinants in~recent years. A Toeplitz+Hankel matrix is naturally the sum $T_n[\phi;r]+H_n[w;s]$, and thus it has a determinant of the form
\begin{gather}
\label{T+H Determinant}
D_n(\phi,w;r,s) := \det
\begin{pmatrix}
\phi_r+w_s & \phi_{r-1}+w_{s+1} & \cdots & \phi_{r-n+1}+w_{s+n-1} \\
\phi_{r+1}+w_{s+1} & \phi_{r}+w_{s+2} & \cdots & \phi_{r-n+2}+w_{s+n} \\
\vdots & \vdots & \ddots & \vdots\\
\phi_{r+n-1}+w_{s+n-1} & \phi_{r+n-2}+w_{s+n} & \cdots & \phi_{r}+w_{s+2n-2}
\end{pmatrix}\!,
\end{gather}
$r,s \in \Z$, where, naturally, $\phi$ and $w$ respectively denote the Toeplitz and Hankel symbols. Although there are no results in the literature for the Toeplitz+Hankel determinants where $w$ is supported on the line, the case where $w$ is supported on the unit circle has been considered under specific assumptions. E.~Basor and T.~Ehrhardt have studied different aspects of these determinants in a series of papers \cite{BE,BE1,BE2,BE3,BE4} via operator-theoretic tools over the last 20 years or so. In \cite{DIK}, the Riemann--Hilbert technique which has already been proven very effective to~study the asymptotics of Toeplitz and Hankel determinants was extended for the first time to~the determinants of Toeplitz+Hankel matrices generated by the same symbol $w=\phi$, where the Hankel weight is supported on $\T$. In that work the symbol was assumed to be of Fisher--Hartwig type and it was further required that the symbol be \textit{even}, i.e., $w=\tilde{w}$.\footnote{\textbf{Notation.} Throughout the paper we will frequently use the notation $\Tilde{f}(z)$, to denote~$f\big(z^{-1}\big)$.} In \cite{BE}, by employing the relevant results in~\cite{DIK}, the authors managed for the first time to find the asymptotics of~Toeplitz+Hankel determinants for certain \textit{non-coinciding} symbols. Indeed, they considered
\begin{gather}\label{BEsymbols}
\phi(z)=c(z)\phi_0(z), \qquad
\mbox{and} \qquad
w(z)=c(z)d(z)w_0(z),
\end{gather}
where the functions $c$ and $d$ are assumed to be smooth and nonvanishing on the unit circle with zero winding number. Neither $c$ nor $d$ are assumed to be even functions but it is further required that $d$ satisfies the conditions $d\tilde{d}=1$ (on the unit circle) and $d(\pm 1)=1$. Furthermore, $\phi_0$ is assumed to be an even function of FH type and $w_0$ is related to $\phi_0$ in one of the following four ways: a) $w_0(z)=\pm \phi_0(z)$, b) $w_0(z)=z\phi_0(z)$ and c) $w_0(z)=\di -z^{-1}\phi_0(z)$.

Since the Riemann--Hilbert analysis carried out in~\cite{DIK} does not allow for different symbols, the~primary goal of this paper is to develop a Riemann--Hilbert framework for asymptotic analysis of~Toeplitz+Hankel determinants $D_n(\phi,w;r,s)$ where $\phi$ and $w$ are not a priori related, at least at~the level of formulation of the problem. Indeed, asymptotics of Toeplitz+Hankel determinants with different symbols are interesting for several reasons that prompts this research project. For~example, the type (\ref{BEsymbols}) of Toeplitz+Hankel determinants has appeared in the very recent work~\cite{Chelkak} in connection with the analysis of Ising model on the $45^{\circ}$ rotated half-plane, or the so-called \textit{zig-zag} half-plane.

Perhaps, our most important motivation behind studying Toeplitz+Hankel determinants is~to~study the large $n$ asymptotics of the eigenvalues of the Hankel matrix $H_n[w]$ associated to~the~symbol $w$. Specifically, we want
to extend the
recent results \cite{DIK2} concerning the spectral asymptotics of the Toeplitz matrices{\footnote{ The large $n$ behavior of the individual
eigenvalues of Toeplitz matrices has been also addressed in a number of works -- see \cite{BotGrud} and references therein.}} to the Hankel case.
The key feature which allows an effective asymptotic spectral analysis of Toeplitz matrices and, in particular, the use of the Riemann--Hilbert method, is that the characteristic polynomial of a Toeplitz
matrix is again a~Toeplitz determinant with a symbol of general Fisher--Hartwig type (i.e., no conditions on the $\beta$-parameters).
The asymptotics of such Toeplitz determinants is given by the Basor--Tracy formula
(first conjectured by E.~Basor and C.~Tracy and then proved in~\cite{DIK}).
However, in the case of Hankel matrices, and this is the crux of the matter,
their characteristic polynomials are not Hankel determinants. Indeed, the characteristic polynomial $\det(H_n[w]-\la I)$ of the Hankel matrix $H_n[w]$ is a particular \textit{Toeplitz+Hankel} determinant, with $\phi(z)\equiv -\la$.
Clearly in the case of characteristic polynomial of a Hankel determinant, there is no relationship between $\phi$ and $w$, so to study the asymptotics of this determinant, one can not refer to the works \cite{DIK} or \cite{BE} mentioned above. Here again we are directed to a methodological issue which has to be addressed at a fundamental level by formulation of a suitable Riemann--Hilbert problem.

In this paper, we are proposing a version of the Riemann--Hilbert formalism for the asymptotic analysis of Toeplitz+Hankel determinants based on a certain $4 \times 4$ Riemann--Hilbert problem. When the Hankel symbol is supported on the unit circle, we introduce the following system of~monic orthogonal polynomials $\{\mathcal{P}_n(z)\}$, $\deg \mathcal{P}_n(z)=n$, associated to $D_n(\phi,w;r,s)$:
\begin{gather*}
\int_{\T} \mathcal{P}_n(z)z^{-k-r}\phi(z) \frac{{\rm d}z}{2\pi {\rm i} z} +\int_{\T} \mathcal{P}_n(z)z^{k+s} \tilde{w}(z) \frac{{\rm d}z}{2\pi {\rm i}z} = h_n \delta_{n,k}, \qquad
k=0,1,\dots,n.
\end{gather*}
We also show that for $r=s=1$, if the symbols are analytic in a neighborhood of the unit circle, one can proceed with a $4\times 4$ analogue of the Deift--Zhou non-linear steepest descent method and arrive at a $4\times 4$ \textit{model Riemann--Hilbert problem} on the unit circle which does not contain the parameter $n$. It is significant to note that one arrives at \textit{the same} model Riemann--Hilbert in the fundamentally different case where $w$ is supported on the interval $[a,b]$, with $0<a<b<1$. In this situation we consider the following system of monic orthogonal polynomials $\{P_n(z)\}$, $\deg P_n(z)=n$, associated to $D_n(\phi,w;r,s)$:
\begin{gather*}
\int_{\T} P_n(z)z^{-k-r}\phi(z) \frac{{\rm d}z}{2\pi {\rm i} z} + \int^{b}_{a} P_n(x)x^{k+s} w(x)\, {\rm d}x = h_n \delta_{n,k}, \qquad
k=0,1,\dots, n.
\end{gather*}
In this case, we can proceed with the Riemann--Hilbert analysis with $r=1$ and an \textit{arbitrary} value for $s\in \Z$.

We have been able to solve the model problem for the class of symbols (\ref{BEsymbols}) considered in~\cite{BE}, in the absence of Fisher--Hartwig singularities. It is important to discuss the relevancy of the two conditions assumed to be satisfied by the function $d$ in~\cite{BE}, in our Riemann--Hilbert framework (see \eqref{BEsymbols} and below). Unlike the condition $d\big({\rm e}^{{\rm i}\theta}\big)d\big({\rm e}^{-{\rm i}\theta}\big)=1$ which is, remarkably, a simplifying condition for the factorization of the model Riemann--Hilbert problem, it should be noticed that the condition $d(\pm1)=1$ is not required in the entirety of our Riemann--Hilbert approach. Solving the model problem allows us to find the asymptotics for the norm~$h_n$ of the associated orthogonal polynomials. We provide the details of this calculation for the case where~$w$ is supported on the unit circle.

In what follows we formulate our main asymptotic result. To do so we need to introduce some notations. Throughout the paper, we will occasionally refer to a symbol $f$ as a \textit{Szeg{\H o}-type} symbol, if a)~it is smooth and nonzero on the unit circle, b)~has no winding number, and c)~admits an analytic continuation in a neighborhood of the unit circle. Also for a given function~$f$, and an oriented contour $\Ga$, we write $f_{+}(z)$ (resp.~$f_-(z)$) to denote the limiting value of $f(\ze)$, as $\ze$ approaches $z \in \Ga$ from the left (resp.\ right) hand side of the oriented contour $\Ga$ with respect to its orientation. We also note that the expressions like $\tilde{f}_{\pm}(z)$, should be understood as the boundary values of the function $\tilde{f}$ at $z$ and should not be confused with the boundary value of~the function $f$ at the point $\frac{1}{z}$, in other words, the operation of $z \mapsto \frac{1}{z}$ precedes that of taking the boundary values. Given the Szeg{\H o}-type symbols $\phi(z)$ and $w(z) = d(z)\phi(z)$, we define
\begin{gather}
\al(z)=\exp \biggl[ \frac{1}{2 \pi {\rm i}} \int_{\T} \frac{\ln(\phi(\tau))}{\tau-z}\,{\rm d}\tau \biggr], \qquad
\be(z)=\exp \biggl[ \frac{1}{2 \pi {\rm i}} \int_{\T} \frac{\ln(d(\tau))}{\tau-z}\,{\rm d}\tau \biggr],\nonumber
\\[.5ex]
\label{C rho}
\mathcal{C}_{\rho}(z) = -\frac{1}{2\pi {\rm i}} \int_{\T} \frac{1}{\be_-(\tau) \be_+(\tau) \tilde{\al}_-(\tau) \al_+(\tau)(\tau-z)}\,{\rm d}\tau,\nonumber
\\[.5ex]
g_{23}(z) = - \frac{\al(0) \tilde{d}(z) \be(z) }{\tilde{\al}(z)}, \qquad
g_{43}(z) = - \al^2(0) \be(z) \biggl( \frac{\al(z) }{\tilde{\phi}(z)} + \frac{\tilde{d}(z) C_{\rho}(z)}{\tilde{\al}(z)} \biggr),\nonumber
\\
\label{T+H integrals in the norm}
R_{1,23}(z;n) = \frac{1}{2\pi {\rm i}}\int_{\Ga'_i} \frac{\mu^ng_{23}(\mu)}{\mu-z}\,{\rm d}\mu, \qquad
R_{1,43}(z;n) = \frac{1}{2\pi {\rm i}}\int_{\Ga'_i} \frac{\mu^ng_{43}(\mu)}{\mu-z}\,{\rm d}\mu,
\end{gather}
and finally
\begin{gather}\label{En0001}
\mathcal{E}(n)= \frac{2}{\al(0)}R_{1,43}(0;n)-C_{\rho}(0)R_{1,23}(0;n).
\end{gather}
In \eqref{T+H integrals in the norm}, the contour $\Ga'_i$ is a circle, oriented counter-clockwise, with radius $r'<1$ so that the functions $\phi$ and $d$ are analytic in the annulus $\{z\colon r'\leq|z|<1\}$.

\begin{Theorem}\label{T+H main thm}
Suppose that $\phi({\rm e}^{i\theta})$ is smooth and nonzero on the unit circle with zero winding number, which admits an analytic continuation in a neighborhood of the unit circle. Let $w=d\phi$, where $d$ satisfies all the properties of $\phi$ in addition to $d\big({\rm e}^{{\rm i}\theta}\big)d\big({\rm e}^{-{\rm i}\theta}\big)=1$, for all $\theta \in [0,2\pi)$. Let also
\begin{gather}\label{annulus}
U_0 := \left\{z\colon r_{i} < |z| < r_{o}\colon 0 < r_{i} < 1 < r_{o}\right\},
\end{gather}
be the neighborhood of the unit circle where both functions, $\phi(z)$ and $d(z)$ are analytic.
Denote
\begin{gather}\label{r0}
r_{0} := \max\big\{r_{i}, r^{-1}_{o}\big\},
\end{gather}
and suppose now that
there exists such $C>0$ that for sufficiently large $n$,
\begin{gather}\label{Enneq0}
|\mathcal{E}(n)| \geq Cr^{n}
\end{gather}
for some $r\in [r_{0}, 1)$,
where $ \mathcal{E}(n)$ is the functional of the weights $\phi$ and $w$ defined in \eqref{En0001}. Then, for sufficiently large~$n$
the determinant $D_{n}(\phi,w;1,1) \neq0$ and the asymptotics of
\begin{gather*}
h_{n-1} \equiv \frac{D_{n}(\phi,w;1,1)}{D_{n-1}(\phi,w;1,1)},
\end{gather*}
is given by
\begin{gather}\label{T+H main result}
h_{n-1}= -\al(0) \frac{\mathcal{E}(n)}{\mathcal{E}(n-1)}\big(1+ \bigO{\big({\rm e}^{-c_1n}\big)}\big), \qquad n \to \infty,
\end{gather}
where $c_1 = -\log\big(\frac{r^{2}_1}{r}\big) >0$, and $r_1$ is any number satisfying the conditions: $ r < r_1 < 1$ and $ r^2_1 < r$.
\end{Theorem}

\begin{Remark}\label{Remark slight generalization}
As a slight generalization of Theorem \ref{T+H main thm}, one could replace the condition~\eqref{Enneq0}~by
\begin{gather}\label{Enneq00}
|\mathcal{E}(n)| \geq Cr^{n} \varkappa(n)
\end{gather}
for some $r\in [r_{0}, 1)$,
where $\varkappa(n)>0$ for all sufficiently large $n$ and is such that ${\rm e}^{-\varepsilon n}=o(\varkappa(n))$ for all $\varepsilon>0$ as~$n \to \infty$. So as a result, \eqref{T+H main result} should be replaced by
\begin{gather*}
h_{n-1}= -\al(0) \frac{\mathcal{E}(n)}{\mathcal{E}(n-1)}
\biggl(1+ \bigO{\biggl(\frac{{\rm e}^{-c_1n}}{\varkappa(n)}\biggr)}\biggr), \qquad
n \to \infty,
\end{gather*}
where the constant $c_1$ above is the same as the one in \eqref{T+H main result}.
\end{Remark}

\begin{Remark} The analyticity of $\phi$ in the neighborhood of the unit circle is a technical condition. It can be lifted
and replaced by certain smoothness conditions using the approximation type arguments
similar to the ones used in~\cite[Section~6.2]{DIK}. However, because of the condition~(\ref{Enneq0}),
the corresponding analysis becomes more subtle.
We shall address this issue together with several other
technical points in the forthcoming publications.
\end{Remark}

\subsection{Outline} In Section~\ref{Section Hankel on T} we will analyze the case where Hankel symbol $w$ is supported on the unit circle. We~will propose a $2\times2$ Riemann--Hilbert problem {\it with a shift} for the associated orthogonal polynomials. For an effective Riemann--Hilbert analysis, we will then propose a $2\times4$ Riemann--Hilbert problem whose jump conditions could be written in the usual form of matrix multiplications. We will then formulate a $4\times4$ Riemann--Hilbert problem which is the suitable framework for our analysis. In the formulation of the $4 \times 4$ Riemann--Hilbert problem, for technical reasons, we will restrict to the particular offset values $r=s=1$. Following the natural steps of steepest descent analysis we will arrive at a model problem in Section~\ref{model}, which we will refer to as the \textit{model Riemann--Hilbert problem for the pair $(\phi,w)$}.

Section~\ref{Section Hankel on I} is devoted to analysis of the case where the Hankel symbol is supported on the interval $[a,b]$, with $0<a<b<1$. Similar to Section~\ref{Section Hankel on T}, we will propose a $2\times2$ Riemann--Hilbert problem
with a shift for the associated system of orthogonal polynomials and for the same reasons mentioned above we pass through a $2\times4$ to arrive at a suitable $4\times4$ Riemann--Hilbert problem. In this case, our methods allow for considering an arbitrary offset for the Hankel part, more precisely, we can pursue the steepest descent analysis for $r=1$ and an arbitrary $s\in \Z$. This steepest descent analysis leads us to a model problem, which is the same as the model problem of Section~\ref{Section Hankel on T}, except that it is for the pair $(\phi,-\tilde{u})$, where
\begin{gather*}
u(z)=z\int^b_{a}\frac{x^{s-1}w(x)}{x-z}\,{\rm d}x.
\end{gather*}
Although there are similarities between the steepest descent analysis of Section~\ref{Section Hankel on T} and Section~\ref{Section Hankel on I}, we feel obliged to lay out a thorough exposition to illustrate the remarkable fact that the same model RH problem emerges in both cases.

In Section~\ref{T+H solvable}, we present the factorization of the model RH problem for the pair $(\phi,d\phi)$, where the functions $\phi$ and $d$ are of Szeg{\H o}-type. The function $d$ further satisfies $d\tilde{d}=1$ on the unit circle. We will then use this solution to construct the solutions to the global parametrix and the small-norm Riemann--Hilbert problems, which finally enables us to find the asymptotics of~the norms $h_n$ of the associated monic orthogonal polynomials. Although we do not formulate our main asymptotic result for the symbol pairs discussed in Section~\ref{Section Hankel on I}, in Remark \ref{Remark result for line} we justify that the analysis of Section~\ref{T+H solvable} is completely relevant for the symbol pairs of Section~\ref{Section Hankel on I} as well.

Finally, in Section~\ref{Section: suggestions for future work} we summarize the still open technical and conceptual questions which we are going to address in our future work.


\section{Toeplitz+Hankel determinants: Hankel weight \\supported on $\T$ }\label{Section Hankel on T}

In this section we assume that $w$ is supported on the unit circle and that both symbols $\phi$ and~$w$,
admit analytic continuations to a neighborhood of the unit circle. A key observation is that the determinant (\ref{T+H Determinant}) is related to the system of monic polynomials $\{\mathcal{P}_n(z)\}$, $\deg \mathcal{P}_n(z)=n$, determined by the orthogonality relations \begin{gather}\label{T+H OP}
\int_{\T} \mathcal{P}_n(z)z^{-k-r}\phi(z) \frac{{\rm d}z}{2\pi {\rm i} z} +\int_{\T} \mathcal{P}_n(z)z^{k+s} \tilde{w}(z) \frac{{\rm d}z}{2\pi {\rm i}z} = h_n \delta_{n,k}, \qquad
k=0,1,\dots,n.
\end{gather}
These polynomials exist and are unique if the Toeplitz+Hankel determinants \eqref{T+H Determinant} are non-zero. This can be seen as follows. Expectedly, if $D_n \equiv D_n(\phi,w;r,s) \neq 0$, the polynomials $\mathcal{P}_{n}$ can be written as the following determinants
\begin{gather}\label{T+H OP Det rep}
\mathcal{P}_n(z) := \frac{1}{D_{n}} \det
\begin{pmatrix}
\phi_r+w_s & \phi_{r-1}+w_{s+1} & \cdots 
& \phi_{r-n}+w_{s+n}
\\
\phi_{r+1}+w_{s+1} & \phi_{r}+w_{s+2} & \cdots 
& \phi_{r-n+1}+w_{s+n+1}
\\
\vdots & \vdots & \ddots & \vdots
\\
\phi_{r+n-1}+w_{s+n-1} & \phi_{r+n-2}+w_{s+n} & \cdots & \phi_{r-1}+w_{s+2n-1} 
\\
1 & z & \cdots 
& z^n
\end{pmatrix}\!.
\end{gather}
Indeed, for the polynomials defined by (\ref{T+H OP Det rep}) we have that
\begin{gather*}
 \int_{\T} \mathcal{P}_n(z)z^{k+s}\tilde{w}(z)\frac{{\rm d}z}{2\pi {\rm i} z}
 \\ \qquad
 =\frac{1}{D_{n}} \det
\begin{pmatrix}\arraycolsep=0pt
\phi_r+w_s & \phi_{r-1}+w_{s+1} & \cdots 
& \phi_{r-n}+w_{s+n}
\\
\phi_{r+1}+w_{s+1} & \phi_{r}+w_{s+2} & \cdots 
& \phi_{r-n+1}+w_{s+n+1}
\\
\vdots & \vdots & \ddots 
& \vdots
\\
\phi_{r+n-1}+w_{s+n-1} & \phi_{r+n-2}+w_{s+n} & \cdots 
& \phi_{r-1}+w_{s+2n-1}
\\
w_{k+s} & w_{k+s+1} & \cdots 
& w_{k+s+n}
\end{pmatrix}\!,
\end{gather*}
and
\begin{gather*}
\int_{\T} \mathcal{P}_n(z)z^{-k-r}\phi(z) \frac{{\rm d}z}{2\pi {\rm i} z}
\\ \qquad
 =\frac{1}{D_{n}} \det \begin{pmatrix}
\phi_r+w_s & \phi_{r-1}+w_{s+1} & \cdots 
& \phi_{r-n}+w_{s+n}
\\
\phi_{r+1}+w_{s+1} & \phi_{r}+w_{s+2} & \cdots & \phi_{r-n+1}+w_{s+n+1} 
\\
\vdots & \vdots & \ddots 
& \vdots
\\
\phi_{r+n-1}+w_{s+n-1} & \phi_{r+n-2}+w_{s+n} & \cdots & \phi_{r-1}+w_{s+2n-1} 
\\
\phi_{k+r} & \phi_{k+r-1} & \cdots 
& \phi_{k+r-n}
\end{pmatrix}\!,
\end{gather*}
hence
\begin{gather*}
\int_{\T} \mathcal{P}_n(z)z^{-k-r}\phi(z) \frac{{\rm d}z}{2\pi {\rm i} z} \!+\!\int_{\T} \mathcal{P}_n(z)z^{k+s} \tilde{w}(z) \frac{{\rm d}z}{2\pi {\rm i}z}
\\
= \frac{1}{D_{n}} \det \begin{pmatrix}
\phi_r+w_s & \phi_{r-1}+w_{s+1} & \cdots
& \phi_{r-n}+w_{s+n}
\\
\phi_{r+1}+w_{s+1} & \phi_{r}+w_{s+2} & \cdots
& \phi_{r-n+1}+w_{s+n+1}
\\
\vdots & \vdots & \ddots & \vdots\\
\phi_{r+n-1}+w_{s+n-1} & \phi_{r+n-2}+w_{s+n} & \cdots 
& \phi_{r-1}+w_{s+2n-1}
\\
\phi_{k+r} + w_{k+s} & \phi_{k+r-1} + w_{k+s+1} & \cdots 
& \phi_{k+r-n} + w_{k+s+n}
\end{pmatrix}
\!= \frac{D_{n+1}}{D_{n}} \delta_{n,k}.
\end{gather*}
The uniqueness of the polynomial $\mathcal{P}_n(z) = z^n + a_{n-1}z^{n-1} + \cdots + a_0$ satisfying \eqref{T+H OP}, simply follows from the fact that one has the following linear system for the coefficients $a_j, 0 \leq j \leq n-1$:
\begin{gather}\label{uniquenessOfPolynomials}
(T_{n}[\phi;r]+H_{n}[w;s]) \begin{pmatrix}
a_0 \\
a_1 \\
\vdots \\
a_{n-1}
\end{pmatrix} = \begin{pmatrix}
-\phi_{-n+r} -w_{n+s} \\
-\phi_{1-n+r} -w_{1+n+s} \\
\vdots \\
-\phi_{-1+r} -w_{2n-1+s}
\end{pmatrix}\!.
\end{gather}
So if $D_{n} \neq 0$, the coefficients $a_j$ and hence $\mathcal{P}_{n}$, can be uniquely determined by inverting the Toeplitz+Hankel matrix in (\ref{uniquenessOfPolynomials}).
So the polynomials defined by (\ref{T+H OP Det rep}) are the unique polynomials satisfying \eqref{T+H OP}, and
\begin{gather}\label{T+H h_n and Det}
h_n=\frac{D_{n+1}(\phi,w;r,s)}{D_{n}(\phi,w;r,s)}.
\end{gather}

\noindent We consider the function $\mathcal{Y}$ defined as
\begin{gather}\label{OP Rep of Solution}
\mathcal{Y}(z;n)=\begin{pmatrix}
\mathcal{P}_n(z) & \di \int_{\T} \frac{\xi^{s} \tilde{w}(\xi) \mathcal{P}_n(\xi) + \xi^{r}\tilde{\phi}(\xi)\tilde{\mathcal{P}}_n(\xi)}{\xi-z}\frac{{\rm d}\xi }{2\pi {\rm i} \xi}\vspace{1ex}
\\
-\di \frac{1}{h_{n-1}}\mathcal{P}_{n-1}(z) & -\di \frac{1}{h_{n-1}} \int_{\T} \frac{\xi^{s} \tilde{w}(\xi) \mathcal{P}_{n-1}(\xi) + \xi^{r}\tilde{\phi}(\xi)\tilde{\mathcal{P}}_{n-1}(\xi)}{\xi-z}\frac{{\rm d}\xi }{2\pi {\rm i} \xi}
\end{pmatrix}\!,
\end{gather}
assuming that $D_n,D_{n-1}\neq 0$, which ensures that a) $h_{n-1}=D_n/D_{n-1}$ is well-defined and nonzero and b) that $\mathcal{P}_n$ and $\mathcal{P}_{n-1}$ exist and are unique. Now consider the following Riemann--Hilbert problem for finding the $2 \times 2$ matrix $\mathcal{Y}$ satisfying:

\begin{itemize}\itemsep=0pt
\item \textbf{RH-$\mathcal{Y}$1} \quad $\mathcal{Y}$ is holomorphic in $\C \setminus \T. $
\item \textbf{RH-$\mathcal{Y}$2} \quad For $z \in \T$ we have
\begin{gather}\label{022}
\mathcal{Y}_+^{(1)}(z;n)=\mathcal{Y}^{(1)}_-(z;n), \qquad z \in \T,
\end{gather}
and
\begin{gather}\label{0023}
\mathcal{Y}_+^{(2)}(z;n)=\mathcal{Y}^{(2)}_-(z;n) \!+\! z^{-1+s} \tilde{w}(z) \mathcal{Y}^{(1)}_-(z;n) \!+\! z^{-1+r}\tilde{\phi}(z)\mathcal{Y}^{(1)}_-\big(z^{-1};n\big), \quad
z \in \T,
\end{gather}
where $\T$ is positively oriented in the counter-clockwise direction.
\item \textbf{RH-$\mathcal{Y}$3} \quad As $z \to \infty$, $\mathcal{Y}$ satisfies
\begin{gather}\label{Yinfty}
\mathcal{Y}(z;n)=\big(I + \mathcal{O}\big(z^{-1}\big)\big) z^{n \sigma_3}
= \begin{pmatrix}
z^n+ \mathcal{O}\big(z^{n-1}\big) & \mathcal{O}\big(z^{-n-1}\big)
\vspace{1ex}\\
\mathcal{O}\big(z^{n-1}\big) & z^{-n} + \mathcal{O}\big(z^{-n-1}\big)
\end{pmatrix}\!,
\end{gather}

\end{itemize}
where $\mathcal{Y}^{(1)}$ and $\mathcal{Y}^{(2)}$ are the first and second columns of $\mathcal{Y}$ respectively and
\begin{gather*}
\sigma_3=\begin{pmatrix}
1 & 0 \\
0 & -1
\end{pmatrix}
\end{gather*}
is the third Pauli matrix. Notice that the jump conditions for the $\mathcal{Y}$-RHP can not be written in the usual form of matrix multiplications. Expectedly, this feature of the $\mathcal{Y}$-RHP (which will be addressed in more detail in Section~\ref{subsection 24 44 Hankel on T}) significantly affects the progression of the Riemann--Hilbert analysis as it necessitates studying larger size Riemann--Hilbert problems.

The next theorem establishes the association of the function \eqref{OP Rep of Solution} with the $\mathcal{Y}$-RHP.

\pagebreak

\begin{Theorem}\label{THM RHP formulation circle}
The following statements are true.
\begin{enumerate}\itemsep=0pt
\item[{\rm 1.}] Suppose that $D_n$, $D_{n-1} \neq 0$. Then, the Riemann--Hilbert problem \emph{\textbf{RH-$\mathcal{Y}$1}} through \emph{\textbf{RH-$\mathcal{Y}$3}} is uniquely solvable and its solution $\mathcal{Y}$ is defined by \eqref{OP Rep of Solution}.
Moreover,
\begin{gather}\label{T+H h_n00}
h_{n-1} = - \lim_{z \to \infty} z^{n-1}/\mathcal{Y}_{21}(z; n).
\end{gather}
\item[{\rm 2.}] Suppose that the Riemann--Hilbert problem \emph{\textbf{RH-$\mathcal{Y}$1}} through \emph{\textbf{RH-$\mathcal{Y}$3}}
has a unique solution. Then $D_n \neq 0 $, $\operatorname{rank} (T_{n-1}[\phi;r] + H_{n-1}[w;s]) \geq n-2$, and $\mathcal{P}_n(z)
= \mathcal{Y}_{11}(z;n)$.
\item[{\rm 3.}] Suppose that the Riemann--Hilbert problem \emph{\textbf{RH-$\mathcal{Y}$1}} through \emph{\textbf{RH-$\mathcal{Y}$3}}
has a unique solution. Suppose also that
\begin{gather*}
\lim_{z \to \infty}\mathcal{Y}_{21}(z;n)z^{-n+1} \neq 0.
\end{gather*}
Then, as before, $D_n \neq 0$, $\mathcal{P}_n(z)
= \mathcal{Y}_{11}(z;n)$, and, in addition,
\begin{gather*}
D_{n-1} \neq 0, \qquad
h_{n-1} = -\lim_{z\to \infty} \mathcal{Y}^{-1}_{21}(z;n) z^{n-1},\qquad
\mathcal{P}_{n-1}(z) = -h_{n-1}\mathcal{Y}_{21}(z;n).
\end{gather*}
\end{enumerate}
\end{Theorem}

\begin{proof}
Assume that $D_n,D_{n-1} \neq 0$, and thus the function $\mathcal{Y}$ is uniquely defined by \eqref{OP Rep of Solution}, being identified with the \textit{unique} orthogonal polynomials satisfying the orthogonality conditions \eqref{T+H OP}. We now show that $\mathcal{Y}$ given by \eqref{OP Rep of Solution} is a unique solution of \textbf{RH-$\mathcal{Y}$1} through \textbf{RH-$\mathcal{Y}$3}. It is clear that \textbf{RH-$\mathcal{Y}$1} is satisfied due to general properties of Cauchy integrals. From (\ref{022}) we see that~$\mathcal{Y}_{11}$ and~$\mathcal{Y}_{21}$ are entire functions, and from (\ref{Yinfty}) we know that $\mathcal{Y}_{11}$ has to be a monic polynomial of degree $n$ and $\mathcal{Y}_{21}$ has to be a polynomial of degree $n-1$ or less,
\begin{gather*}
\mathcal{Y}_{11}(z; n) = z^n + \sum_{k=0}^{n-1} \al_kz^{k}, \qquad
\mathcal{Y}_{21}(z;n) = \sum_{k=0}^{n-1} \be_kz^{k}.
\end{gather*}
From (\ref{0023}) and what we just mentioned about $\mathcal{Y}_{11}$ we would have
\begin{gather*}
\mathcal{Y}_{12,+}(z;n) - \mathcal{Y}_{12,-}(z;n) = z^{-1+s} \tilde{w}(z) \mathcal{Y}_{11}(z;n) + z^{-1+r}\tilde{\phi}(z)\widetilde{\mathcal{Y}}_{11}(z;n).
\end{gather*}
So by Plemelj--Sokhotskii formula we have
\begin{gather*}
\mathcal{Y}_{12}(z;n)=\frac{1}{2\pi {\rm i}} \int_{\T} \frac{\xi^{-1+s} \tilde{w}(\xi) \mathcal{Y}_{11}(\xi;n) + \xi^{-1+r}\tilde{\phi}(\xi)\widetilde{\mathcal{Y}}_{11}(\xi;n)}{\xi-z}
\,{\rm d}\xi.
\end{gather*}
Using the identity
\begin{gather*}
\frac{1}{\xi-z}=-\sum_{k=0}^{n} \frac{\xi^{k}}{z^{k+1}} + \frac{\xi^{n+1}}{(\xi-z)z^{n+1}},
\end{gather*}
we get
\begin{gather*}
\mathcal{Y}_{12}(z;n) = -\sum_{k=0}^{n} \frac{1}{z^{k+1}} \int_{\T} \big[ \xi^{-1+s} \tilde{w}(\xi) \mathcal{Y}_{11}(\xi;n) \xi^{k} + \xi^{-1+r}\tilde{\phi}(\xi)\widetilde{\mathcal{Y}}_{11}(\xi;n) \xi^{k} \big]\frac{{\rm d}\xi}{2\pi {\rm i}}
\\
\phantom{\mathcal{Y}_{12}(z;n)={}}+ \frac{1}{z^{n+1}} \int_{\T} \frac{\xi^{n+1}}{(\xi-z)} \big[ \xi^{-1+s} \tilde{w}(\xi) \mathcal{Y}_{11}(\xi;n) + \xi^{-1+r}\tilde{\phi}(\xi)\widetilde{\mathcal{Y}}_{11}(\xi;n)\big]\frac{{\rm d}\xi}{2\pi {\rm i}}.
\end{gather*}
Note that since $\mathcal{Y}_{12}(z;n)=O(z^{-n-1})$, we must have:
\begin{gather*}
\int_{\T} \tilde{w}(\xi) \mathcal{Y}_{11}(\xi;n) \xi^{k+s} \frac{{\rm d}\xi}{2\pi {\rm i} \xi} + \int_{\T} \tilde{\phi}(\xi)\widetilde{\mathcal{Y}}_{11}(\xi;n) \xi^{k+r} \frac{{\rm d}\xi}{2\pi {\rm i} \xi} = 0, \qquad 0\leq k \leq n-1.
\end{gather*}
In the second integral we make the change of variable $\xi \mapsto \tau := \xi^{-1}$ and as a result we will arrive at
\begin{gather*}
\int_{\T} \mathcal{Y}_{11}(\xi;n)\xi^{k+s} \tilde{w}(\xi) \frac{{\rm d}\xi}{2\pi {\rm i}\xi} + \int_{\T} \mathcal{Y}_{11}(\tau;n)\tau^{-k-r}\phi(\tau) \frac{{\rm d}\tau}{2\pi {\rm i} \tau} = 0 , \qquad
0 \leq k \leq n-1.
\end{gather*}
Since $\mathcal{Y}_{11}$ satisfies the orthogonality relations \eqref{T+H OP} we necessarily have \begin{gather*}
\mathcal{Y}_{11}(z;n) = \mathcal{P}_{n}(z).
\end{gather*}
In a similar fashion one can show that
\begin{gather*}
\mathcal{Y}_{22}(z;n)=\frac{1}{2\pi {\rm i}} \int_{\T} \frac{\xi^{-1+s} \tilde{w}(\xi) \mathcal{Y}_{21}(\xi;n) + \xi^{-1+r}\tilde{\phi}(\xi)\widetilde{\mathcal{Y}}_{21}(\xi;n)}{\xi-z}
{\rm d}\xi.
\end{gather*}
The asymptotic condition, $\mathcal{Y}_{22}(z;n) =z^{-n} + O\big(z^{-n-1}\big)$, would then yield the
orthogonality relations,
\begin{gather*}
\int_{\T} \mathcal{Y}_{21}(\xi;n)\xi^{k+s} \tilde{w}(\xi) \frac{{\rm d}\xi}{2\pi {\rm i}\xi} + \int_{\T} \mathcal{Y}_{21}(\tau;n)\tau^{-k-r}\phi(\tau) \frac{{\rm d}\tau}{2\pi {\rm i} \tau} = -\delta_{k,n-1},\qquad
0 \leq k \leq n-1.
\end{gather*}
These relations in turn are equivalent to the following linear system for the coefficients $\be_j$,
\begin{gather*}
(T_{n}[\phi;r]+H_{n}[w;s]) \begin{pmatrix}
\be_0 \\
\be_1 \\
\vdots \\
\be_{n-1}
\end{pmatrix} = \begin{pmatrix}
0 \\
0 \\
\vdots \\
-1
\end{pmatrix}\!.
\end{gather*}
From this it follows that $\be_{n-1}$ is necessarily not zero.
Indeed, using Cramer's rule, we arrive at~the formula,
\begin{gather*}
\be_{n-1} = -\frac{D_{n-1}}{D_n} \neq 0.
\end{gather*}
This would also mean that $\be_{n-1} = -\frac{1}{h_{n-1}}$ and
\begin{gather}\label{0001}
\mathcal{Y}_{21}(z;n) =- \frac{1}{h_{n-1}} \mathcal{P}_{n-1}(z).
\end{gather}
This finishes the argument why \eqref{OP Rep of Solution} uniquely satisfies both \textbf{RH-$\mathcal{Y}$2} and \textbf{RH-$\mathcal{Y}$3} as well. Also, equation (\ref{0001}) implies equation (\ref{T+H h_n00}).

To prove the second statement, now assume that the Riemann--Hilbert problem \textbf{RH-$\mathcal{Y}$1} through \textbf{RH-$\mathcal{Y}$3} has a unique solution. Because of similar considerations mentioned above, it~can be written as
\begin{gather}\label{OP Rep of Solution_00}
\mathcal{Y}(z;n)=\begin{pmatrix}
\mathcal{Q}_n(z) & \di \int_{\T} \frac{\xi^{s} \tilde{w}(\xi) \mathcal{Q}_n(\xi) + \xi^{r}\tilde{\phi}(\xi)\tilde{\mathcal{Q}}_n(\xi)}{\xi-z}\frac{{\rm d}\xi}{2\pi {\rm i} \xi}
\\[1ex]
\mathcal{R}_{n-1}(z) & \di \int_{\T} \frac{\xi^{s} \tilde{w}(\xi) \mathcal{R}_{n-1}(\xi) + \xi^{r}\tilde{\phi}(\xi)\tilde{\mathcal{R}}_{n-1}(\xi)}{\xi-z}\frac{{\rm d}\xi}{2\pi {\rm i} \xi}
\end{pmatrix}\!,
\end{gather}
where $\mathcal{Q}_n(z)$ and $\mathcal{R}_{m}(z)$ are polynomials,
\begin{gather*}
\mathcal{Q}_n(z) = z^n + \sum_{k=0}^{n-1} \al_kz^{k}, \qquad
\mathcal{R}_{n-1}(z) = \sum_{k=0}^{n-1} \be_kz^{k},
\end{gather*}
satisfying the orthogonality conditions,
\begin{gather}\label{ort1}
\int_{\T} \mathcal{Q}_n(z)z^{-k-r}\phi(z) \frac{{\rm d}z}{2\pi {\rm i} z} +\int_{\T} \mathcal{Q}_n(z)z^{k+s} \tilde{w}(z) \frac{{\rm d}z}{2\pi {\rm i}z} = 0, \qquad
k=0,1,\dots,n-1,
\end{gather}
and
\begin{gather}\label{ort2}
\int_{\T} \mathcal{R}_{n-1}(z)z^{-k-r}\phi(z) \frac{{\rm d}z}{2\pi {\rm i} z} +\int_{\T} \mathcal{R}_{n-1}(z)z^{k+s} \tilde{w}(z) \frac{{\rm d}z}{2\pi {\rm i}z}
= -\de_{k,n-1},
\\
k=0,1,\dots,n-1.\nonumber
\end{gather}
Note that the above equations can be written as the following linear systems on the coeffi\-cients~$\al_j$ and~$\be_j$:
\begin{gather}\label{linsys22}
(T_{n}[\phi;r]+H_{n}[w;s]) \begin{pmatrix}
\al_0 \\
\al_1 \\
\vdots \\
\al_{n-1}
\end{pmatrix} = \begin{pmatrix}
-\phi_{-n+r} -w_{n+s} \\
-\phi_{1-n+r} -w_{1+n+s} \\
\vdots \\
-\phi_{-1+r} -w_{2n-1+s}
\end{pmatrix}\!,
\end{gather}
and
\begin{gather}\label{linsys12}
(T_{n}[\phi;r]+H_{n}[w;s]) \begin{pmatrix}
\be_0 \\
\be_1 \\
\vdots \\
\be_{n-1}
\end{pmatrix} = \begin{pmatrix}
0 \\
0\\
\vdots \\
-1
\end{pmatrix}\!.
\end{gather}
Note that since we have assumed that the solution of $\mathcal{Y}$-RHP exists, both linear systems \eqref{linsys22} and \eqref{linsys12} must have solutions; and since the solution of $\mathcal{Y}$-RHP is unique, it follows that the determinant of the systems (\ref{linsys22}) and \eqref{linsys12}, which is nothing but $D_n$, is nonzero. Thus, we~certainly have
that
\begin{gather*}
\mathcal{Q}_n(z) = \mathcal{P}_n(z).
\end{gather*}
It is now obvious that $\mbox{rank}\, (T_{n-1}[\phi;r] + H_{n-1}[w;s]) \geq n-2$. Because if $D_{n-1} = 0$ and $D_n \neq 0$, one can easily show that $\mbox{rank}\, (T_{n-1}[\phi;r] + H_{n-1}[w;s]) = n-2$. Although the proof of the second statement is complete, it is worthwhile to say more about $\mathcal{Y}_{21}$ when $D_{n-1}=0$, because then the objects $\mathcal{P}_{n-1}(z)$ and $h_{n-1}$ do not exist. Indeed the homogenous system,
\begin{gather}\label{homsys}
(T_{n-1}[\phi;r]+H_{n-1}[w;s]) \begin{pmatrix}
c_0 \\
c_1 \\
\vdots \\
c_{n-2}
\end{pmatrix} = 0,
\end{gather}
has infinitely many nontrivial solutions. Let us take one of them
and define the following polynomial of degree $n-2$
\begin{gather*}
\hat{\mathcal{R}}_{n-2}(z) = \sum_{j=0}^{n-2} c_jz^{j},
\end{gather*}
which automatically satisfies
\begin{gather*}
\int_{\T} \hat{\mathcal{R}}_{n-2}(z)z^{-k-r}\phi(z) \frac{{\rm d}z}{2\pi {\rm i} z} +\int_{\T} \hat{\mathcal{R}}_{n-2}(z)z^{k+s} \tilde{w}(z) \frac{{\rm d}z}{2\pi {\rm i}z}
= 0, \qquad
k=0,1,\dots,n-2.
\end{gather*}
Put
\begin{gather*}
\delta_{n-1}:= \int_{\T} \hat{\mathcal{R}}_{n-2}(z)z^{-n+1-r}\phi(z) \frac{{\rm d}z}{2\pi {\rm i} z} +\int_{\T}\hat{ \mathcal{R}}_{n-2}(z)z^{n-1+s} \tilde{w}(z) \frac{{\rm d}z}{2\pi {\rm i}z}.
\end{gather*}
Note that $\delta_{n-1} \neq 0$. Otherwise, we could redefine $\mathcal{R}_{n-1}(z)$ as
\begin{gather*}
\mathcal{R}_{n-1}(z) \to \mathcal{R}_{n-1}(z) + \hat{\mathcal{R}}_{n-2}(z),
\end{gather*}
which, on the contrary, means that the solution of the $\mathcal{Y}$-RHP is not unique. Hence, we can find a \textit{unique} normalization of the nontrivial solution of \eqref{homsys}, to assure that
\begin{gather*}
\delta_{n-1 } = -1.
\end{gather*}
This, together with the uniqueness of the solution of the $\mathcal{Y}$-RHP implies that $\mathcal{R}_{n-1}(z)= \hat{\mathcal{R}}_{n-2}(z)$.

To prove the third statement, we notice that, due to the additional condition at $z=\infty$, the unique solution of the Riemann--Hilbert problem \textbf{RH-$\mathcal{Y}$1} through \textbf{RH-$\mathcal{Y}$3} can be written in the same form (\ref{OP Rep of Solution_00})
with polynomial $\mathcal{R}_{n-1}(z)$ replaced by
\begin{gather*}
\mathcal{R}_{n-1}(z) = q_n\mathcal{Q}_{n-1}(z), \qquad q
\neq 0,
\end{gather*}
where both $\mathcal{Q}_n(z)$ and $\mathcal{Q}_{n-1}(z)$ are monic polynomials,
\begin{gather*}
\mathcal{Q}_n(z) = z^n + \sum_{k=0}^{n-1} \al_kz^{k}, \qquad
\mathcal{Q}_{n-1}(z) = z^{n-1} + \sum_{k=0}^{n-2} \be_kz^{k},
\end{gather*}
satisfying the orthogonality conditions,
\begin{gather*}
\int_{\T} \mathcal{Q}_n(z)z^{-k-r}\phi(z) \frac{{\rm d}z}{2\pi {\rm i} z}
+\int_{\T} \mathcal{Q}_n(z)z^{k+s} \tilde{w}(z) \frac{{\rm d}z}{2\pi {\rm i}z} = 0, \qquad
k=0,1,\dots,n-1,
\end{gather*}
and
\begin{gather}\label{ort22}
q_n\Biggl(\int_{\T} \mathcal{Q}_{n-1}(z)z^{-k-r}\phi(z) \frac{dz}{2\pi {\rm i} z} +\int_{\T} \mathcal{Q}_{n-1}(z)z^{k+s} \tilde{w}(z) \frac{dz}{2\pi {\rm i}z}\Biggr)
= -\de_{k,n-1},
\\
k=0,1,\dots,n-1.\nonumber
\end{gather}
As before with equations (\ref{ort1}) and (\ref{ort2}), the above equations can be written as the following linear systems on the coefficients $\al_j$ and $\be_j$:
\begin{gather*}
(T_{n}[\phi;r]+H_{n}[w;s]) \begin{pmatrix}
\al_0 \\
\al_1 \\
\vdots \\
\al_{n-1}
\end{pmatrix} = \begin{pmatrix}
-\phi_{-n+r} -w_{n+s} \\
-\phi_{1-n+r} -w_{1+n+s} \\
\vdots \\
-\phi_{-1+r} -w_{2n-1+s}
\end{pmatrix}\!,
\end{gather*}
and
\begin{gather*}
(T_{n-1}[\phi;r]+H_{n-1}[w;s]) \begin{pmatrix}
\be_0 \\
\be_1 \\
\vdots \\
\be_{n-2}
\end{pmatrix} = \begin{pmatrix}
-\phi_{-n+1+r} -w_{n-1+s} \\
-\phi_{1-n+1+r} -w_{1+n-1+s} \\
\vdots \\
-\phi_{-1+r} -w_{2n-3+s}
\end{pmatrix}\!.
\end{gather*}

We now argue, that both $D_n$ and $D_{n-1}$ are nonzero and thus $\mathcal{Q}_n=\mathcal{P}_{n}$ and $\mathcal{Q}_{n-1}=\mathcal{P}_{n-1}$ are given by \eqref{T+H OP Det rep}. Otherwise, we shall have more than one solution for one or both of the
above linear systems, which means that we could find distinct polynomials $\mathcal{Q}_n(z)$ (if $D_n=0$), or~distinct polynomials $\mathcal{Q}_{n-1}(z)$ (if $D_{n-1}=0$), or both (if $D_n,D_{n-1}=0$). Either way, using these distinct polynomials, we could construct distinct solutions of the $\mathcal{Y}$-RHP which contradicts our assumption. This finishes the proof of the theorem. Note that from \eqref{ort22} we necessarily have $q_n \neq 0$, and since $\mathcal{Q}_{n-1}=\mathcal{P}_{n-1}$ we conclude that $h_{n-1} = -1/q_n \neq 0,$ which could also be~seen from $h_{n-1}=D_n/D_{n-1}$.
\end{proof}

\begin{Corollary}\label{Corollary 2.1.1}
Suppose that the $\mathcal{Y}$-RH problem has a unique solution for $n$ and $n-1$.
Then
\begin{gather*}
D_n \neq 0, \qquad D_{n-1} \neq 0, \qquad
\mbox{and} \qquad
h_{n-1} \neq 0,
\end{gather*}
where $h_{n-1}$ can be reconstructed form the RHP data as
\begin{gather}\label{T+H h_n}
h_{n-1} = - \lim_{z \to \infty} z^{n-1}/\mathcal{Y}_{21}(z;n).
\end{gather}
\end{Corollary}
Equation (\ref{T+H h_n}), as in the pure Toeplitz or Hankel cases, in conjunction with \eqref{T+H h_n and Det} reduces
the asymptotic analysis of the Toeplitz+Hankel determinants to the asymptotic analysis of the Riemann--Hilbert problem for $\mathcal{Y}$.


\subsection[The associated $2 x 4$ and $4 x 4$ Riemann--Hilbert problems]{The associated $\boldsymbol{2\times 4}$ and $\boldsymbol{4 \times 4}$ Riemann--Hilbert problems}\label{subsection 24 44 Hankel on T}

In the rest of this section we will develop a $4 \times 4$ analogue of the Deift/Zhou non-linear steepest descent method for the Toeplitz+Hankel determinants \eqref{T+H Determinant}. The $\mathcal{Y}$-RHP is a particular case of the matrix {\it Riemann--Hilbert problem with a shift}, or the matrix analytical boundary problem of the \textit{Carleman} type. Indeed, the matrix form of the equations (\ref{022})--(\ref{0023}) reads as follows
\begin{gather*}
\mathcal{Y}_{+}(z;n) = \mathcal{Y}_{-}(z;n)G_1(z) + \mathcal{Y}_{-}(\kappa(z);n)G_2(z),
\end{gather*}
where
\begin{gather*}
G_1(z) = \begin{pmatrix}
1 & z^{s-1}\tilde{w}(z) \\
0 & 1
\end{pmatrix}\!,\qquad
G_2(z) = \begin{pmatrix}
0 & z^{r-1}\tilde{\phi}(z) \\
0 & 0
\end{pmatrix}\!,
\end{gather*}
and the ``shift'' $\kappa$ is the mapping
\begin{gather*}
\kappa(z) = \frac{1}{z}.
\end{gather*}
The presence of the shift makes it impossible to directly apply the usual $2\times 2$ version of the Deift--Zhou nonlinear steepest
descent method to the $\mathcal{Y}$-RHP. However,
the mapping $\kappa$ satisfies the Carleman condition, $\kappa(\kappa(z)) = z$, and hence we can translate the
$2 \times 2$ $\mathcal{Y}$-RHP to the usual matrix form by doubling the relevant matrix sizes.
More precisely, we first propose the associated $2 \times 4$ and then the associated $4 \times 4$ Riemann--Hilbert problems. Although more complicated, the analysis of the proposed $4 \times 4$ Riemann--Hilbert problem follows in the same spirit as the lower dimensional RHPs until we get to the model Riemann--Hilbert problem for Toeplitz+Hankel determinants introduced in Section~\ref{model}.

Let us define the $2\times 4$ matrix $\overset{\circ}{\mathcal{X}}$ out of the columns of $\mathcal{Y}$ as follows
\begin{gather}\label{X naught to Y}
\overset{\circ}{\mathcal{X}}(z;n) :=
\bigl(\mathcal{Y}^{(1)}(z;n), \widetilde{\mathcal{Y}}^{(1)}(z;n), \mathcal{Y}^{(2)}(z;n), \widetilde{\mathcal{Y}}^{(2)}(z;n)\bigr).
\end{gather}

\pagebreak

From (\ref{022}), (\ref{0023}) and (\ref{Yinfty}) we obtain the following Riemann--Hilbert problem for $\overset{\circ}{\mathcal{X}}$:
\begin{itemize}\itemsep=0pt
\item \textbf{RH-}$\overset{\circ}{\mathcal{X}}$\textbf{1} \ \ \ $\overset{\circ}{\mathcal{X}}$ is holomorphic in $\C \setminus \left(\T \cup\{0\}\right)$.
\item \textbf{RH-}$\overset{\circ}{\mathcal{X}}$\textbf{2} \quad For $z \in \T$, $\overset{\circ}{\mathcal{X}}$ satisfies
\begin{gather}\label{JXo}
\overset{\circ}{\mathcal{X}}_+(z;n)=\overset{\circ}{\mathcal{X}}_-(z;n)
\begin{pmatrix}
1 & 0 & z^{s-1}\tilde{w}(z) & -z^{-r+1}\phi(z) \\
0 & 1 & z^{r-1}\tilde{\phi}(z) & -z^{-s+1}w(z) \\
0 & 0 & 1 & 0 \\
0 & 0 & 0 & 1
\end{pmatrix}\!.
\end{gather}
\item \textbf{RH-}$\overset{\circ}{\mathcal{X}}$\textbf{3} \quad As $z \to \infty$ we have \begin{gather}
\overset{\circ}{\mathcal{X}}(z;n) = \begin{pmatrix}
1 + \mathcal{O}\big(z^{-1}\big) & C_1(n)+\mathcal{O}\big(z^{-1}\big) & \mathcal{O}\big(z^{-1}\big) & C_3(n) + \mathcal{O}\big(z^{-1}\big) \\
\mathcal{O}\big(z^{-1}\big) & C_2(n)+\mathcal{O}\big(z^{-1}\big) & 1 + \mathcal{O}\big(z^{-1}\big) & C_4(n) + \mathcal{O}\big(z^{-1}\big)
\end{pmatrix}\nonumber
\\
\label{n07}
\phantom{\overset{\circ}{\mathcal{X}}(z;n)=}{}\times\begin{pmatrix}
z^n & 0 & 0 & 0\\
0 & 1 & 0 & 0 \\
0 & 0 & z^{-n} & 0 \\
0 & 0 & 0 & 1
\end{pmatrix}\!.
\end{gather}

\item \textbf{RH-}$\overset{\circ}{\mathcal{X}}$\textbf{4} \quad As $z \to 0$ we have
\begin{gather}
\overset{\circ}{\mathcal{X}}(z;n) = \begin{pmatrix}
C_1(n) + \mathcal{O}(z) & 1 + \mathcal{O}(z) & C_3(n)+\mathcal{O}(z) & \mathcal{O}(z) \\
C_2(n) + \mathcal{O}(z) & \mathcal{O}(z) & C_4(n)+\mathcal{O}(z) & 1 + \mathcal{O}(z)
\end{pmatrix}\nonumber
\\
\label{n088}
\phantom{\overset{\circ}{\mathcal{X}}(z;n)=}{}\times\begin{pmatrix}
1 & 0 & 0 & 0\\
0 & z^{-n} & 0 & 0 \\
0 & 0 & 1 & 0 \\
0 & 0 & 0 & z^n
\end{pmatrix}\!,
\end{gather}
\end{itemize}
where
\begin{gather*}
C_1(n)= \mathcal{Y}_{11}(0;n), \qquad
C_3(n)= \mathcal{Y}_{12}(0;n), \qquad
C_2(n)= \mathcal{Y}_{21}(0;n), \qquad
C_4(n)= \mathcal{Y}_{22}(0;n).
\end{gather*}

\begin{Remark}\label{Remark r,s}
It is well-known that for construction of the global parametrix in the Riemann--Hilbert analysis, one has to construct the so-called Szeg{\H o} functions of the following type:
\begin{gather*}
\mathcal{S}_{f}(z):= \exp\Bigg[\frac{1}{2\pi {\rm i}} \int_{\T} \frac{\ln(f(\tau))}{\tau-z} {\rm d} \tau \Bigg],
\end{gather*}
where one assumes that $f$ is non-zero and continuous on the unit circle with zero winding number. The natural progression of the Riemann--Hilbert analysis with general offset values $r,s\in \Z$, leads us to constructing Szeg{\H o} functions for $f_1(z)=z^{1-r}\phi(z)$ and $f_2(z)=z^{1-s}w(z)$ (or, equivalently, for $f_1$ and $f_2/f_1$). Although there are ways to ``peel off'' the winding generating parts of $f_1$ and $f_2$ (see Section~\ref{section-general-r,s}). The first natural case to consider is when $f_1$ and $f_2$ have zero winding numbers. To that end, in this work we shall only consider the case $r=s=1$ and the symbols $\phi$ and $w$ which are of Szeg{\H o} type and analytic in a neighborhood of the unit circle. We~also discuss prospects of extension of our method to symbols with Fisher--Hartwig singularities in~Section~\ref{T+H Extension to FH}. Also it is worth mentioning that our method in the present work trivially extends to the case where we consider general, but fixed, $r, s \in \Z$ with symbols $\phi$ and $w$ respectively with winding numbers $r-1$ and $s-1$ (so that the overall winding number of $f_1$ and $f_2$ are zero, which is the main condition to be considered). It should be also mentioned that in the pure Toeplitz case there is another, and historically the first, way to
put the analysis of Toeplitz determinants in the framework of the Riemann--Hilbert scheme.
This approach was suggested by Jinho Baik, Percy Deift, and Kurt Johansson in~\cite{BDJ}.
The Riemann--Hilbert problem that appear there yields to the Szeg\H{o} function
for $f(z)$ {\it coinciding} with the original Toeplitz symbol~$\phi(z)$. Therefore, the Riemann--Hilbert
method of \cite{BDJ} is applied directly to the symbols with zero winding numbers. The crux of the matter is
that the construction of \cite{BDJ} can not be extended to the mixed, T+H, situation. The $\mathcal{Y}$-RHP
we are introducing in this paper seems to be the only way to put the mixed problem into
the Riemann--Hilbert setting; hence the inevitability of the choice $r=s=1$, or equivalently, starting with symbols that have nonzero winding numbers.
\end{Remark}

In view of this remark, in the rest of the paper we assume that $r =s=1$. In a natural way we now consider the following $4 \times 4$ Riemann--Hilbert problem which we will refer to as the $\mathcal{X}$-RHP.

\begin{itemize}\itemsep=0pt

\item \textbf{RH-$\mathcal{X}$1} \quad $\mathcal{X}$ is holomorphic in $\C \setminus \left(\T \cup\{0\}\right)$.

\item \textbf{RH-$\mathcal{X}$2} \quad For $z \in \T$, $\mathcal{X}$ satisfies
\begin{gather}\label{X-jump}
\mathcal{X}_+(z;n)=\mathcal{X}_-(z;n) \begin{pmatrix}
1 & 0 & \tilde{w}(z) & -\phi(z) \\
0 & 1 & \tilde{\phi}(z) & -w(z) \\
0 & 0 & 1 & 0 \\
0 & 0 & 0 & 1
\end{pmatrix}\!.
\end{gather}

\item \textbf{RH-$\mathcal{X}$3} \quad As $z \to \infty$ we have
\begin{gather}\label{Xinfinity}
\mathcal{X}(z;n)=\big(\di I+\mathcal{O}\big(z^{-1}\big)\big)\begin{pmatrix}
z^n & 0 & 0 & 0\\
0 & 1 & 0 & 0 \\
0 & 0 & z^{-n} & 0 \\
0 & 0 & 0 & 1
\end{pmatrix}\!.
\end{gather}

\item \textbf{RH-$\mathcal{X}$4} \quad As $z \to 0$ we have
\begin{gather}\label{Xzero}
\mathcal{X}(z;n)=P(n)(I+\mathcal{O}(z))\begin{pmatrix}
1 & 0 & 0 & 0\\
0 & z^{-n} & 0 & 0 \\
0 & 0 & 1 & 0 \\
0 & 0 & 0 & z^n
\end{pmatrix}\!.
\end{gather}
\end{itemize}

\begin{Remark}\label{Remark X-solution is unique} The uniqueness of the solution of $\mathcal{X}$-RHP is established using the standard
Liouville theorem-based arguments. We also note that the matrix factor $P(n)$ in (\ref{Xzero}) is not a~priori prescribed.
\end{Remark}

\begin{Remark}
It is easy to see that the solution $\mathcal{X}$ of the $\mathcal{X}$-RHP satisfies the symmetry
relation,
\begin{gather}\label{sym101}
WP^{-1}(n)\mathcal{X}\big(z^{-1};n\big)W = \mathcal{X}(z;n),
\end{gather}
where
\begin{gather}\label{W}
W = \begin{pmatrix}
0 & 1 & 0 & 0\\
1 & 0 & 0 & 0 \\
0 & 0 & 0 & 1 \\
0 & 0 & 1 & 0
\end{pmatrix}\!.
\end{gather}
Equation (\ref{sym101}) in turn yields the following symmetry equation for $P(n)$,
\begin{gather}\label{sym102}
P(n) = WP^{-1}(n)W,
\end{gather}
or (taking into account that $W^{-1} = W$),
\begin{gather*}
\bigl(WP(n)\bigr)^2=\bigl(P(n)W\bigr)^2 = I_4.
\end{gather*}
\end{Remark}

\begin{Remark}\label{rank}
The matrix $P(n)W - I_4$ has rank 2. Here is the proof of this statement which is motivated by some of the referee's remarks.\footnote{These valuable remarks also led us to formulate and prove Lemmas \ref{Lemma Unique Solvability of C's} and \ref{Lemma unique solvability of C's}, and their analogues in Section~\ref{Section Hankel on I}.}

If we take the limit $z \to 1$, with $|z|>1$ in (\ref{sym101}) it follows that
\begin{gather*}
WP^{-1}(n) \mathcal{X}_+(1;n)W = \mathcal{X}_-(1;n),
\end{gather*}
or
\begin{gather*}
\mathcal{X}_+(1;n)W \mathcal{X}^{-1}_-(1;n) = P(n)W.
\end{gather*}
Using (\ref{X-jump}), the last equation reads,
\begin{gather*}
\mathcal{X}_-(1;n)G_{\mathcal{X}}(1)W \mathcal{X}^{-1}_-(1;n) = P(n)W,
\end{gather*}
where $G_{\mathcal{X}}(z)$ is the jump matrix in (\ref{X-jump}). Hence the matrices $P(n)W$ and
$G_{\mathcal{X}}(1)W$ are similar, and thus the claim is true if the matrix
$G_{\mathcal{X}}(1)W$ has rank $2$. We have,
\begin{gather*}
G_{\mathcal{X}}(1)W - I_4 =
\begin{pmatrix}
-1 & 1& -\phi & w\\
1 & -1 & -w & \phi\\
0 & 0 & -1 & 1 \\
0 & 0 & 1 & -1
\end{pmatrix}\!,
\end{gather*}
where $\phi \equiv \phi(1) = \tilde{\phi}(1)$, and $w\equiv w(1) = \tilde{w}(1)$. The row echelon form of
this matrix is
\begin{gather*}
\begin{pmatrix}
-1 & 1& -\phi & w\\
0 & 0 & -1& 1\\
0 & 0 & 0 & 0 \\
0 & 0 & 0 & 0
\end{pmatrix}\!,
\end{gather*}
and the statement about the rank of matrix $G_{\mathcal{X}}(1)W - I_4$ (and hence the rank of matrix $P(n)W - I_4$) follows.
\end{Remark}

\subsection[Relation of the $2 x 4$ and the $4 x 4$ Riemann--Hilbert problems]
{Relation of the $\boldsymbol{2\times4}$ and the $\boldsymbol{4\times4}$ Riemann--Hilbert problems}\label{Relation 24 44 Hankel T}

Put
\begin{gather}\label{24-to-44}
\mathfrak{R}(z;n) := \overset{\circ}{\mathcal{X}}(z;n) \mathcal{X}^{-1}(z;n).
\end{gather}
From (\ref{JXo}) and (\ref{X-jump}) it is clear that $\mathfrak{R}$ has no jumps. From (\ref{n088}) and (\ref{Xzero}) we can obtain the behavior of $\mathfrak{R}$ near zero:
\begin{gather*}
\mathfrak{R}(z;n) = \begin{pmatrix}
C_1(n) + \mathcal{O}(z) & 1 + \mathcal{O}(z) & C_3(n)+\mathcal{O}(z) & \mathcal{O}(z) \\
C_2(n) + \mathcal{O}(z) & \mathcal{O}(z) & C_4(n)+\mathcal{O}(z) & 1 + \mathcal{O}(z)
\end{pmatrix} P^{-1}(n).
\end{gather*}
Therefore $\mathfrak{R}$ is an entire function. Also note that from (\ref{n07}) and (\ref{Xinfinity}) we have
\begin{gather*}
\mathfrak{R}(z;n) = \begin{pmatrix}
1 + \mathcal{O}\big(z^{-1}\big) & C_1(n)+\mathcal{O}\big(z^{-1}\big) & \mathcal{O}\big(z^{-1}\big) & C_3(n) + \mathcal{O}\big(z^{-1}\big) \\
\mathcal{O}\big(z^{-1}\big) & C_2(n)+\mathcal{O}\big(z^{-1}\big) & 1 + \mathcal{O}\big(z^{-1}\big) & C_4(n) + \mathcal{O}\big(z^{-1}\big)
\end{pmatrix}\!, \qquad z\to \infty.
\end{gather*}
Therefore by Liouville's theorem we conclude that
\begin{gather}\label{mathfrakR}
\mathfrak{R}(z;n) = \begin{pmatrix}
1 & C_1(n) & 0 & C_3(n) \\
0 & C_2(n) & 1 & C_4(n)
\end{pmatrix}\!.
\end{gather}
And therefore we have
\begin{gather}\label{C's-to-P}
\begin{pmatrix}
1 & C_1(n) & 0 & C_3(n) \\
0 & C_2(n) & 1 & C_4(n)
\end{pmatrix} = \begin{pmatrix}
C_1(n) & 1 & C_3(n) & 0 \\
C_2(n) & 0 & C_4(n) & 1
\end{pmatrix} P^{-1}(n).
\end{gather}
We argue that this system, under certain generic assumptions, is a well-defined linear system on $C_j(n)$ which is uniquely solvable.
To see this, let us first double the system, that is consider instead of \eqref{C's-to-P}, the system
\begin{gather}\label{C's-to-P_2}
\begin{pmatrix}
1 & C_1(n) & 0 & C_3(n) \\
0 & C_2(n) & 1 & C_4(n)
\end{pmatrix} = \begin{pmatrix}
C'_1(n) & 1 & C'_3(n) & 0 \\
C'_2(n) & 0 & C'_4(n) & 1
\end{pmatrix} P^{-1}(n).
\end{gather}
This is an $8\times 8$ system of linear equations for eight unknowns -- $C_j(n)$ and $C'_j(n)$. By a~straightforward and,
in fact, rather simple calculations one finds that the determinant of this $8\times 8$ system~is
\begin{gather*}
\big(P_{22}(n)P_{44}(n) - P_{42}(n)P_{24}(n)\big)^2,
\end{gather*}
where $P_{jk}(n)$, $j, k = 1, ..., 4$ denote the entries of matrix $P(n)$.
Hence, assuming that
\begin{gather}\label{Csolvcond0}
P_{22}(n)P_{44}(n) - P_{42}(n)P_{24}(n) \neq 0,
\end{gather}
we would have the unique solvability of system (\ref{C's-to-P_2}). That is, there is only
one 8-vector,
\begin{gather}\label{vecC}
\vec{C}(n): =\big(C_1(n), C_2(n), C_{3}(n), C_4(n), C'_1(n), C'_2(n), C'_{3}(n), C'_4(n)\big),
\end{gather}
which solves (\ref{C's-to-P_2}). At the same time, the symmetry relation (\ref{sym102}) would imply that
\begin{gather*}
\begin{pmatrix}
1 & C_1(n) & 0 & C_3(n) \\
0 & C_2(n) & 1 & C_4(n)
\end{pmatrix}W = \begin{pmatrix}
C'_1(n) & 1 & C'_3(n) & 0 \\
C'_2(n) & 0 & C'_4(n) & 1
\end{pmatrix}W P(n),
\end{gather*}
or
\begin{gather*}
\begin{pmatrix}
1 & C'_1(n) & 0 & C'_3(n) \\
0 & C'_2(n) & 1 & C'_4(n)
\end{pmatrix} = \begin{pmatrix}
C_1(n) & 1 & C_3(n) & 0 \\
C_2(n) & 0 & C_4(n) & 1
\end{pmatrix} P^{-1}(n).
\end{gather*}
In other words, together with (\ref{vecC}), the system (\ref{C's-to-P_2}) will be also solved by the vector
\begin{gather*}
\vec{C}'(n): =\big(C'_1(n), C'_2(n), C'_{3}(n), C'_4(n), C_1(n), C_2(n), C_{3}(n), C_4(n)\big).
\end{gather*}
Because of the uniqueness, we conclude that the vectors $\vec{C}'(n)$ and $\vec{C}(n)$ must coincide and hence we must have,
\begin{gather*}
C_{j}(n) = C'_{j}(n), \qquad
j = 1, 2, 3, 4.
\end{gather*}
Therefore we have the unique solvability of the original system \eqref{C's-to-P} under the generic condition~(\ref{Csolvcond0}).

Indeed, we can make the solvability condition of the system \eqref{C's-to-P} more flexible. This is shown in the next lemma.
\begin{Lemma}\label{Lemma Unique Solvability of C's}
Assume that at least one of the following six inequalities is true,
\begin{gather}
P_{22}(n)P_{44}(n) - P_{42}(n)P_{24}(n) \neq 0,\label{Csolvcond01} \\
(1-P_{21}(n))P_{42}(n) + P_{22}(n)P_{41}(n) \neq 0,\label{gencond} \\
(1-P_{43}(n))P_{22}(n) + P_{23}(n)P_{42}(n) \neq 0,\label{gencond1}\\
(1-P_{21}(n))P_{44}(n) +P_{41}(n)P_{24}(n) \neq 0,\label{Csolvcond1}\\
(1-P_{21}(n))(P_{43}(n) - 1) +P_{41}(n)P_{23}(n) \neq 0,\label{Csolvcond2}\\
(1-P_{43}(n))P_{24}(n) + P_{23}(n)P_{44}(n) \neq 0.\label{gencond2}
\end{gather}
Then, the system \eqref{C's-to-P} is a well-defined linear system on $C_j(n)$ which is uniquely solvable.
\end{Lemma}
\begin{proof}
We first notice that system \eqref{C's-to-P} can be rewritten as
\begin{gather*}
\begin{pmatrix}
1 & C_1(n) & 0 & C_3(n) \\
0 & C_2(n) & 1 & C_4(n)
\end{pmatrix}P(n)W = \begin{pmatrix}
C_1(n) & 1 & C_3(n) & 0 \\
C_2(n) & 0 & C_4(n) & 1
\end{pmatrix}W =
\begin{pmatrix}
1 & C_1(n) & 0 & C_3(n) \\
0 & C_2(n) & 1 & C_4(n)
\end{pmatrix}\!,
\end{gather*}
or
\begin{gather}\label{C's-to-P-22}
\begin{pmatrix}
1 & C_1(n) & 0 & C_3(n) \\
0 & C_2(n) & 1 & C_4(n)
\end{pmatrix}(P(n)W - I_4) = 0.
\end{gather}
Notice that
\begin{gather*}
P(n)W - I_4 =
\begin{pmatrix}
P_{12}(n ) - 1 & P_{11}(n) & P_{14}(n)& P_{13}(n) \\
P_{22}(n ) & P_{21}(n) -1 & P_{24}(n)& P_{23}(n) \\
P_{32}(n ) & P_{31}(n) & P_{34}(n) - 1& P_{33}(n) \\
P_{42}(n ) & P_{41}(n) & P_{44}(n)& P_{43}(n)-1 \\
\end{pmatrix}\!.
\end{gather*}
Therefore, the conditions of this lemma mean that the second and the forth rows of $P(n)W - I_4$ are linearly independent. This, in conjunction with Remark \ref{rank},
implies that the other two rows are linear combinations of the second and the fourth rows. In other words, there exists a {\it unique} collection of four numbers, $C_1(n)$, $C_2(n)$, $C_3(n)$, $C_4(n)$
such that
\begin{gather}
\bigl(P_{12}(n ) - 1, P_{11}(n), P_{14}(n), P_{13}(n)\bigr) =
-C_1(n) \bigl(P_{22}(n ), P_{21}(n) -1,P_{24}(n), P_{23}(n)\bigr) \nonumber
\\ \label{PC1PC2}\qquad
{}-C_3(n) \bigl(P_{42}(n ), P_{41}(n), P_{44}(n), P_{43}(n) - 1\bigr),
\end{gather}
and
\begin{gather}
\bigl(P_{32}(n ), P_{31}(n), P_{34}(n) -1, P_{33}(n)\bigr) =
-C_2(n) \bigl(P_{22}(n ), P_{21}(n) -1,P_{24}(n), P_{23}(n)\bigr)\nonumber
\\ \label{PC3PC4}\qquad
{}-C_4(n) \bigl(P_{42}(n ), P_{41}(n), P_{44}(n), P_{43}(n) - 1\bigr).
\end{gather}
Vector equations (\ref{PC1PC2}), (\ref{PC3PC4}) are just the rows of matrix equation (\ref{C's-to-P-22}).
\end{proof}

Investigation of the possibility of linear independence of other rows of the matrix $P(n)W-I_4$, leads us to the following observation about the uniqueness of a solution to~\eqref{C's-to-P}, if one exists.
\begin{Lemma}\label{Lemma unique solvability of C's}
If the system \eqref{C's-to-P} has a solution, it has to be unique.
\end{Lemma}

\begin{proof}
For simplicity of notation let us drop the dependence on $n$ in $C_j(n)$ and $P(n)$, and respectively write $C_j$ and $P$ instead. Denote the rows of the matrix $PW - I_4$ by $\mathcal{R}_j$, $1 \leq j \leq 4$, and thus the system of equations \eqref{C's-to-P-22} can be written as the following vector equations
\begin{gather}\label{our system 1}
\mathcal{R}_1+C_1\mathcal{R}_2+C_3\mathcal{R}_4=0,
\end{gather}
and
\begin{gather}\label{our system 2}
\mathcal{R}_3+C_2\mathcal{R}_2+C_4\mathcal{R}_4=0.
\end{gather}
Since $PW - I_4$ is of rank $2$, necessarily, at least one of the six pairs of vectors has to be linearly independent. We consider each possibility separately.

First, let us assume that $\mathcal{R}_1$ and $\mathcal{R}_3$ are linearly independent. Note that $C_1C_4-C_2C_3 \neq 0$, because otherwise, multiplying \eqref{our system 1} by $-C_2$ and \eqref{our system 2} by $C_1$ and adding the results yields that $ \mathcal{R}_1$ and $\mathcal{R}_3$ are linearly dependent to the contrary. There exist unique constants $a_1$, $a_2$, $b_1$, and~$b_2$ such that
\begin{gather}\label{13 lin ind 24}
\mathcal{R}_2=a_1\mathcal{R}_1+a_2\mathcal{R}_3, \qquad
\mbox{and} \qquad \mathcal{R}_4=b_1\mathcal{R}_1+b_2\mathcal{R}_3.
\end{gather}
Also \looseness=-1 observe that, $a_1 b_2 - a_2b_1=0$ contradicts our assumption that $\mathcal{R}_1$ and $\mathcal{R}_3$ are linearly independent. Indeed, multiplying the first member and second member of \eqref{13 lin ind 24} respectively by $b_2$ and $-a_2$ and adding the results yields that $\mathcal{R}_2$ and $\mathcal{R}_4$ are linearly dependent, and consequently the equations~\eqref{our system 1} and \eqref{our system 2} imply that $\mathcal{R}_1$ and $\mathcal{R}_3$ are linearly dependent, a~contradiction.

Solving the system \eqref{our system 1}--\eqref{our system 2} for $\mathcal{R}_2$ and $\mathcal{R}_4$ and comparing with \eqref{13 lin ind 24} yields
\begin{alignat*}{3}
& a_1=-\frac{C_4}{C_1C_4-C_2C_3},\qquad &&
a_2=\frac{C_3}{C_1C_4-C_2C_3},&
\\
& b_1=\frac{C_2}{C_1C_4-C_2C_3}, \qquad &&
b_2=-\frac{C_1}{C_1C_4-C_2C_3}. &
\end{alignat*}
Using these equalities we find $a_1b_2-a_2b_1 = 1/(C_1C_4-C_2C_3)$, and since $a_1 b_2 - a_2b_1 \neq 0$, $C_1C_4-C_2C_3$ is \textit{uniquely} expressed in terms of $a_1 b_2 - a_2b_1$ and hence we can write the unique solution as
\begin{alignat*}{3}
& C_1= -\frac{b_2}{a_1 b_2 - a_2b_1}, \qquad&&
C_2 = \frac{b_1}{a_1 b_2 - a_2b_1},& \\
& C_3=\frac{a_2}{a_1 b_2 - a_2b_1}, \qquad&&
C_4= -\frac{a_1}{a_1 b_2 - a_2b_1}.&
\end{alignat*}

Next, assume that $\mathcal{R}_1$ and $\mathcal{R}_2$ are linearly independent and thus from \eqref{our system 1} we necessarily have
\begin{gather*}
C_3 \neq 0.
\end{gather*}
 So there exist unique constants $a_1$, $a_2$, $b_1$, and $b_2$ such that
\begin{gather}\label{12 lin ind 34}
\mathcal{R}_3=a_1\mathcal{R}_1+a_2\mathcal{R}_2, \qquad
\mbox{and} \qquad
\mathcal{R}_4=b_1\mathcal{R}_1+b_2\mathcal{R}_2.
\end{gather}
Note that $b_1=0$ contradicts our assumption that $\mathcal{R}_1$ and $\mathcal{R}_2$ are linearly independent, because then $\mathcal{R}_4=b_2 \mathcal{R}_2$ and substitution into \eqref{our system 1} yields $\mathcal{R}_1=-(C_1+b_2C_3)\mathcal{R}_2$, contrary to our assumption. Solving the system \eqref{our system 1}--\eqref{our system 2} for $\mathcal{R}_3$ and $\mathcal{R}_4$ and comparing with \eqref{12 lin ind 34} yields the equalities:
\begin{alignat*}{3}
& a_1=\frac{C_4}{C_3} , \qquad&&
a_2=\frac{C_4C_1-C_2C_3}{C_3}, &\\
& b_1=-\frac{1}{C_3} , \qquad&&
b_2=-\frac{C_1}{C_3} .&
\end{alignat*}
These relationships yield the unique solution for the system \eqref{our system 1}--\eqref{our system 2}, because $C_3$ is uniquely determined from $b_1$ (since $b_1 \neq 0$), then $C_1$ is uniquely determined from $C_3$ and $b_2$, and simultaneously $C_4$ is uniquely determined from $C_3$ and $a_1$, and finally $C_2$ is uniquely determined from $C_3$, $C_1$, $C_4$ and $a_2$. Thus, we can write the unique solution as
\begin{alignat*}{3}
& C_1= \frac{b_2}{b_1}, \qquad &&
C_2 = \frac{a_1b_2-b_1a_2}{b_1}, &\\
& C_3=-\frac{1}{b_1}, \qquad &&
C_4= -\frac{a_1}{b_1}.&
\end{alignat*}
Among the four remaining cases,the argument for each of the following three:

\begin{itemize}\itemsep=0pt
\item $\mathcal{R}_1$ and $\mathcal{R}_4$ are linearly independent,
\item $\mathcal{R}_2$ and $\mathcal{R}_3$ are linearly independent, or
\item $\mathcal{R}_3$ and $\mathcal{R}_4$ are linearly independent,
\end{itemize}
is similar to the one presented above for linear independence of $\mathcal{R}_1$ and $\mathcal{R}_2$ and thus we do not provide the details here. Finally, for the case where $\mathcal{R}_2$ and $\mathcal{R}_4$ are linearly independent, we refer to Lemma \ref{Lemma Unique Solvability of C's} where we prove the stronger assertion that the system \eqref{C's-to-P} is uniquely solvable.
\end{proof}

\begin{Lemma}\label{Lemma Unique reconstruction of Y from X}
Suppose that the solution of the $\mathcal{X}$-RHP exists. Then, if at least one of the con\-di\-tions \eqref{Csolvcond01} through \eqref{gencond2} holds, one can uniquely reconstruct the solution of the $\mathcal{Y}$-RHP.
\end{Lemma}

\begin{proof}
If the solution of the $\mathcal{X}$-RHP exists, then the expression for $P(n)$ can be found from \begin{gather}\label{T+H P from X}
P(n)= \mathcal{X}(z;n)\left.\begin{pmatrix}
1 & 0 & 0 & 0\\
0 & z^{n} & 0 & 0 \\
0 & 0 & 1 & 0 \\
0 & 0 & 0 & z^{-n}
\end{pmatrix}\right|_{z=0},
\end{gather}
and due to our assumption, the constants $C_j(n)$ can be uniquely found according to Lemma \ref{Lemma Unique Solvability of C's}. Then, according to \eqref{24-to-44} and \eqref{mathfrakR} we find the solution to the $\overset{\circ}{\mathcal{X}}$-RHP as
\begin{gather}\label{2 by 4 solution unique}
\overset{\circ}{\mathcal{X}}(z;n)= \begin{pmatrix}
1 & C_1(n) & 0 & C_3(n) \\
0 & C_2(n) & 1 & C_4(n)
\end{pmatrix}\mathcal{X}(z;n).
\end{gather}
Because, if it exists, the solution of the $\mathcal{X}$-RHP is unique (recall Remark~\ref{Remark X-solution is unique}), we note that~\eqref{2 by 4 solution unique} is the unique solution of the $\overset{\circ}{\mathcal{X}}$-RHP. Indeed,
suppose that $\overset{\circ}{\mathcal{X}}'$ is another solution of the same $\overset{\circ}{\mathcal{X}}$-RHP.
Put
\begin{gather*}
\mathfrak{R}'(z;n) := \overset{\circ}{\mathcal{X}'}(z;n) \mathcal{X}^{-1}(z;n).
\end{gather*}
Then, as with $\mathfrak{R}(z;n)$ before, we will arrive to the conclusion that
\begin{gather*}
\mathfrak{R}'(z;n) \equiv \begin{pmatrix}
1 & C'_1(n) & 0 & C'_3(n) \\
0 & C'_2(n) & 1 & C'_4(n)\end{pmatrix}\!,
\end{gather*}
where the constants $C'_j(n)$ satisfy the system \eqref{C's-to-P}. Since the solution of this system is unique according to Lemma \ref{Lemma unique solvability of C's}, we conclude that $ \mathfrak{R}'(z;n) = \mathfrak{R}(z;n)$ and, hence,
\begin{gather*}
\overset{\circ}{\mathcal{X}}'(z;n) = \overset{\circ}{\mathcal{X}}(z;n).
\end{gather*}
Now, observe that the symmetry relation (\ref{sym101}) implies
\begin{gather}\label{sym101_22}
\overset{\circ}{\mathcal{X}}\big(z^{-1};n\big) = \overset{\circ}{\mathcal{X}}(z;n)W.
\end{gather}
Indeed, we have
\begin{gather*}
\overset{\circ}{\mathcal{X}}\big(z^{-1} ;n\big) = \mathfrak{R}(n)\mathcal{X}(z^{-1};n) =
\mathfrak{R}(n)P(n)W\mathcal{X}(z;n)W.
\end{gather*}
But, by \eqref{C's-to-P-22}, $\mathfrak{R}(n)P(n)W=\mathfrak{R}(n)$, therefore
\begin{gather*}
\overset{\circ}{\mathcal{X}}\big(z^{-1} ;n\big) = \mathfrak{R}(n)\mathcal{X}(z;n)W
= \overset{\circ}{\mathcal{X}}(z;n)W.
\end{gather*}
Now, by $\overset{\circ}{\mathcal{X}^{(j)}}$, $j = 1, 2, 3, 4$ denote the columns of the $2\times4$ matrix $\overset{\circ}{\mathcal{X}}$.
Equation (\ref{sym101_22}) means that
\begin{gather*}
\overset{\circ}{\mathcal{X}}^{(2)}(z;n) = \overset{\circ}{\mathcal{X}}^{(1)}\big(z^{-1};n\big),\qquad
\mbox{and}\qquad
\overset{\circ}{\mathcal{X}}^{(4)}(z;n) = \overset{\circ}{\mathcal{X}}^{(3)}\big(z^{-1};n\big).
\end{gather*}
In other words, the matrix valued function $\overset{\circ}{\mathcal{X}}$ can be written in the form (\ref{X naught to Y}), i.e.,
\begin{gather*}
\overset{\circ}{\mathcal{X}}(z;n) =
\big(\mathcal{Y}^{(1)}(z;n), \widetilde{\mathcal{Y}}^{(1)}(z;n), \mathcal{Y}^{(2)}(z;n), \widetilde{\mathcal{Y}}^{(2)}(z;n)\big),
\end{gather*}
with
\begin{gather*}
\mathcal{Y}^{(1)}(z;n):= \overset{\circ}{\mathcal{X}}^{(1)}(z;n), \qquad
\mbox{and}\qquad
\mathcal{Y}^{(2)}(z;n):= \overset{\circ}{\mathcal{X}}^{(3)}(z;n).
\end{gather*}
Furthermore, the $2\times2$ matrix valued function
\begin{gather*}
\mathcal{Y}(z;n):= \begin{pmatrix}
\mathcal{Y}^{(1)}(z;n), \mathcal{Y}^{(2)}(z;n)\end{pmatrix}\!,
\end{gather*}
will be a solution of the $2\times2$ $\mathcal{Y}$-RHP. From the unique solvability of the $\overset{\circ}{\mathcal{X}}$-RHP, it follows that this solution of the $\mathcal{Y}$-RHP is unique. Because otherwise, if $\mathcal{Y}'$ is another solution of the $\mathcal{Y}$-RHP, via \eqref{X naught to Y} we can construct another solution of the $\overset{\circ}{\mathcal{X}}$-RHP on the contrary.
\end{proof}

The following corollary is the direct consequence of Corollary \ref{Corollary 2.1.1} and Lemma \ref{Lemma Unique reconstruction of Y from X}.
\begin{Corollary}\label{corollary2.8.1}
Suppose that the solution of the $\mathcal{X}$-RHP exists for $n$ and $n-1$, then if at least one of the conditions \eqref{Csolvcond01} through \eqref{gencond2} holds also for $n$ and $n-1$, then we have
\begin{gather*}
D_n \neq 0, \qquad
D_{n-1}\neq 0, \qquad
\mbox{and} \qquad
h_{n-1}\neq0.
\end{gather*}
Moreover,
\begin{gather}\label{T+H h_n000}
h_{n-1} = - \lim_{z \to \infty} z^{n-1}/\mathcal{Y}_{21}(z;n).
\end{gather}
\end{Corollary}

\subsection{The primary opening of the lenses}
Let us consider the contour $\Ga:=\Ga_i \cup \T \cup \Ga_o$ shown in Fig.~\ref{Figure Primary opening of the lenses}. Define the function $\mathcal{Z}$ as
\begin{gather*}
\mathcal{Z}(z;n) := \mathcal{X}(z;n)
\begin{cases}
J^{-1}_{\mathcal{X},i}(z), & z \in \Om_{1}, \vspace{1mm}\\
J_{\mathcal{X},o}(z), & z \in \Om_{2}, \vspace{1mm}\\
I, & z \in \Om_{0} \cup \Om_{\infty},
\end{cases}
\end{gather*}
where $J_{\mathcal{X},i}$ and $J_{\mathcal{X},o}$ are defined in the following factorization for the jump matrix of the $\mathcal{X}$-RHP, which we denote by $J_{\mathcal{X}}$:
\begin{gather}
J_{\mathcal{X}}(z) :=
\begin{pmatrix}
1 & 0 & \tilde{w}(z) & -\phi(z) \\
0 & 1 & \tilde{\phi}(z) & -w(z) \\
0 & 0 & 1 & 0 \\
0 & 0 & 0 & 1
\end{pmatrix}
= \begin{pmatrix}
1 & 0 & 0 & 0 \\
0 & 1 & 0 & -w(z) \\
0 & 0 & 1 & 0 \\
0 & 0 & 0 & 1
\end{pmatrix}
\begin{pmatrix}
1 & 0 & 0 & -\phi(z) \\
0 & 1 & \tilde{\phi}(z) & 0 \\
0 & 0 & 1 & 0 \\
0 & 0 & 0 & 1
\end{pmatrix} \nonumber
\\
\label{GXfactorization}\phantom{J_{\mathcal{X}}(z):=}\times
\begin{pmatrix}
1 & 0 & \tilde{w}(z) & 0 \\
0 & 1 & 0 & 0 \\
0 & 0 & 1 & 0 \\
0 & 0 & 0 & 1
\end{pmatrix}
\equiv J_{\mathcal{X},o}(z)J_{\mathcal{X},\T}(z)J_{\mathcal{X},i}(z).
\end{gather}
We remind that the symbol $w$, and hence $\tilde{w}$, are analytic in the neighborhood $U_0$ (cf.~\eqref{annulus}) of~$\T$
which is supposed to include the domains $\Omega_1$ and $\Omega_2$.
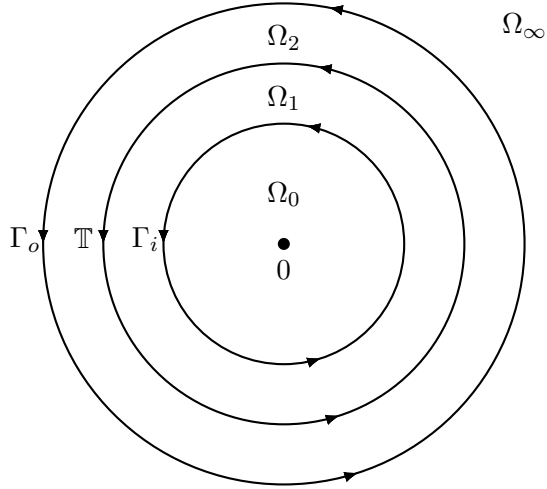
\begin{figure}[t]
\centering

\begin{tikzpicture}[scale=0.8]
\draw[ ->-=0.22,->-=0.5,->-=0.8,thick] (-3,0) circle (3cm);

\draw[ ->-=0.22,->-=0.5,->-=0.8,thick] (-3,0) circle (2cm);

\draw[ ->-=0.22,->-=0.5,->-=0.8,thick] (-3,0) circle (4cm);



\fill[black] (-3,0) circle (.1cm);

\node at (-3,-0.1) [below] {$0$};

\node at (-3,2.8) [below] {$\Om_1$};

\node at (-3,3.8) [below] {$\Om_2$};

\node at (-7.3,-0.3) [above] {$\Ga_{o}$};



\node at (-6,0.10) [left] {$\T$};

\node at (-5.3,0.45) [below] {$\Ga_{i}$};

\node at (-3,1.2) [below] {$\Om_0$};



\node at (1,4) [below] {$\Om_{\infty}$};
\end{tikzpicture}

\caption{The jump contour $\Ga$ for the $\mathcal{Z}$, $T$ and the global parametrix Riemann--Hilbert problems.}
\label{Figure Primary opening of the lenses}

\end{figure}

The function $\mathcal{Z}$ satisfies the following Riemann--Hilbert problem: \begin{itemize}\itemsep=0pt
\item \textbf{RH-$\mathcal{Z}$1} \quad $\mathcal{Z}$ is holomorphic in $\C \setminus \left(\Ga \cup \{0\} \right)$.

\item \textbf{RH-$\mathcal{Z}$2} \quad $\mathcal{Z}_+(z;n)=\mathcal{Z}_-(z;n)J_{\mathcal{Z}}(z)$, \ where \begin{gather*}
J_{\mathcal{Z}}(z)=\begin{cases}
J_{\mathcal{X},\T}(z), & z \in \T, \\
J_{\mathcal{X},i}(z), & z\in \Ga_{i}, \\
J_{\mathcal{X},o}(z), & z\in \Ga_{o}.
\end{cases}
\end{gather*}

\item \textbf{RH-$\mathcal{Z}$3} \quad As $z \to \infty$ we have
\begin{gather*}
\mathcal{Z}(z;n)=\big(\di I+\mathcal{O}\big(z^{-1}\big)\big) \begin{pmatrix}
z^n & 0 & 0 & 0 \\
0 & 1 & 0 & 0 \\
0 & 0 & z^{-n} & 0 \\
0 & 0 & 0 & 1
\end{pmatrix}\!.
\end{gather*}
\item \textbf{RH-$\mathcal{Z}$4} \quad As $z \to 0$ we have
\begin{gather*}
\mathcal{Z}(z;n)=P(n)(I+\mathcal{O}(z)) \begin{pmatrix}
1 & 0 & 0 & 0 \\
0 & z^{-n} & 0 & 0 \\
0 & 0 & 1 & 0 \\
0 & 0 & 0 & z^n
\end{pmatrix}\!.
\end{gather*}
\end{itemize}

\begin{Remark}
The term ``opening of the lenses'' is usually used to describe situations where the jump matrix on the added contours is exponentially close to the identity matrix for large values of the parameter $n$. The passage $\T \mapsto \Ga$, corresponding to the RH transformation $\mathcal{X}\mapsto \mathcal{Z}$, is~clearly not of this type. However, our secondary opening of the lenses (the passage $\Ga\mapsto \Ga_S$ which corresponds to the RH transformation $T \mapsto S$) in Section~\ref{Section: The secondary opening of the lenses} is an example of a usual opening of the lenses.
\end{Remark}

\begin{Remark}\label{Why primary openning of lenses?!}\looseness=1
The primary opening of the lenses is essential for the progression of the RH analysis in the following sections. This is due to a~technical reason that will be elaborated at~the end of next section. Since the structure of jump matrices is different in Section~\ref{Section Hankel on I}, we do not have an analogous step when the Hankel symbol is supported on $[a,b]$, \mbox{$0<a<b<1$}.
\end{Remark}

\subsection[Normalization of behaviours at 0 and infty]{Normalization of behaviours at $\mathbf{0}$ and $\boldsymbol{\infty}$}

Following the natural steps of Riemann--Hilbert analysis, we will normalize the behavior of $\mathcal{Z}$ at $0$ and $\infty$; to this end let us define
\begin{gather}\label{1}
T(z;n):=\mathcal{Z}(z;n)\begin{cases}
\begin{pmatrix}
z^{-n} & 0 & 0 & 0 \\
0 & 1 & 0 & 0 \\
0 & 0 & z^{n} & 0 \\
0 & 0 & 0 & 1
\end{pmatrix}\!, & |z|>1,\vspace{1mm}
\\
\begin{pmatrix}
1 & 0 & 0 & 0 \\
0 & z^{n} & 0 & 0 \\
0 & 0 & 1 & 0 \\
0 & 0 & 0 & z^{-n}
\end{pmatrix}\!, & |z|<1.
\end{cases}
\end{gather}
It is very important to note that in order to have a suitable Riemann--Hilbert analysis, the normalization of behaviors at $0$ and $\infty$ can only be carried out only after the undressing $\mathcal{X} \mapsto \mathcal{Z}$; this is due to technical reasons that will be further commented about at the end of this section. We have the following RHP for $T$:

\begin{itemize}\itemsep=0pt
\item \textbf{RH-$T$1} \quad $T$ is holomorphic in $\C \setminus \left( \T \cup \Ga_{i} \cup \Ga_{o} \right)$.

\item \textbf{RH-$T$2} \quad $T_+(z;n)=T_-(z;n)J_T(z;n)$, where
\begin{gather}\label{Jump matrices T Hankel on T}
J_T(z;n)= \begin{cases}
\widehat{J}(z;n), & z \in \T, \\
J_{\mathcal{X},i}(z), & z\in \Ga_{i}, \\
J_{\mathcal{X},o}(z), & z\in \Ga_{o},
\end{cases} \quad
\mbox{where} \quad
\widehat{J}(z;n) = \begin{pmatrix}
z^n & 0 & 0 & -\phi(z) \\
0 & z^n & \tilde{\phi}(z) & 0 \\
0 & 0 & z^{-n} & 0 \\
0 & 0 & 0 & z^{-n}
\end{pmatrix}\!,\!\!\!
\end{gather}
and the matrices $J_{\mathcal{X},i}$ and $J_{\mathcal{X},o}$ are defined by \eqref{GXfactorization}.
\item \textbf{RH-$T$3} \quad As $z \to \infty$, we have $T(z;n)=\big(\di I+\mathcal{O}\big(z^{-1}\big)\big)$.

\end{itemize}
We observe that for $z \in \T$, $J_T$ can be factorized as follows
\begin{gather}
\widehat{J}(z;n)= \begin{pmatrix}
I_2 & 0_2 \\
z^{-n} \Phi^{-1}(z) & I_2
\end{pmatrix} \begin{pmatrix}
0_2 & \Phi(z) \\
- \Phi^{-1}(z) & 0_2
\end{pmatrix} \begin{pmatrix}
I_2 & 0_2 \\
z^{n} \Phi^{-1}(z) & I_2
\end{pmatrix} \nonumber
\\
\label{GTfactorization}\phantom{\widehat{J}(z;n)}
\equiv J_{T,o}(z;n)\overset{\circ}{J}(z)J_{T,i}(z;n),
\end{gather}
where $0_2$ and $I_2$ are respectively $2 \times 2$ zero and identity matrices and
\begin{gather*}
\Phi(z)= \begin{pmatrix}
0 & -\phi(z) \\
\tilde{\phi}(z) & 0
\end{pmatrix}\!.
\end{gather*}
Note that $J_{T,i}$ is exponentially close to the identity matrix for $z$ inside of the unit circle and $J_{T,o}$ is exponentially close to the identity matrix for $z$ outside of the unit circle.

Now we are in a position to address Remark~\ref{Why primary openning of lenses?!} in the previous section. Indeed, if one normalizes the behaviors at $0$ and $\infty$ without the undressing transformation $\mathcal{X} \mapsto \mathcal{Z}$; i.e., by~directly defining the function $\mathcal{T}$ as \begin{gather*}
\mathcal{T}(z;n):= \mathcal{X}(z;n) \begin{cases}
\begin{pmatrix}
z^{-n} & 0 & 0 & 0 \\
0 & 1 & 0 & 0 \\
0 & 0 & z^{n} & 0 \\
0 & 0 & 0 & 1
\end{pmatrix}\!, & |z|>1, \vspace{1mm}
\\
\begin{pmatrix}
1 & 0 & 0 & 0 \\
0 & z^{n} & 0 & 0 \\
0 & 0 & 1 & 0 \\
0 & 0 & 0 & z^{-n}
\end{pmatrix}\!, & |z|<1.
\end{cases}
\end{gather*}
Then the jump matrix $J_{\mathcal{T}}:= \mathcal{T}^{-1}_-\mathcal{T}_+$ on the unit circle would be
\begin{gather*}
J_{\mathcal{T}}(z;n) = \begin{pmatrix}
z^n & 0 & z^{n}\tilde{w}(z) & -\phi(z) \\
0 & z^n & \tilde{\phi}(z) & -z^{-n}w(z) \\
0 & 0 & z^{-n} & 0 \\
0 & 0 & 0 & z^{-n}
\end{pmatrix}\!,
\end{gather*}
for which finding a factorization like (\ref{GTfactorization}) remains a challenge, mainly due to presence of the large parameter $n$ in the $13$ and $24$ elements of $J_{\mathcal{T}}$. This fact justifies the necessity of the undressing step $\mathcal{X} \mapsto \mathcal{Z}$. Indeed, due to the specific matrix structure of the jump matrices $J_{\mathcal{X},i}$ and $J_{\mathcal{X},o}$ they {\it do not} change
under the transformation~\eqref{1}.

\subsection{The secondary opening of the lenses}\label{Section: The secondary opening of the lenses}

The next Riemann--Hilbert transformation $T \mapsto S$, provides us with a problem with jump conditions on five contours where three jump matrices do not depend on $n$ and the other two converge exponentially fast to the identity matrix as $n \to \infty$. Let us define the function~$S$, suggested by~(\ref{GTfactorization}), as
\begin{gather*}
S(z;n):=T(z;n) \times \begin{cases}
J^{-1}_{T,i}(z;n), & z \in \Om'_1, \\
J_{T,o}(z;n), & z \in \Om'_2, \\
I, & z \in \Om''_1 \cup \Om''_2 \cup \Om_0 \cup \Om_{\infty},
\end{cases}
\end{gather*}
where the regions $\Om'_1$, $\Om'_2$, $\Om''_1$ and $\Om''_2$ are shown in Fig.~\ref{Fig2}. We have the following Riemann--Hilbert problem for $S$

\begin{itemize}\itemsep=0pt
\item \textbf{RH-$S$1} \quad $S$ is holomorphic in $\C \setminus \left( \T \cup \Ga_{i} \cup \Ga_{o} \cup \Ga'_{i} \cup \Ga'_{o} \right)$.

\item \textbf{RH-$S$2} \quad $S_+(z;n)=S_-(z;n)J_S(z;n)$, where
\begin{gather*}
J_S(z;n)=\begin{cases}
\overset{\circ}{J}(z), & z \in \T, \\
J_{T,i}(z;n), & z \in \Ga'_{i}, \\
J_{T,o}(z;n), & z \in \Ga'_{o}, \\
J_{\mathcal{X},i}(z), & z \in \Ga_{i}, \\
J_{\mathcal{X},o}(z), & z \in \Ga_{o}.
\end{cases}
\end{gather*}

\item \textbf{RH-$S$3} \quad As $z \to \infty$, we have $S(z;n)=\di I+\mathcal{O}\big(z^{-1}\big)$.
\end{itemize}

\begin{figure}
\centering

\begin{tikzpicture}[scale=1]

\draw[ ->-=0.22,->-=0.5,->-=0.8,thick] (-3,0) circle (3cm);

\draw[ ->-=0.22,->-=0.5,->-=0.8,dashed] (-3,0) circle (3.5cm);

\draw[ ->-=0.22,->-=0.5,->-=0.8,dashed] (-3,0) circle (2.5cm);

\draw[ ->-=0.22,->-=0.5,->-=0.8,thick] (-3,0) circle (2cm);

\draw[ ->-=0.22,->-=0.5,->-=0.8,thick] (-3,0) circle (4cm);



\fill[black] (-3,0) circle (.05cm);

\node at (-3,-0.1) [below] {$0$};

\node at (-3,3) [below] {$\Om'_1$};
\node at (-4,2.3) [below] {$\Om''_1$};

\node at (-4,3.87) [below] {$\Om''_2$};
\node at (-3,3.5) [below] {$\Om'_2$};

\node at (-6.9,0) [left] {$\Ga_{o}$};
\node at (-6.4,0) [left] {$\Ga'_{o}$};



\node at (-6.05,0) [right] {$\T$};

\node at (-5,0) [right] {$\Ga_{i}$};
\node at (-5.58,0) [right] {$\Ga'_{i}$};

\node at (-3,1.2) [below] {$\Om_0$};



\node at (1,4) [below] {$\Om_{\infty}$};
\end{tikzpicture}

\caption{The jump contour $\Ga_S$ of the $S$-RHP.}
\label{Fig2}
\end{figure}
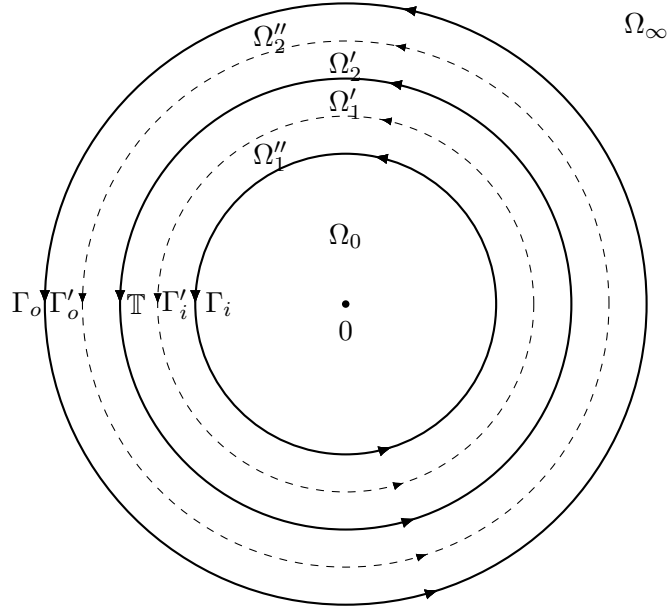

In the usual way, we will first try to solve this Riemann--Hilbert problem by disregarding the jump matrices which depend on $n$, this solution is denoted by $\overset{\circ}{S}$ and will be referred to as the global parametrix. Once we construct the global parametrix, we will consider the small-norm Riemann--Hilbert problem for the ratio $R:=S(\overset{\circ}{S})^{-1}$ and discuss its solvability in the forthcoming sections.

\subsection[The global parametrix and the model Riemann--Hilbert problem for the pair $(phi,w)$]{The global parametrix and the model Riemann--Hilbert problem \\for the pair $\boldsymbol{(\phi,w)}$}\label{model}

The $S$-RHP reduces to the following Riemann--Hilbert problem for the global parametrix $\overset{\circ}{S}$, when we ignore the jump matrices which are exponentially close to the identity matrix:

\begin{itemize}\itemsep=0pt
\item \textbf{RH-$\overset{\circ}{S}$1} \quad $\overset{\circ}{S}$ is holomorphic in $\C \setminus (\T \cup \Ga_{i} \cup \Ga_{o})$.

\item \textbf{RH-$\overset{\circ}{S}$2} \quad $\overset{\circ}{S}_+(z)=\overset{\circ}{S}_-(z)J_{\overset{\circ}{S}}(z)$, where \begin{gather*}
J_{\overset{\circ}{S}}(z)=\begin{cases}
\overset{\circ}{J}(z), & z \in \T, \\
J_{\mathcal{X},i}(z), & z \in \Ga_{i}, \\
J_{\mathcal{X},o}(z), & z \in \Ga_{o}. \\
\end{cases}
\end{gather*}

\item \textbf{RH-$\overset{\circ}{S}$3} \quad As $z\to \infty$, we have $\overset{\circ}{S}(z)=\di I+\mathcal{O}\big(z^{-1}\big)$.
\end{itemize}

\noindent
And we finally \textit{dress} the $\overset{\circ}{S}$-RHP to obtain a model problem for the global parametrix having jumps only on the unit circle. We define the function $\La$ as
\begin{gather}\label{9}
\La(z) := \overset{\circ}{S}(z) \times\begin{cases}
J_{\mathcal{X},i}(z), & z \in \Om_1, \\
J^{-1}_{\mathcal{X},o}(z), & z \in \Om_2, \\
I, & z \in \Om_{0} \cup \Om_{\infty}.
\end{cases}
\end{gather}
Now we arrive at the following Riemann--Hilbert problem for $\La$ that from now on we will refer to as \textit{the model Riemann--Hilbert problem for the pair $(\phi,w)$}:

\begin{itemize}\itemsep=0pt
\item \textbf{RH-$\La$1} \quad $\La$ is holomorphic in $\C \setminus \T$.

\item \textbf{RH-$\La$2} \quad $\La_+(z)=\La_-(z)J_{\La}(z)$, for $z \in \T$, where
\begin{gather*}
J_{\La}(z) = \begin{pmatrix}
0 & 0 & 0 & -\phi(z) \\
\di-\frac{w(z)}{\phi(z)} & 0 & \di \tilde{\phi}(z) - \frac{w(z)\tilde{w}(z)}{\phi(z)} & 0 \\
0 & \di -\frac{1}{\tilde{\phi}(z)} & 0 & 0 \\
\di\frac{1}{\phi(z)} & 0 & \di \frac{\tilde{w}(z)}{\phi(z)} & 0
\end{pmatrix}\!.
\end{gather*}

\item \textbf{RH-$\La$3} \quad As $z \to \infty$, we have $\La(z)=\di I+\mathcal{O}\big(z^{-1}\big)$.
\end{itemize}

\looseness=1 The conditions on $w$ and $\phi$ which ensure the solvability of this model problem are not completely known and categorized at this point. We also want to stress that the appearance of the $4\times4$
model $\Lambda$-problem in the asymptotic analysis of the original $\mathcal{X}$-RHP is the crucial
difference of the Toeplitz+Hankel case we consider in this work comparing to the pure Toeplitz or
pure Hankel or Toeplitz+Hankel with the same symbols cases. Indeed, even if the pair $(\phi, w)$ is such that
the $\Lambda$-RHP is solvable it does not mean that it is {\it explicitly} solvable. Hence, one should not
expect the closed form of the asymptotic answer in the case of the generic pair
$(\phi, w)$.\footnote{In this respect
the Toeplitz+Hankel determinants with generic symbols are similar to the block Toeplitz determinants, where
the explicit answers can be obtained only in two cases: (a) the Fourier expansion of the
corresponding matrix symbol is one side truncated or (b) one can produce an explicit Wiener--Hopf factorization of the
symbol.} However, in Section~\ref{T+H solvable} we will present a detailed analysis of this model problem for a specific family of~pairs $(\phi, w)$ within the broader class of Toeplitz and Hankel weights considered by~E.~Basor and T.~Ehrhardt in~\cite{BE} for which the model $\Lambda$-problem is explicitly solvable.

Remarkably, we arrive at the \textit{same} model $\Lambda$-Riemann--Hilbert problem, if we start with a~Hankel weight supported on the interval $[a,b]$, $0<a<b<1$. This will be shown in the next section.

\section[Toeplitz+Hankel determinants: Hankel weight supported on the interval $[a,b{]}$, $0<a<b<1$]{Toeplitz+Hankel determinants: Hankel weight supported\\ on~the interval $\boldsymbol{[a,b]}$, $\boldsymbol{0<a<b<1}$}\label{Section Hankel on I}

In this section we consider the determinant \eqref{T+H Determinant} where $w_k$ exist and are given by \eqref{moments of w}, with $I=[a,b]$, $0<a<b<1$. We further assume that $w$ does not have Fisher--Hartwig singularities (see \cite{Charlier,CharlierGharakhloo, ItsKrasovsky}, or \cite{Krasovsky} for instances of Fisher--Hartwig singularities on the real line). Let us again assume that the symbol $\phi$ is of Szeg{\H o}-type. The Riemann--Hilbert approach outlined in this section can be naturally extended to the three other cases: \textit{i})~$-1<a<b<0$, \textit{ii})~$-\infty<a<b<-1$, and \textit{iii})~$1<a<b<\infty$.
We consider the system of orthogonal polynomials $\{P_n(z)\}$, $\deg P_n(z)=n$, satisfying the following orthogonality conditions
\begin{gather}\label{T+H2 Orthogonality}
\int^{b}_{a} P_n(x)x^{k+s} w(x)\, {\rm d}x + \int_{\T} P_n(z)z^{-k-r}\phi(z) \frac{{\rm d}z}{2\pi {\rm i} z} = h_n \delta_{n,k}, \qquad
k=0,1,\dots, n.
\end{gather}
One can write a determinantal formula for $P_n$ like in \eqref{T+H OP Det rep}, which yields
\begin{gather*}
h_n=\frac{D_{n+1}(\phi,w;r,s)}{D_{n}(\phi,w;r,s)}.
\end{gather*}
By similar considerations as those mentioned in Section~\ref{Section Hankel on T}, the orthogonal polynomials $P_n$ exist and are unique if $D_n \neq 0$. Now, assuming that $D_n$, $D_{n-1}\neq0$, we consider the function $Y$ defined as
\begin{gather}\label{OP Rep of Solution interval}
Y(z;n)
\!=\!\begin{pmatrix}
P_n(z) & \di \int^{b}_{a} \frac{P_n(x)x^sw(x)}{x-z}{\rm d}x + \int_{\T} \frac{ \tilde{\phi}(\xi)\xi^{r}\tilde{P}_n(\xi)}{\xi-z} \frac{{\rm d}\xi}{2\pi {\rm i} \xi}
\\[2ex]
-\di \frac{P_{n-1}(z)}{h_{n-1}} & -\di \frac{1}{h_{n-1}} \Bigg\lbrace \int^{b}_{a} \frac{P_{n-1}(x)x^sw(x)}{x-z}{\rm d}x+ \!\! \int_{\T} \!\frac{ \tilde{\phi}(\xi)\xi^r\tilde{P}_{n-1}(\xi)}{\xi-z} \frac{{\rm d}\xi}{2\pi {\rm i} \xi} \Bigg\rbrace
\end{pmatrix}\!,
\end{gather}
built from the orthogonal polynomials $P_n$ satisfying \eqref{T+H2 Orthogonality}. Consider the following Riemann--Hilbert problem for finding the $2 \times 2$ matrix $Y$ satisfying
\begin{itemize}\itemsep=0pt
\item \textbf{RH-Y1} \quad $Y$ is holomorphic in $\C \setminus (\T \cup [a,b])$.
\item \textbf{RH-Y2} \quad For $z \in \T$ we have
\begin{gather*}
Y_+^{(1)}(z;n)=Y^{(1)}_-(z;n),
\end{gather*}
and
\begin{gather*}
Y_+^{(2)}(z;n)=Y^{(2)}_-(z;n) + z^{r-1}\tilde{\phi}(z)Y^{(1)}_-\big(z^{-1};n\big).
\end{gather*}

\item \textbf{RH-Y3} \quad For $x \in (a,b)$ we have
\begin{gather*}
Y_+^{(1)}(x;n)=Y^{(1)}_-(x;n),
\end{gather*}
and
\begin{gather*}
Y_+^{(2)}(x;n)=Y^{(2)}_-(x;n) + 2\pi {\rm i} x^s w(x)Y^{(1)}_-(x;n).
\end{gather*}

\item \textbf{RH-Y4} \quad As $z \to \infty$
\begin{gather*}
Y(z;n)=\bigg( I + O\bigg(\frac{1}{z}\bigg) \bigg) z^{n \sigma_3} = \begin{pmatrix}
z^n+\mathcal{O}\big(z^{n-1}\big) & \mathcal{O}\big(z^{-n-1}\big) \\
\mathcal{O}\big(z^{n-1}\big) & z^{-n} + \mathcal{O}\big(z^{-n-1}\big)
\end{pmatrix}\!,
\end{gather*}
\end{itemize}
where $Y^{(1)}$ and $Y^{(2)}$ are the first and second columns of $Y$, respectively. We have the analogue of Theorem \ref{THM RHP formulation circle} here as well.

\begin{Theorem}\label{THM RHP formulation interval}
The following statements are true.
\begin{enumerate}\itemsep=0pt
\item[$\rm 1.$] Suppose that $D_n,D_{n-1} \neq 0$. Then, the Riemann--Hilbert problem \emph{\textbf{RH-$Y$1}} through \emph{\textbf{RH-$Y$4}} is uniquely solvable and its solution $\mathcal{Y}$ is defined by \eqref{OP Rep of Solution interval}. Moreover,
\begin{gather*}
h_{n-1} = - \lim_{z \to \infty} z^{n-1}/Y_{21}(z; n).
\end{gather*}
\item[$\rm 2.$] Suppose that the Riemann--Hilbert problem \emph{\textbf{RH-$Y$1}} through \emph{\textbf{RH-$Y$4}}
has a unique solution. Then $D_n \neq 0 $, $\operatorname{rank} (T_{n-1}[\phi;r] + H_{n-1}[w;s]) \geq n-2$, and $P_n(z)
= Y_{11}(z;n)$.
\item[$\rm 3.$] Suppose that the Riemann--Hilbert problem \emph{\textbf{RH-$Y$1}} through \emph{\textbf{RH-$Y$4}}
has a unique solution. Suppose also that
\begin{gather*}
\lim_{z \to \infty}Y_{21}(z;n)z^{-n+1} \neq 0.
\end{gather*}
Then, as before, $D_n \neq 0$, $P_n(z)
= Y_{11}(z;n)$, and, in addition,
\begin{gather*}
D_{n-1} \neq 0, \qquad
h_{n-1} = -\lim_{z\to \infty} \mathcal{Y}^{-1}_{21}(z;n) z^{n-1},\qquad
P_{n-1}(z) = -h_{n-1}Y_{21}(z;n).
\end{gather*}
\end{enumerate}
\end{Theorem}
We omit the proof here as it is similar to the proof of Theorem \ref{THM RHP formulation circle}.

\begin{Corollary}\label{Corollary 3.1.1} Suppose that the $Y$-RH problem has a unique solution for $n$ and $n-1$.
Then
\begin{gather*}
D_n \neq 0, \qquad
D_{n-1} \neq 0, \qquad
\mbox{and} \qquad
h_{n-1} \neq 0.
\end{gather*}
Moreover,
\begin{gather*}
h_{n-1} = - \lim_{z \to \infty} z^{n-1}/Y_{21}(z;n).
\end{gather*}
\end{Corollary}

\subsection[The associated 2 x 4 and 4 x 4 Riemann--Hilbert problems]{The associated $\boldsymbol{2\times 4}$ and $\boldsymbol{4 \times 4}$ Riemann--Hilbert problems}

The formulation of the $2\times 4$ and $4\times 4$ Riemann--Hilbert problems are very similar to those of Section~\ref{subsection 24 44 Hankel on T}, however there are minor differences that convinces us to practice clarity in our exposition. Let us consider the following $2\times 4$ matrix function, constructed from the columns of $Y$ given by \eqref{OP Rep of Solution interval}:
\begin{figure}[t]
\centering
\begin{tikzpicture}[scale=0.5]
\draw[ ->-=0.22,->-=0.5,->-=0.8,thick] (-3,0) circle (4.5cm);
\draw [,->-=0.7,thick] (-2,0)
to (0,0);
\draw [->-=0.5,thick] (3,0)
to (9,0);

\fill[black] (-3,0) circle (.1cm);
\fill[black] (-2,0) circle (.1cm);
\fill[black] (0,0) circle (.1cm);
\fill[black] (3,0) circle (.1cm);
\fill[black] (9,0) circle (.1cm);

\node at (-3,-0.1) [below] {$0$};
\node at (-2,-0.1) [below] {$a$};
\node at (0,-0.1) [below] {$b$};
\node at (2.9,-0.1) [below] {$b^{-1}$};
\node at (9.1,-0.1) [below] {$a^{-1}$};
\node at (-7.5,0.05) [left] {$\T$};
\end{tikzpicture}
\caption{The jump contour $\Sigma$.}
\label{SIGMA CONTOUR}
\end{figure}
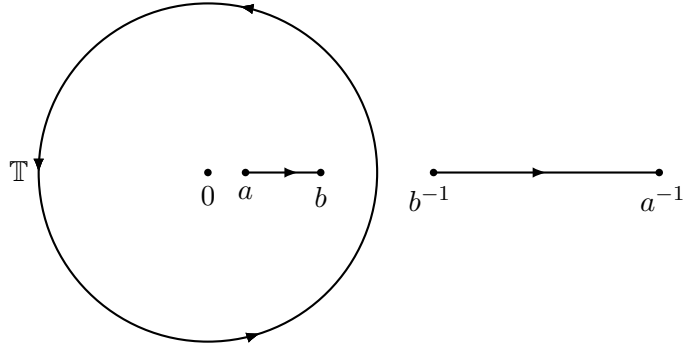
\begin{gather}\label{2by4 Hankel on I}
\overset{\circ}{X}(z;n) :=
\big(Y^{(1)}(z;n), \widetilde{Y}^{(1)}(z;n), Y^{(2)}(z;n), \widetilde{Y}^{(2)}(z;n)\big).
\end{gather}
Let us define $\Sigma := \T \cup [a,b] \cup \big[b^{-1},a^{-1}\big]$, and $\Sigma':=\Sigma \setminus\big\{a,b,b^{-1},a^{-1}\big\}$. $\overset{\circ}{X}(z;n)$ satisfies the following Riemann--Hilbert problem
\begin{itemize}\itemsep=0pt
\item \textbf{RH-$\overset{\circ}{X}$1} \qquad $\overset{\circ}{X}$ is analytic in
$\C \setminus \left( \Sigma \cup \{0\} \right)$,
\item \textbf{RH-$\overset{\circ}{X}$2} \qquad For $z\in \Sigma' $, we have $\overset{\circ}{X}_+(z;n)=\overset{\circ}{X}_-(z;n)J_{\overset{\circ}{X}}(z)$, where
\begin{gather*}
J_{\overset{\circ}{X}}(z) = \begin{cases}
\begin{pmatrix}
1 & 0 & 0 & -z^{-r+1}\phi(z) \\
0 & 1 & z^{r-1}\tilde{\phi}(z) & 0 \\
0 & 0 & 1 & 0 \\
0 & 0 & 0 & 1
\end{pmatrix}\!, & z \in \T,\vspace{1mm}
\\
\begin{pmatrix}
1 & 0 & 2\pi {\rm i}x^sw(x) & 0 \\
0 & 1 & 0 & 0 \\
0 & 0 & 1 & 0 \\
0 & 0 & 0 & 1
\end{pmatrix}\!, & z\equiv x \in (a,b),\vspace{1mm}
\\
\begin{pmatrix}
1 & 0 & 0 & 0 \\
0 & 1 & 0 & -2 \pi {\rm i} x^{-s} \tilde{w}(x) \\
0 & 0 & 1 & 0 \\
0 & 0 & 0 & 1
\end{pmatrix}\!, & z\equiv x \in \big(b^{-1},a^{-1}\big).
\end{cases}
\end{gather*}
\item \textbf{RH-$\overset{\circ}{X}$3} \qquad As $z\to \infty$
\begin{gather*}
\overset{\circ}{X}(z;n) =
\begin{pmatrix}
1 + O\big(z^{-1}\big) & E_1(n)+O\big(z^{-1}\big) & O\big(z^{-1}\big) & E_3(n) + O\big(z^{-1}\big) \vspace{1mm}\\
O\big(z^{-1}\big) & E_2(n)+O\big(z^{-1}\big) & 1 + O\big(z^{-1}\big) & E_4(n) + O\big(z^{-1}\big)
\end{pmatrix}
\\
\phantom{\overset{\circ}{X}(z;n) =}\times
\begin{pmatrix}
z^n & 0 & 0 & 0\\
0 & 1 & 0 & 0 \\
0 & 0 & z^{-n} & 0 \\
0 & 0 & 0 & 1
\end{pmatrix}\!.
\end{gather*}
\item \textbf{RH-$\overset{\circ}{X}$4} \qquad As $z\to 0$
\begin{gather*}
\overset{\circ}{X}(z;n) \!=\! \begin{pmatrix}
E_1(n) + O(z) & 1 + O(z) & E_3(n)+O(z) & O(z) \\
E_2(n) + O(z) & O(z) & E_4(n)+O(z) & 1 + O(z)
\end{pmatrix}\begin{pmatrix}
1 & 0 & 0 & 0\\
0 & z^{-n} & 0 & 0 \\
0 & 0 & 1 & 0 \\
0 & 0 & 0 & z^n
\end{pmatrix}\!,
\end{gather*}
\end{itemize}
where \begin{alignat*}{3}
& E_1(n)= Y_{11}(0;n), \qquad&&
E_3(n)= Y_{12}(0;n), &\\
& E_2(n)= Y_{21}(0;n), \qquad &&
E_4(n)= Y_{22}(0;n). &
\end{alignat*}
It is straightforward to check that $\overset{\circ}{X}$ given by \eqref{2by4 Hankel on I} and \eqref{OP Rep of Solution interval} satisfies the Riemann--Hilbert problem \textbf{RH-$\overset{\circ}{X}$1} through \textbf{RH-$\overset{\circ}{X}$4}.

\looseness=1 One of the differences between the case when the Hankel symbol is supported on the unit circle versus the case when it is supported on the interval $[a,b]$, is discussed in the following remark about the values of offsets that can be handled without much difficulty in each case.

\begin{Remark}\label{Remark r,s Hankel on the line}
Let $u$ be defined by
\begin{gather}\label{u definition Hankel on the line}
u(z):= z\int^b_{a} \frac{t^{s-1}w(t)}{t-z}{\rm d}t.
\end{gather}
When the Hankel symbol is supported on the interval, the natural progression of the Riemann--Hilbert analysis with general offset values $r, s \in \Z$, finally requires us to construct Szeg{\H o} functions for the functions $f_1(z)=z^{r-1}\phi(z)$ and $f_2(z)=\tilde{u}(z)$ (Compare with Remark~\ref{Remark r,s}). Note that the function $u$ given by~\eqref{u definition Hankel on the line} (and $\tilde{u}$) has no winding number for all $s \in \Z$. Therefore, in this work it seems natural for us to focus on the determinants of the type $D_n(\phi,w;1,s)$, in view of~the points mentioned above and in Remark~\ref{Remark r,s}.
\end{Remark}

Similar to our approach in Section~\ref{subsection 24 44 Hankel on T}, we introduce the following Riemann--Hilbert problem of finding the $4 \times 4$ matrix function $X$ satisfying:
\begin{itemize}\itemsep=0pt
\item \textbf{RH-X1} \qquad $X$ is analytic in
$\C \setminus (\Sigma \cup \{0\})$.

\item \textbf{RH-X2} \qquad For $z\in \Sigma'$, we have $X_+(z;n)=X_-(z;n)J_{X}(z)$, where
\begin{gather}\label{2by4jumps Hankel on I}
J_{X}(z) = \begin{cases}
\begin{pmatrix}
1 & 0 & 0 & -\phi(z) \\
0 & 1 & \tilde{\phi}(z) & 0 \\
0 & 0 & 1 & 0 \\
0 & 0 & 0 & 1
\end{pmatrix}\!, & z \in \T,\vspace{1mm}
\\
\begin{pmatrix}
1 & 0 & 2\pi {\rm i}x^sw(x) & 0 \\
0 & 1 & 0 & 0 \\
0 & 0 & 1 & 0 \\
0 & 0 & 0 & 1
\end{pmatrix}\!, & z\equiv x \in (a,b), \vspace{1mm}
\\
\begin{pmatrix}
1 & 0 & 0 & 0 \\
0 & 1 & 0 & -2 \pi {\rm i} x^{-s} \tilde{w}(x) \\
0 & 0 & 1 & 0 \\
0 & 0 & 0 & 1
\end{pmatrix}\!, & z\equiv x \in \big(b^{-1},a^{-1}\big),
\end{cases}
\end{gather}
that is, $J_X$ is exactly equal to $J_{\overset{\circ}{X}}$ when $r=1$.
\item \textbf{RH-X3} \qquad As $z\to \infty$
\begin{gather*}
X(z;n) = \big(I+\bigO{\big(z^{-1}\big)}\big)\begin{pmatrix}
z^n & 0 & 0 & 0\\
0 & 1 & 0 & 0 \\
0 & 0 & z^{-n} & 0 \\
0 & 0 & 0 & 1
\end{pmatrix}\!.
\end{gather*}

\item \textbf{RH-X4} \qquad As $z\to 0$ \begin{gather}\label{Xzero line}
X(z;n) =Q(n) (I+\bigO{(z)})\begin{pmatrix}
1 & 0 & 0 & 0\\
0 & z^{-n} & 0 & 0 \\
0 & 0 & 1 & 0 \\
0 & 0 & 0 & z^n
\end{pmatrix}\!,
\end{gather}
\end{itemize}
where we emphasize that the matrix factor $Q(n)$ in (\ref{Xzero line}) is not a priori prescribed. Now we are going to briefly mention some facts about this Riemann--Hilbert problem which are similar to those of the $\mathcal{X}$-RHP. Using the usual
Liouville theorem-based arguments one can easily show that the solution of $X$-RHP is unique, if it exists. Also, without much difficulty one can show that the function $WP^{-1}(n)X\big(z^{-1};n\big)W$ is also a solution of the $X$-RHP, and thus due to the uniqueness of the solution, we get the symmetry relation
\begin{gather}\label{sym101 line}
WP^{-1}(n)\mathcal{X}\big(z^{-1};n\big)W = \mathcal{X}(z;n),
\end{gather}
where $W$ is given by \eqref{W}. Equation (\ref{sym101 line}) yields the following symmetry equation for $Q(n)$,
\begin{gather*}
Q(n) = WQ^{-1}(n)W,
\end{gather*}
or,
\begin{gather*}
(WQ(n))^2=(Q(n)W)^2 = I_4.
\end{gather*}
Exact similar argument used in Remark \ref{rank} proves that the matrix $Q(n)W-I_4$ has rank $2$. Here we also have the relationship between $\overset{\circ}{X}$ and $X$ given by
\begin{gather*}
\overset{\circ}{X}(z;n)=\begin{pmatrix}
1 & E_1(n) & 0 & E_3(n) \\
0 & E_2(n) & 1 & E_4(n)
\end{pmatrix}X(z;n),
\end{gather*}
and moreover,
\begin{gather}\label{E's-to-P}
\begin{pmatrix}
1 & E_1(n) & 0 & E_3(n) \\
0 & E_2(n) & 1 & E_4(n)
\end{pmatrix} = \begin{pmatrix}
E_1(n) & 1 & E_3(n) & 0 \\
E_2(n) & 0 & E_4(n) & 1
\end{pmatrix} Q^{-1}(n).
\end{gather}
Since this is exactly the system \eqref{C's-to-P}, where $C_j(n)$ and $P(n)$ are respectively replaced by $E_j(n)$ and $Q(n)$, $1\leq j \leq 4$, and because $P(n)$ and $Q(n)$ enjoy the same symmetry and rank properties, we readily have the following statements, whose counterparts are already proven in Section~\ref{Section Hankel on T}.
\begin{Lemma}\label{Lemma Unique Solvability of C's line}
Let $Q_{jk}(n)$, $1 \leq j,k \leq 4$, be the entries of the matrix $Q(n)$. Assume that at least one of the following six inequalities is true,
\begin{gather}
Q_{22}(n)Q_{44}(n) - Q_{42}(n)Q_{24}(n) \neq 0,\label{Csolvcond01 line} \\
(1-Q_{21}(n))Q_{42}(n) + Q_{22}(n)Q_{41}(n) \neq 0,\label{gencond line} \\
(1-Q_{43}(n))Q_{22}(n) + Q_{23}(n)Q_{42}(n) \neq 0,\label{gencond1 line}\\
(1-Q_{21}(n))Q_{44}(n) +Q_{41}(n)Q_{24}(n) \neq 0,\label{Csolvcond1 line}\\
(1-Q_{21}(n))(Q_{43}(n) - 1) +Q_{41}(n)Q_{23}(n) \neq 0,\label{Csolvcond2 line}\\
(1-Q_{43}(n))Q_{24}(n) + Q_{23}(n)Q_{44}(n) \neq 0.\label{gencond2 line}
\end{gather}
Then, the system \eqref{E's-to-P} is a well-defined linear system on $E_j(n)$ which is uniquely solvable.
\end{Lemma}

\begin{Lemma}\label{Lemma unique solvability of C's line}
If the system \eqref{E's-to-P} has a solution, it has to be unique.
\end{Lemma}

\begin{Lemma}\label{Lemma Unique reconstruction of Y from X line}
Suppose that the solution of the $X$-RHP exists. Then, if at least one of the conditions \eqref{Csolvcond01 line} through \eqref{gencond2 line} holds, one can uniquely reconstruct the solution of the $Y$-RHP.
\end{Lemma}
\begin{Corollary}
Suppose that the solution of the $X$-RHP exists for $n$ and $n-1$, then if at least one of the conditions \eqref{Csolvcond01 line} through \eqref{gencond2 line} also holds for $n$ and $n-1$, then we have
\begin{gather*}
D_n \neq 0, \qquad D_{n-1}\neq 0, \qquad \mbox{and} \qquad h_{n-1}\neq0,
\end{gather*}
 where $h_{n-1}$ can be reconstructed form the RHP data as
\begin{gather*}
h_{n-1} = - \lim_{z \to \infty} z^{n-1}/Y_{21}(z;n).
\end{gather*}
\end{Corollary}

\subsection[Normalization of behaviors at 0 and infty]
{Normalization of behaviors at $\boldsymbol{0}$ and $\boldsymbol{\infty}$}

Unlike the situation in Section~\ref{Section Hankel on T} where we had to make the transformation $\mathcal{X} \mapsto \mathcal{Z}$ before normalization of behaviors at zero and infinity, when the Hankel symbol is supported on the interval $[a,b]$ we can immediately normalize the asymptotic behaviors at $0$ and infinity, due to the desired structure of jump matrices. Indeed, it is natural to define
\begin{gather}\label{T Hankel on I}
T(z;n):=X(z;n)\begin{cases}
\begin{pmatrix}
z^{-n} & 0 & 0 & 0 \\
0 & 1 & 0 & 0 \\
0 & 0 & z^{n} & 0 \\
0 & 0 & 0 & 1
\end{pmatrix}\!, & |z|>1, \vspace{1mm}
\\
\begin{pmatrix}
1 & 0 & 0 & 0 \\
0 & z^{n} & 0 & 0 \\
0 & 0 & 1 & 0 \\
0 & 0 & 0 & z^{-n}
\end{pmatrix}\!, & |z|<1.
\end{cases}
\end{gather}
The function $T$ satisfies the following Riemann--Hilbert problem:

\begin{itemize}\itemsep=0pt
\item \textbf{RH-$T$1} \quad $T$ is holomorphic in $\C \setminus \Sigma$.

\item \textbf{RH-$T$2} \quad For $z\in \Sigma'$, we have $T_+(z;n)=T_-(z;n)J_T(z;n)$, where
\begin{gather*}
J_T(z;n)= \begin{cases}
\widehat{J}(z;n), & z \in \T, \\
J_{X}(z), & z\in (a,b)\cup\big(b^{-1},a^{-1}\big).
\end{cases}
\end{gather*}
We recall that $\widehat{J}$ is given by \eqref{Jump matrices T Hankel on T} and the matrices $J_{X}$ for $z\in(a,b)$ and $z \in \big(b^{-1},a^{-1}\big)$ are given by \eqref{2by4jumps Hankel on I}.
\item \textbf{RH-$T$3} \quad As $z \to \infty$, we have $T(z;n)=\big(\di I+\mathcal{O}\big(z^{-1}\big)\big)$.
\end{itemize}
We bring the reader's attention to the fact that the transformation (\ref{T Hankel on I}) does not change the jump matrices $J_{X}$.
\subsection{Opening of the lenses}

Using \eqref{GTfactorization}, we open the lenses off the unit circle as shown in the Fig.~\ref{Opening lenses Hankel on I} and we define
\begin{gather*}
S(z;n):=T(z;n) \times \begin{cases}
J^{-1}_{T,i}(z;n), & z \in \Om_1, \\
J_{T,o}(z;n), & z \in \Om_2, \\
I, & z \in \C \setminus \left( \overline{\Om_1} \cup \overline{\Om_2} \cup [a,b] \cup \big[b^{-1},a^{-1}\big] \right),
\end{cases}
\end{gather*}
where $J_{T,i}$ and $J_{T,o}$ are defined in \eqref{GTfactorization}. Let $\Sigma_S \equiv \Sigma \cup \Sigma_o \cup \Sigma_i$ and $\Sigma'_S \equiv \Sigma' \cup \Sigma_o \cup \Sigma_i$ (see Fig.~\ref{Opening lenses Hankel on I}). It is straightforward to check that $S$ satisfies the following Riemann--Hilbert problem

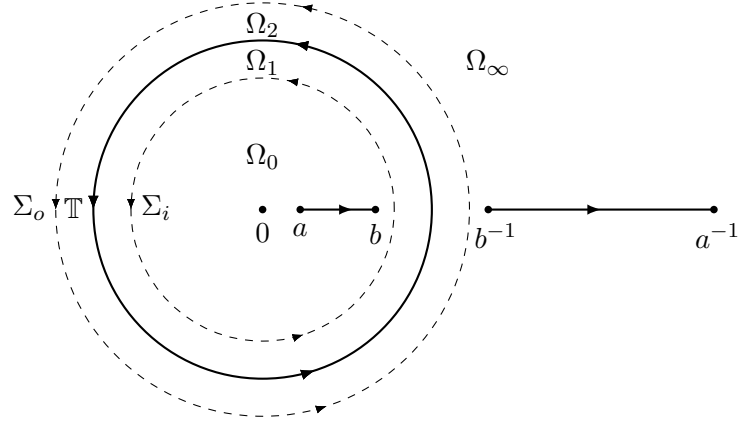
\begin{figure}[t]
\centering
\begin{tikzpicture}[scale=0.5]

\draw[ ->-=0.22,->-=0.5,->-=0.8,thick] (-3,0) circle (4.5cm);
\draw[ ->-=0.22,->-=0.5,->-=0.8,dashed] (-3,0) circle (3.5cm);
\draw[ ->-=0.22,->-=0.5,->-=0.8,dashed] (-3,0) circle (5.5cm);
\draw [,->-=0.7,thick] (-2,0)
to (0,0);
\draw [->-=0.5,thick] (3,0)
to (9,0);

\fill[black] (-3,0) circle (.1cm);
\fill[black] (-2,0) circle (.1cm);
\fill[black] (0,0) circle (.1cm);
\fill[black] (3,0) circle (.1cm);
\fill[black] (9,0) circle (.1cm);

\node at (-3,-0.1) [below] {$0$};
\node at (-2,-0.1) [below] {$a$};
\node at (0,-0.1) [below] {$b$};
\node at (3.2,-0.1) [below] {$b^{-1}$};
\node at (9.1,-0.1) [below] {$a^{-1}$};
\node at (-7.5,0.05) [left] {$\T$};
\node at (-8.5,0.05) [left] {$\Sigma_o$};
\node at (-6.5,0.05) [right] {$\Sigma_i$};
\node at (-3,2) [below] {$\Om_0$};
\node at (-3,4.5) [below] {$\Om_1$};
\node at (-3,5.5) [below] {$\Om_2$};
\node at (3,4.5) [below] {$\Om_{\infty}$};

\end{tikzpicture}
\caption{The jump contour $\Sigma_S \equiv \Sigma \cup \Sigma_o \cup \Sigma_i$ of the $S$-RHP.}
\label{Opening lenses Hankel on I}
\end{figure}

\begin{itemize}\itemsep=0pt
\item \textbf{RH-$S$1} \quad $S$ is holomorphic in $\C \setminus \Sigma_S$,

\item \textbf{RH-$S$2} \quad For $z \in \Sigma'_S$ we have $S_+(z;n)=S_-(z;n)J_S(z;n)$, where \begin{gather*}
J_S(z;n)=\begin{cases}
\overset{\circ}{J}(z), & z \in \T, \\
J_{T,i}(z;n), & z \in \Sigma_{i}, \\
J_{T,o}(z;n), & z \in \Sigma_{o}, \\
J_{X}(z), & z \in (a,b) \cup \big(b^{-1},a^{-1}\big), \\
\end{cases}
\end{gather*}where these matrices are defined in \eqref{GTfactorization} and \eqref{2by4jumps Hankel on I}.

\item \textbf{RH-$S$3} \quad As $z \to \infty$, we have $S(z;n)=\di I+\mathcal{O}\big(z^{-1}\big)$.
\end{itemize}

\subsection{The global parametrix and a model Riemann--Hilbert problem}\label{model2}

Let us consider the following Riemann--Hilbert for $\overset{\circ}{S}$, or the global parametrix, which is expected to be a good approximation to $S$ for large parameter~$n$. This RHP is simply obtained from the $S$-RHP by ignoring the jumps on $\Sigma_i$ and $\Sigma_o$:

\begin{itemize}\itemsep=0pt
\item \textbf{RH-}$\overset{\circ}{S}$\textbf{1} \quad $\overset{\circ}{S}$ is holomorphic in $\C \setminus \Sigma$ (see Fig.~\ref{SIGMA CONTOUR}).

\item \textbf{RH-$\overset{\circ}{S}$2} \quad For $z\in \Sigma'$, we have
$\overset{\circ}{S}_+(z)=\overset{\circ}{S}_-(z)J_{\overset{\circ}{S}}(z)$, where
\begin{gather*}
J_{\overset{\circ}{S}}(z)=\begin{cases}
\overset{\circ}{J}(z), & z \in \T, \\
J_{X}(z), & z \in (a,b) \cup \big(b^{-1},a^{-1}\big). \\
\end{cases}
\end{gather*}

\item \textbf{RH-$\overset{\circ}{S}$3} \quad As $z \to \infty$, we have $\overset{\circ}{S}(z)=\di I+\mathcal{O}\big(z^{-1}\big)$.
\end{itemize}
Let us recall the function $u$ defined in Remark \ref{Remark r,s Hankel on the line}:
\begin{gather*}
u(z)= z\int^b_{a} \frac{t^{s-1}w(t)}{t-z}\,{\rm d}t.
\end{gather*}
The Plemelj--Sokhotskii formula implies that
\begin{gather*}
u_+(x)-u_-(x)=2\pi {\rm i} x^sw(x), \qquad x \in (a,b),
\\
\tilde{u}_+(x)-\tilde{u}_-(x)= -2\pi {\rm i} x^{-s}\tilde{w}(x), \qquad x \in \big(b^{-1},a^{-1}\big).
\end{gather*}
Put
\begin{gather*}
\Theta(z):= \overset{\circ}{S}(z) \begin{cases}
\begin{pmatrix}
1 & 0 & -u(z) & 0 \\
0 & 1 & 0 & 0 \\
0 & 0 & 1 & 0 \\
0 & 0 & 0 & 1
\end{pmatrix}\!, & |z|<1, \vspace{1mm}
\\
\begin{pmatrix}
1 & 0 & 0 & 0 \\
0 & 1 & 0 & -\tilde{u}(z) \\
0 & 0 & 1 & 0 \\
0 & 0 & 0 & 1
\end{pmatrix}\!, & |z|>1.
\end{cases}
\end{gather*}
It can be checked that $\Theta$ does not have jumps on the intervals $(a,b)$ and $\big(b^{-1},a^{-1}\big)$. We have arrived at the following \textit{model Riemann--Hilbert problem} on the unit circle:

\begin{itemize}\itemsep=0pt
\item \textbf{RH-$\Theta$1} \quad $\Theta$ is holomorphic in $\C \setminus \T$.

\item \textbf{RH-$\Theta$2} \quad $\Theta_+(z)=\Theta_-(z)J_{\Theta}(z)$, for $z \in \T$, where
\begin{gather*}
J_{\Theta}(z) = \begin{pmatrix}
0 & 0 & 0 & -\phi(z) \\
\di\frac{\tilde{u}(z)}{\phi(z)} & 0 & \di \tilde{\phi}(z) - \frac{u(z)\tilde{u}(z)}{\phi(z)} & 0 \\
0 & \di -\frac{1}{\tilde{\phi}(z)} & 0 & 0 \\
\di\frac{1}{\phi(z)} & 0 & \di -\frac{u(z)}{\phi(z)} & 0
\end{pmatrix}\!.
\end{gather*}

\item \textbf{RH-$\Theta$3} \quad As $z \to \infty$, we have $\Theta(z)=\di I+\mathcal{O}\big(z^{-1}\big)$,
\end{itemize}
where, in \textbf{RH-$\Theta$3} we have used the fact that $\tilde{u}(\infty)=u(0)=0$.

\begin{Remark}\label{Remark Hankel on Line}
Recalling Section~\ref{model}, we note that this is exactly the model Riemann--Hilbert problem for the pair $(\phi,-\tilde{u})$. Hence, it can be concluded that the study of the Toeplitz+Hankel determinants both when the Hankel symbol is supported on the unit circle and also when it is supported on the interval $[a,b]$, reduces to the study of the model Riemann--Hilbert problem \textbf{RH-$\La$1} through \textbf{RH-$\La$3}.
\end{Remark} For a specific class of symbols $\phi$ and $w$, we will present the solution to the model Riemann--Hilbert problem \textbf{RH-$\La$1} through \textbf{RH-$\La$3} for the pair $(\phi,w)$ in the next section.

\section{Analysis of the model problem and a solvable pair}\label{T+H solvable}

As mentioned before, it is an ambitious task to classify all the pairs $(\phi,w)$ for which the $\La$-model Riemann--Hilbert problem is solvable. However, it is reasonable to start our analysis with the class of symbols \eqref{BEsymbols} considered in~\cite{BE}. Since in our work the symbols are not assumed to be of the Fisher--Hartwig type (which needs a more delicate treatment, see Section~\ref{T+H Extension to FH}), we should still expect that the model Riemann--Hilbert problem be solvable for the class of symbols \eqref{BEsymbols} when there is no Fisher--Hartwig singularity($a_0(z)=b_0(z) \equiv 1$). Indeed this is the case as will be elaborated in this section. As commented in the beginning of Section~\ref{subsection 24 44 Hankel on T}, asymptotics of $D_n(\phi,d\phi;r,s)$, for general $r$ and $s$ requires a more delicate approach (see Section~\ref{section-general-r,s}) and we do not discuss the details here. So let us consider $D_n(\phi,d\phi;1,1)$, where $d$ is of Szeg{\H o}-type and further satisfies the condition $d(z)\tilde{d}(z)=1$ on the unit circle. For instance, a class of functions satisfying these conditions is given by
\begin{gather}\label{T+H Def d}
d(z)= \prod^{m}_{j=1} d_j(z), \qquad
d_j(z)=\pm\bigg(\frac{z-b_j}{z-a_j}\bigg)^{\al_j}\bigg(\frac{a_{j}z-1}{b_{j}z-1}\bigg)^{\al_j},
\end{gather}
where $\al_j\in \C$, all factors are defined by their principal branch, and
\begin{gather*}
0<a_1<b_1<a_2<b_2< \cdots < a_m< b_m < 1.
\end{gather*}
Note that a similar construction can be found for $-1< b_m < a_m < \cdots < b_1< a_1 < 0$, and thus a larger class of functions can be found from multiplying functions of the first class with those of the second class. Although we have a class of functions satisfying the required properties, a complete categorization of functions satisfying the four required properties for $d$ is yet to be found. We emphasize that the conditions $d(\pm 1)=1$ required in~\cite{BE} do not play a role in the Riemann--Hilbert analysis. Indeed, for $d$ as defined in \eqref{T+H Def d} one can check that $d(\pm 1)=(-1)^{\epsilon_0}$,
where $\epsilon_0$ is the number of the $d_j$-factors in whose definition the sign ``$-$'' is taken. So in this sense we are considering functions $d$ which are slightly more general than those considered in~\cite{BE}. At~the same time, we have our technical assumption of analyticity
of the symbols in a~neighborhood of the unit circle which is not needed in the analysis of \cite{BE}.

Note that the condition $d\tilde{d}=1$ on the unit circle renders the 23-element of the jump mat\-rix~$J_{\La}$ zero; indeed
\begin{gather*}
J_{\La,23}(z) = \di \tilde{\phi}(z) - \frac{w(z)\tilde{w}(z)}{\phi(z)} = \tilde{\phi}(z) \big( 1 - d(z) \tilde{d}(z)\big)=0.
\end{gather*}
Hence, for the particular choices made above, the jump matrix $G_{\La}$ reduces to
\begin{gather}\label{T+H model jump reduced}
J_{\La}(z) = \begin{pmatrix}
0 & 0 & 0 & -\phi(z) \\
\di- d(z) & 0 & 0 & 0 \\
0 & \di -\frac{1}{\tilde{\phi}(z)} & 0 & 0 \\
\di\frac{1}{\phi(z)} & 0 & \di \frac{\tilde{w}(z)}{\phi(z)} & 0
\end{pmatrix}\!. \end{gather}
In order to factorize $J_{\La}$, let us first consider the following Szeg{\H o} functions
\begin{gather}\label{alpha}
\al(z)=\exp \bigg[ \frac{1}{2 \pi {\rm i} } \int_{\T} \frac{\ln(\phi(\tau))}{\tau-z}\,{\rm d}\tau \bigg], \qquad
\be(z)=\exp \bigg[ \frac{1}{2 \pi {\rm i} } \int_{\T} \frac{\ln(d(\tau))}{\tau-z}\,{\rm d}\tau \bigg].
\end{gather}
By Plemelj--Sokhotskii formula $\al$, $\be$, $\tilde{\al}$ and $\tilde{\be}$ satisfy the following jump conditions on the unit circle:
\begin{gather}\label{T+H jumps of Szego functions}
\begin{split}
& \al_+(z)=\al_-(z) \phi(z), \qquad \be_+(z)=\be_-(z) d(z), \\ & \tilde{\al}_-(z)=\tilde{\al}_+(z) \tilde{\phi}(z), \qquad \tilde{\be}_-(z)=\tilde{\be}_+(z) \tilde{d}(z).
\end{split}
\end{gather}
It turns out that knowing the value of $\be(0)$ is crucial for finding an asymptotic expression for~$h_n$ (see Section~\ref{Section T+H Asymp h_n}) and the condition $d\tilde{d}=1$ on the unit circle allows us to evaluate $\be(0)$ easily. Indeed
\begin{gather*}
\int_{\T}\ln(d(\tau)) \,\frac{{\rm d}\tau}{\tau} = \int_{\T}\ln\big(\tilde{d}(\tau)\big) \,\frac{{\rm d}\tau}{\tau}= \int_{\T}\ln\big(d^{-1}(\tau)\big) \,\frac{{\rm d}\tau}{\tau}= - \int_{\T}\ln(d(\tau)) \,\frac{{\rm d}\tau}{\tau}.
\end{gather*}
Thus
\begin{gather}\label{T+H be0}
\int_{\T}\ln(d(\tau)) \frac{{\rm d}\tau}{\tau}=0, \qquad \mbox{and therefore,} \qquad \be(0)=1.
\end{gather}
Next, we show that $\be=\tilde{\be}$. Note that
\begin{gather*}
\begin{split}
\Tilde{\be}(z) & = \exp \bigg[ \frac{1}{2 \pi {\rm i} } \int_{\T} \frac{\ln(d(\tau))}{\tau-z^{-1}}\,{\rm d}\tau \bigg] = \exp \bigg[ -\frac{z}{2 \pi {\rm i} } \int_{\T} \frac{\ln(d(\tau))}{\tau^{-1}-z} \,\frac{{\rm d}\tau}{\tau} \bigg]
\\
& = \exp \bigg[ -\frac{z}{2 \pi {\rm i} } \int_{\T} \frac{\ln\big(\Tilde{d}(\tau)\big)}{\tau-z} \frac{{\rm d}\tau}{\tau} \bigg] = \exp \bigg[ \frac{z}{2 \pi {\rm i} } \int_{\T} \frac{\ln(d(\tau))}{\tau-z} \,\frac{{\rm d}\tau}{\tau} \bigg],
\end{split}
\end{gather*}
where we have again used the fact that $d\tilde{d}\equiv 1$ on the unit circle. Using
\begin{gather*}
\frac{1}{(\tau-z)\tau} = \frac{z^{-1}}{\tau-z} - \frac{z^{-1}}{\tau},
\end{gather*}
we can write the last expression for $\tilde{\be}$ as
\begin{gather*}
\Tilde{\be}(z) = \exp \bigg[ \frac{1}{2 \pi {\rm i} } \int_{\T} \frac{\ln(d(\tau))}{\tau-z}\,{\rm d}\tau - \frac{1}{2 \pi {\rm i} } \int_{\T} \frac{\ln(d(\tau))}{\tau} \, {\rm d}\tau \bigg] = \frac{\be(z)}{\be(0)} = \be(z),
\end{gather*}
by \eqref{T+H be0}. To show that $\be=\tilde{\be}$ one could also argue that they both solve \textit{the same} scalar RHP which has a unique solution. We also note that $\al(z),\be(z) = 1 + \mathcal{O}(z^{-1})$, and $\tilde{\al}(z)= \al(0)(1+\mathcal{O}(z^{-1}))$ as $z \to \infty$. Now we can write the solution of the $\La$-RHP (in the case $d\tilde{d}\equiv 1$ on $\T$) as
\begin{gather}\label{Lambda}
\La(z)= \La^{-1}_{\infty}\!\begin{pmatrix}
1 & 0 & 0 & 0 \\
\mathcal{C}_{\rho}(z) & 1 & 0 & 0 \\
0 & 0 & 1 & 0 \\
0 & 0 & 0 & 1
\end{pmatrix} \di \!\times \!\begin{cases}
\begin{pmatrix}
\di -\be(z) & 0 & 0 & 0 \\
0 & 0 & \di \frac{1}{\tilde{\al}(z)\be(z) \al(z)} & 0 \\
0 & \di -\tilde{\al}(z) & 0 & 0 \\
0 & 0 & 0 & \di -\al(z)
\end{pmatrix}\!, & |z|<1, \\
\begin{pmatrix}
0 & \be(z) & 0 & 0 \\
0 & 0 & 0 & \di \frac{1}{\be(z)\tilde{\al}(z)\al(z)} \\
0 & 0 & \tilde{\al}(z) & 0 \\
\al(z) & 0 & 0 & 0
\end{pmatrix}\!, & |z|>1,
\end{cases}\hspace{-10mm}
\end{gather}
where $\mathcal{C}_f(z)$ is the Cauchy-transform of $f(z)$:
\begin{gather*}
\mathcal{C}_f(z) = \frac{1}{2\pi {\rm i} } \int_{\T} \frac{f(\tau)}{\tau-z}\, {\rm d}\tau,
\end{gather*}
and
\begin{gather}\label{Lambda infty}
\La^{-1}_{\infty} = \begin{pmatrix}
0 & 0 & 0 & 1 \\
1 & 0 & 0 & 0 \\
0 & 0 & \di \frac{1}{\al(0)} & 0 \\
0 & \al(0) & 0 & 0
\end{pmatrix}\!, \qquad \rho(z) = \di -\frac{1}{\be_-(z) \be_+(z) \tilde{\al}_-(z) \al_+(z) }.
\end{gather}
Using \eqref{T+H jumps of Szego functions}, the Plemelj--Sokhotskii formula and general properties of the Cauchy integral, it~can be checked that $\La$ given by (\ref{Lambda}) satisfies the $\La$-RHP.

\subsection[The small-norm Riemann--Hilbert problem associated to D_n(phi,dphi,1,1)]{The small-norm Riemann--Hilbert problem associated to $\boldsymbol{D_n(\phi,d\phi,1,1)}$}

Let us consider
\begin{gather*}
R(z;n):= S(z;n) \overset{\circ}{S}(z)^{-1}.
\end{gather*}
This function clearly has no jumps on $\Ga_i$, $\Ga_o$ and $\T$, since $S$ and $\overset{\circ}{S}$ have the same jumps on these contours. Thus, $R$ satisfies the following small-norm Riemann--Hilbert problem

\begin{itemize}\itemsep=0pt
\item \textbf{RH-$R$1} \quad $R$ is holomorphic in $\C \setminus \Ga_R$.

\item \textbf{RH-$R$2} \quad $R_+(z;n)=R_-(z;n)J_{R}(z;n)$, for $z \in \Ga_R$.
\item \textbf{RH-$R$3} \quad As $z \to \infty$, $R(z;n)= I + \mathcal{O}\big(z^{-1}\big)$,
\end{itemize}
where $\Ga_R := \Ga_i' \cup \Ga_o'$, and $J_R$ is given by
\begin{gather*}
J_R(z;n) = \overset{\circ}{S}(z) J_S(z;n) \overset{\circ}{S}(z)^{-1} = \begin{cases}
\overset{\circ}{S}(z) J_{T,i}(z;n) \overset{\circ}{S}(z)^{-1}, & z \in \Ga_i', \\
\overset{\circ}{S}(z) J_{T,o}(z;n) \overset{\circ}{S}(z)^{-1}, & z \in \Ga_o'.
\end{cases}
\end{gather*}
Using \eqref{Lambda}, \eqref{Lambda infty}, \eqref{9}, \eqref{GTfactorization} and \eqref{GXfactorization} we find
\begin{gather}\label{JR000}
J_R(z;n)-I = \begin{cases}
z^n \begin{pmatrix}
0 & g_{12}(z) & 0 & g_{14}(z) \\
0 & 0 & g_{23}(z) & 0 \\
0 & 0 & 0 & 0 \\
0 & 0 & g_{43}(z) & 0
\end{pmatrix}\!, & z \in \Ga_i', \\
z^{-n} \begin{pmatrix}
0 & 0 & 0 & 0 \\
g_{21}(z) & 0 & 0 & 0 \\
0 & g_{32}(z) & 0 & g_{34}(z) \\
g_{41}(z) & 0 & 0 & 0
\end{pmatrix}\!, & z \in \Ga_o',
\end{cases}
\end{gather}
where
\begin{gather*}
 g_{12}(z) = - \frac{\al(z)}{\phi(z)\be(z)} - \frac{\tilde{w}(z) C_{\rho}(z)}{\phi(z)\be(z)\tilde{\al}(z)}, \qquad
 g_{14}(z) = \frac{\tilde{w}(z)}{\phi(z)\be(z)\tilde{\al}(z) \al(0) },
 \\
 g_{23}(z) = - \frac{\al(0) \tilde{w}(z) \be(z) }{\tilde{\phi}(z)\tilde{\al}(z)}, \qquad
 g_{43}(z) = - \al^2(0) \bigg( \frac{\al(z) \be(z)}{\tilde{\phi}(z)} + \frac{\be(z)\tilde{w}(z) C_{\rho}(z)}{\tilde{\al}(z)\tilde{\phi}(z)} \bigg),
 \\
 g_{21}(z) = \frac{w(z) \be(z) }{\phi(z)\al(z)}, \qquad
 g_{32}(z) = - \frac{1}{\al(0)\tilde{\phi}(z)} \bigg( \frac{\tilde{\al}(z)}{\be(z)} - w(z)\tilde{\al}^2(z)\be(z)\al(z)C_{\rho}(z) \bigg),
 \\
 g_{34}(z) = \frac{w(z) \tilde{\al}^2(z)\be(z)\al(z) }{\tilde{\phi}(z)\al^2(0)}, \qquad
 g_{41}(z) = -\frac{\al(0)}{\phi(z)} \bigg( \frac{1}{\tilde{\al}(z)\be(z)\al^2(z)} - \frac{w(z)\be(z)C_{\rho}(z)}{\al(z)} \bigg).
\end{gather*}
From (\ref{JR000}) it follows that the jump matrix $J_R$ satisfies on $\Gamma_R$
the small-norm estimate,
\begin{gather}\label{smallnorn000}
||J_R - I||_{L_{2} \cap L_{\infty}} \leq C {\rm e}^{-cn},
\end{gather}
for some positive $C$ and $c = -\log r_1$, where $r_1$ is any number satisfying the condition
$r_0 < r_1 < 1$ (see (\ref{annulus}) and (\ref{r0}) for the definition and meaning of
the number $r_0$). Therefore, by standard theory of small-norm Riemann--Hilbert problems \cite{Deiftetal,Deiftetal2}, there exists $n_*$ such that for all $n > n_*$ the $R$-RH problem is solvable and
\begin{gather*}
R(z) = I + R_1(z) + R_2(z) + R_3(z) + \cdots, \qquad
z \in \C \setminus \Ga_R, \qquad
n \geq n_*,
\end{gather*}
where each $R_k$ is of order $O\big({\rm e}^{-kcn}\big)$ and they can be found recursively from
\begin{gather}\label{Rk}
R_k(z) = \frac{1}{2\pi {\rm i}}\int_{\Ga_R} \frac{[ R_{k-1}(\mu)]_- (J_R(\mu)-I)}{\mu-z}{\rm d}\mu, \qquad
z \in \C \setminus \Ga_R, \qquad
k\geq1.
\end{gather}
Note that this recurrence also means that
\begin{gather*}
R_{k+1}(n) = o(R_k(n)),\qquad
n \to \infty.
\end{gather*}
In particular,
\begin{gather}
R_1(z;n)=\frac{1}{2\pi {\rm i}}\int_{\Ga_R} \frac{J_R(\mu;n)-I}{\mu-z} \, {\rm d}\mu \nonumber
\\
\label{R1000}\phantom{R_1(z;n)}
= \begin{pmatrix}
0 & R_{1,12}(z;n) & 0 & R_{1,14}(z;n) \\
R_{1,21}(z;n) & 0 & R_{1,23}(z;n) & 0 \\
0 & R_{1,32}(z;n) & 0 & R_{1,34}(z;n) \\
R_{1,41}(z;n) & 0 & R_{1,43}(z;n) & 0
\end{pmatrix}\!,
\end{gather}
where
\begin{gather}
\label{T+H entries of R_1}
\begin{split}
& R_{1,jk}(z;n) = \frac{1}{2\pi {\rm i}}\int_{\Ga'_i} \frac{\mu^ng_{jk}(\mu)}{\mu-z}\,{\rm d}\mu, \qquad jk=12,14,23,43,
\\
& R_{1,jk}(z;n) = \frac{1}{2\pi {\rm i}}\int_{\Ga'_o} \frac{\mu^{-n}g_{jk}(\mu)}{\mu-z}\, {\rm d}\mu, \qquad jk=21,32,34,41.
\end{split}
\end{gather}
Also from \eqref{Rk} we can write an expression for $R_2$:
\begin{gather*}
R_2(z;n)= \frac{1}{2\pi {\rm i}}\int_{\Ga_R} \frac{[ R_{1}(\mu;n)]_- (J_R(\mu;n)-I)}{\mu-z}\,{\rm d}\mu
\\
\phantom{R_2(z;n)}= \begin{pmatrix}
R_{2,11}(z;n) & 0 & R_{2,13}(z;n) & 0 \\
0 & R_{2,22}(z;n) & 0 & R_{2,24}(z;n) \\
R_{2,31}(z;n) & 0 & R_{2,33}(z;n) & 0 \\
0 & R_{2,42}(z;n) & 0 & R_{2,44}(z;n)
\end{pmatrix}\!,
\end{gather*}
where
\begin{gather*}
R_{2,kj}(z;n) = \begin{cases}
\di \sum_{\ell \in \{2,4\}} \frac{1}{2\pi {\rm i}} \int_{\Ga'_o} \frac{\mu^{-n} [R_{1,k\ell}(\mu;n)]_-g_{\ell j}(\mu)}{\mu-z}\,{\rm d}\mu, \qquad j=1,\ k=1,3,
\\
\di \sum_{\ell \in \{2,4\}} \frac{1}{2\pi {\rm i}} \int_{\Ga'_i} \frac{\mu^{n} [R_{1,k\ell}(\mu;n)]_-g_{\ell j}(\mu)}{\mu-z}\,{\rm d}\mu, \qquad j=3,\ k=1,3,
\\
\di \frac{1}{2\pi {\rm i}} \int_{\Ga'_i} \frac{\mu^{n} [R_{1,k1}(\mu;n)]_-g_{1j}(\mu)}{\mu-z}\,{\rm d}\mu
\\ \phantom{\frac{1}{2\pi {\rm i}} \int_{\Ga'_i}}
\di +\frac{1}{2\pi {\rm i}} \int_{\Ga'_o} \frac{\mu^{-n} [R_{1,k3}(\mu;n)]_-g_{3j}(\mu)}{\mu-z}\,{\rm d}\mu, \qquad k, j=2,4.
\end{cases}
\end{gather*}
Moreover, using \eqref{Rk} and a straightforward calculation one can justify that the matrix structure (i.e., the location of zero and nonzero elements) of $R_{2k+1}$ and $R_{2k}$, $k\geq 1$, are similar to that of~$R_1$ and $R_2$, respectively. It is also straightforward to show that
 \begin{gather*}
R_{k,ij}(z;n) = \frac{\bigO{\big({\rm e}^{-kcn}\big)}}{|z|+1}, \qquad
n \to \infty, \qquad
k\geq 1,
\end{gather*}
uniformly for $z \in \C \setminus \Ga_R$, and the positive constant $c$ is the same as in (\ref{smallnorn000}).

\subsection[Asymptotics of h{_}n]{Asymptotics of $\boldsymbol{h_n}$}\label{Section T+H Asymp h_n}

The analysis of the previous section shows that the $\mathcal{X}$-RH problem has unique solution for all
$n > n_{*}$. We now proceed to the reconstruction of the corresponding functions $\overset{\circ}{\mathcal{X}}(z; n)$,
$\mathcal{Y}(z; n)$, and to the asymptotics of~$h_n$. To this end we need the asymptotic information about
the matrix~$P(n)$ which is needed to determine constants $C_j$ that participate in equation~(\ref{2 by 4 solution unique}).

Tracing back the Riemann--Hilbert transformations, we find that for $z \in \Om_0$ we have
\begin{equation}\label{T+H Trace back X}
\mathcal{X}(z;n)=R(z;n)\La(z)\begin{pmatrix}
1 & 0 & 0 & 0\\
0 & z^{-n} & 0 & 0 \\
0 & 0 & 1 & 0 \\
0 & 0 & 0 & z^{n}
\end{pmatrix}\!.
\end{equation}
From \eqref{T+H P from X} and \eqref{T+H Trace back X} we conclude that \begin{equation}\label{T+H P(n)}
P(n) = R(0;n)\La(0).
\end{equation}
From (\ref{R1000}) we arrive at the estimate
\begin{gather}
R(0;n) = I + R_1(0;n) +O\big({\rm e}^{-2cn}\big) \nonumber
\\
\label{R0n000}\phantom{R(0;n)}
= \begin{pmatrix}
1 & R_{1,12}(0;n) & 0 & R_{1,14}(0;n) \\
R_{1,21}(0;n) & 1 & R_{1,23}(0;n) & 0 \\
0 & R_{1,32}(0;n) & 1 & R_{1,34}(0;n) \\
R_{1,41}(0;n) & 0 & R_{1,43}(0;n) & 1
\end{pmatrix} + O\big({\rm e}^{-2cn}\big).
\end{gather}
Simultaneously, from \eqref{Lambda} and \eqref{Lambda infty} we have
\begin{gather}\label{T+H LA(0)}
\La(0)= \begin{pmatrix}
0 & 0 & 0 & -\al(0) \\
-1 & 0 & 0 & 0 \\
0 & -\di \frac{1}{\al(0)} & 0 & 0 \\
-C_{\rho}(0)\al(0) & 0 & 1 & 0
\end{pmatrix}\!.
\end{gather}
Equations (\ref{T+H P(n)}), (\ref{R0n000}) and (\ref{T+H LA(0)}) yield the following
asymptotic formula for $P(n)$,
\begin{gather}
P(n)\!=\!\!\begin{pmatrix}
-C_{\rho}(0)\al(0)R_{1,14}(0;\!n)\!-\!R_{1,12}(0;\!n) & 0 & R_{1,14}(0;\!n) & -\al(0)
\\[1ex]
-1 & -\di \frac{R_{1,23}(0;\!n)}{\al(0)} & 0 & -\al(0)R_{1,21}(0;\!n)
\\[2ex]
-C_{\rho}(0)\al(0)R_{1,34}(0;\!n)\!-\!R_{1,32}(0;\!n) & -\di \frac{1}{\al(0)} & R_{1,34}(0;\!n) & 0
\\[2ex]
-C_{\rho}(0)\al(0) & -\di \frac{R_{1,43}(0;\!n)}{\al(0)} & 1 & -\al(0)R_{1,41}(0;\!n)
\nonumber
\end{pmatrix}
\\ \label{T+H P(n) Expansion}\phantom{P(n)}
+ \bigO{\big({\rm e}^{-2cn}\big)},
\end{gather}
as $n \to \infty$.

It is time now for the conditions of Lemma~\ref{Lemma Unique Solvability of C's}.
We are not going to study each and every condition of this lemma, rather as a case study
we consider in particular the condition~\eqref{gencond}:
\begin{gather}\label{condition 2222}
(1-P_{21}(n))P_{42}(n) + P_{22}(n)P_{41}(n) \neq 0.
\end{gather}
From (\ref{T+H P(n) Expansion}) we have that
\begin{gather}\label{condition22220}
(1-P_{21}(n))P_{42}(n) + P_{22}(n)P_{41}(n) = -\mathcal{E}(n) + O\big({\rm e}^{-2cn}\big), \qquad
n \to \infty,
\end{gather}
where (cf.~(\ref{En0001}))
\begin{gather}\label{En000}
\mathcal{E}(n):= \frac{2}{\al(0)}R_{1,43}(0;n)-C_{\rho}(0)R_{1,23}(0;n).
\end{gather}
This is when we arrive at condition (\ref{Enneq0}) of Theorem \ref{T+H main thm}. Indeed, let us suppose that there exist such $C>0$ and $n_0 \geq n_{*}$ that
\begin{gather*}
|\mathcal{E}(n)| \geq Cr^{n}, \qquad
\mbox{for some}\quad r\colon
r_{0}\leq r < 1, \qquad \mbox{and}\qquad n> n_0,
\end{gather*}
and choose $r_1$ in the definition of the constant $c$ so that $ r< r_1 <1$ and $r^{2}_1 < r$. Then, estimate (\ref{condition22220}) can be rewritten as
\begin{gather}\label{cond111}
(1-P_{21}(n))P_{42}(n) + P_{22}(n)P_{41}(n) = -\mathcal{E}(n)\big( 1 + O\big({\rm e}^{-c_1n} \big)\big), \qquad
n \to \infty,
\end{gather}
where $c_1 = -\log\big(\frac{r^2_1}{r}\big) >0$. The last estimate in turn means that there exists such \mbox{$n_1 \geq n_0$} that for all $n > n_1 +1$
we should have that condition (\ref{condition 2222}) holds for $n$ and $n-1$ and
hence by~Lemma~\ref{Lemma Unique reconstruction of Y from X} and Corollary (\ref{corollary2.8.1}), we can uniquely reconstruct the solution of the $\mathcal{Y}$-RHP, having already the unique solution of the $\mathcal{X}$-RHP, and, moreover,
we could use equation (cf.~\eqref{T+H h_n000}),
\begin{gather}\label{T+H h_n good}
-\frac{1}{h_{n-1}} = \lim_{z \to 0} z^{n-1} \mathcal{Y}_{21}\big(z^{-1};n\big).
\end{gather}
for evaluation of the large $n$ behavior of $h_n$.
Let us denote \begin{gather}\label{T+H A}
\mathcal{A}(z;n) := P^{-1}(n)\mathcal{X}(z;n)\begin{pmatrix}
1 & 0 & 0 & 0\\
0 & z^{n} & 0 & 0 \\
0 & 0 & 1 & 0 \\
0 & 0 & 0 & z^{-n}
\end{pmatrix}\!,
\end{gather}
and also let us define the matrix $\mathcal{B}(n)$ in the following expansion for $\mathcal{A}(z;n)$, which is equivalent to \textbf{RH-$\mathcal{X}$4}:
\begin{gather}\label{T+H A exp zero}
\mathcal{A}(z;n) = I + \mathcal{B}(n)z + \bigO{\big(z^2\big)}, \qquad
z \to 0.
\end{gather}
Therefore by \eqref{24-to-44}, \eqref{mathfrakR}, \eqref{C's-to-P} and \eqref{T+H A} we can write
\begin{gather}\label{X-naught-good-equation}
\overset{\circ}{\mathcal{X}}(z,n) = \begin{pmatrix}
C_1(n) & 1 & C_3(n) & 0 \\
C_2(n) & 0 & C_4(n) & 1
\end{pmatrix} \mathcal{A}(z;n) \begin{pmatrix}
1 & 0 & 0 & 0\\
0 & z^{-n} & 0 & 0 \\
0 & 0 & 1 & 0 \\
0 & 0 & 0 & z^{n}
\end{pmatrix}\!.
\end{gather}
Using \eqref{X naught to Y} and \eqref{X-naught-good-equation} we can write
\begin{gather}\label{Y_21(z^-1) to C's}
\mathcal{Y}_{21}\big(z^{-1};n\big) = \overset{\circ}{\mathcal{X}}_{22}(z;n) = C_2(n) \mathcal{A}_{12}(z;n)z^{-n} \!+ C_4(n)\mathcal{A}_{32}(z;n)z^{-n} \!+ \mathcal{A}_{42}(z;n)z^{-n}.\!
\end{gather}
From \eqref{T+H A exp zero} we have
\begin{gather*}
z^{-n}\mathcal{A}(z;n) = z^{-n} I + z^{-n+1}\mathcal{B}(n) + \mathcal{O}\big(z^{-n+2}\big), \qquad
z \to 0.
\end{gather*}
Therefore, as $z \to 0$
\begin{gather}\label{T+H Asym A vs. B}
z^{-n}\mathcal{A}_{ij}(z;n) =
\begin{cases}
z^{-n+1}\mathcal{B}_{ij}(n) + \mathcal{O}\big(z^{-n+2}\big), & i \neq j,
\\
z^{-n} + z^{-n+1}\mathcal{B}_{ii}(n) + \mathcal{O}\big(z^{-n+2}\big), & i=j.
\end{cases}
\end{gather}
Therefore by \eqref{T+H h_n good}, \eqref{Y_21(z^-1) to C's} and \eqref{T+H Asym A vs. B} we have
\begin{gather}\label{T+H h_n C's B's}
-\frac{1}{h_{n-1}} = C_2(n) \mathcal{B}_{12}(n) + C_4(n)\mathcal{B}_{32}(n) + \mathcal{B}_{42}(n).
\end{gather}
Let us denote the coefficients in the expansions of $R(z;n)$ and $\La(z)$, as $z \to 0$, by
\begin{gather*}
R(z;n) = R(0;n) + R^{(1)}(n) z + R^{(2)}(n) z^2 + \bigO{\big(z^3\big)},
\\
\La(z) = \La(0) + \La^{(1)} z + \La^{(2)} z^2 + \bigO{\big(z^3\big)}.
\end{gather*}
From \eqref{T+H A exp zero}, \eqref{T+H Trace back X}, and \eqref{T+H P(n)} we have
\begin{gather}\label{T+H B}
\mathcal{B}(n) = \La^{-1}(0)R^{-1}(0;n)R^{(1)}(n)\La(0) + \La^{-1}(0)\La^{(1)}.
\end{gather}
Note that
\begin{gather*}
R^{(1)}(n) = \frac{1}{2\pi {\rm i}} \int_{\Ga_R}\! (J_R(\mu;n)-I) \frac{{\rm d}\mu}{\mu^2} + \mathcal{O}\big({\rm e}^{-2cn}\big), \qquad\!\!
R^{-1}(0;n) = I - R_1(0;n) + \bigO{\big({\rm e}^{-2cn}\big)},
\end{gather*}
as $n \to \infty$. More precisely, we have
\begin{gather}\label{T+H R^(1)}
R^{(1)}(n)= \begin{pmatrix}
0 & R^{(1)}_{12}(n) & 0 & R^{(1)}_{14}(n) \\
R^{(1)}_{21}(n) & 0 & R^{(1)}_{23}(n) & 0 \\
0 & R^{(1)}_{32}(n) & 0 & R^{(1)}_{34}(n) \\
R^{(1)}_{41}(n) & 0 & R^{(1)}_{43}(n) & 0
\end{pmatrix} + \mathcal{O}\big({\rm e}^{-2cn}\big), \end{gather}
as $n \to \infty$, where\begin{gather}\label{T+H Entries of the subleading term of R in expansion near zero}
\begin{split}
& R^{(1)}_{jk}(n) = \frac{1}{2\pi {\rm i}}\int_{\Ga'_i} \mu^{n-2}g_{jk}(\mu)\,{\rm d}\mu, \qquad \ \ jk=12,14,23,43, \\
& R^{(1)}_{jk}(n) = \frac{1}{2\pi {\rm i}}\int_{\Ga'_o} \mu^{-n-2}g_{jk}(\mu)\,{\rm d}\mu, \qquad jk=21,32,34,41,
\end{split}
\end{gather}
and
\begin{gather}\label{T+H R-inverse at 0}
R^{-1}(0;n)\!=\!\begin{pmatrix}
1 & -R_{1,12}(0;n) & 0 & -R_{1,14}(0;n) \\
-R_{1,21}(0;n) & 1 & -R_{1,23}(0;n) & 0 \\
0 & -R_{1,32}(0;n) & 1 & -R_{1,34}(0;n) \\
-R_{1,41}(0;n) & 0 & -R_{1,43}(0;n) & 1
\end{pmatrix}\! + \bigO{\big({\rm e}^{-2cn}\big)},
\\ \qquad
 n \to \infty.\nonumber
\end{gather}
Regarding the coefficients of $\Lambda(z)$, we will actually only need
the zero-term, $\Lambda(0)$, which has already been presented in (\ref{T+H LA(0)}).
From \eqref{T+H B}, \eqref{T+H R^(1)}, \eqref{T+H R-inverse at 0} and \eqref{T+H LA(0)} we find that
\begin{gather}
 \mathcal{B}_{12}(n) = \frac{R^{(1)}_{23}(n)}{\al(0)} + \mathcal{O}\big({\rm e}^{-2cn}\big), \qquad
\mathcal{B}_{32}(n) = C_{\rho}(0)R^{(1)}_{23}(n)-\frac{R^{(1)}_{43}(n)}{\al(0)} + \mathcal{O}\big({\rm e}^{-2cn}\big), \nonumber\\ \mathcal{B}_{42}(n) = - \frac{1}{\al^2(0)} \big( R_{1,12}(0;n)R^{(1)}_{23}(n)+R_{1,14}(0;n)R^{(1)}_{43}(n) \big)
+\mathcal{O}\big({\rm e}^{-3cn}\big).\label{T+H Selected Elements of B}
\end{gather}
Note that $\mathcal{B}_{12}(n)$, $\mathcal{B}_{32}(n)$ are of order $\bigO{\big({\rm e}^{-cn}\big)}$, while $\mathcal{B}_{42}(n)$ is of order $\bigO{\big({\rm e}^{-2cn}\big)}$.

Revisiting \eqref{C's-to-P} we have
\begin{gather*}
\begin{pmatrix}
1 & C_1(n) & 0 & C_3(n) \\
0 & C_2(n) & 1 & C_4(n)
\end{pmatrix} P(n) = \begin{pmatrix}
C_1(n) & 1 & C_3(n) & 0 \\
C_2(n) & 0 & C_4(n) & 1
\end{pmatrix}\!.
\end{gather*}
In particular, in view of \eqref{condition 2222}, we can write the following two equations for the constants $C_2$ and $C_4$
\begin{gather*}
C_2(n) P_{21}(n) + P_{31}(n) + C_4(n)P_{41}(n) = C_2(n),
\\
C_2(n)P_{22}(n) + P_{32}(n) + C_4(n)P_{42}(n) = 0.
\end{gather*}
Solving for $C_2$ and $C_4$ we find
\begin{gather*}
C_2(n) = \frac{P_{42}(n)P_{31}(n)-P_{41}(n)P_{32}(n)}{(1-P_{21}(n))P_{42}(n)+P_{41}(n)P_{22}(n)}, \\ C_4(n) =- \frac{P_{22}(n)P_{31}(n)+[1-P_{21}(n)]P_{32}(n)}{(1-P_{21}(n))P_{42}(n)+P_{41}(n)P_{22}(n)}.
\end{gather*}
From \eqref{T+H P(n) Expansion} and (\ref{cond111}) we have
\begin{gather}\label{T+H C2 ghadim}
C_2(n)= \frac{C_{\rho}(0)}{\mathcal{E}(n) } \big( 1+ \bigO{\big({\rm e}^{-c_1n}\big)}\big),
\end{gather}
and
\begin{gather}\label{T+H C4 ghadim}
C_4(n) = - \frac{2}{\al(0)\mathcal{E}(n) } \big( 1+ \bigO{\big({\rm e}^{-c_1n}\big)}\big).
\end{gather}
Combining \eqref{T+H h_n C's B's}, \eqref{T+H Selected Elements of B}, \eqref{T+H C2 ghadim} and \eqref{T+H C4 ghadim} we obtain
\begin{gather}\label{T+H h_n asymp 1}
h_{n-1} = -\al(0) \frac{\mathcal{E}(n)}{\frac{2}{\al(0)}R^{(1)}_{43}(n)-C_{\rho}(0)R^{(1)}_{23}(n)}\big(1+ \bigO{\big({\rm e}^{-c_1n}\big)}\big), \qquad n \to \infty.
\end{gather}
Note that from \eqref{T+H entries of R_1} and \eqref{T+H Entries of the subleading term of R in expansion near zero} we have
\begin{gather*}
R_{1,jk}(0;n) = R^{(1)}_{jk}(n+1), \qquad \mbox{for} \quad jk=12,14,23,43,
\\
R_{1,jk}(0;n) = R^{(1)}_{jk}(n-1), \qquad \mbox{for} \quad jk=21,32,34,41.
\end{gather*}
This allows us to rewrite \eqref{T+H h_n asymp 1} as
\begin{gather}\label{T+H h_n asymp 2}
h_{n-1}= -\alpha(0)\frac{\mathcal{E}(n)}{\mathcal{E}(n-1)}\big(1+ \bigO{\big({\rm e}^{-c_1n}\big)}\big), \qquad n \to \infty.
\end{gather}
This concludes the proof of Theorem~\ref{T+H main thm}.

Let $\mathfrak{C}$ denote the class of symbol pairs $(\phi,d \phi)$, where $\phi$ and $d$ satisfy the properties mentioned in Theorem~\ref{T+H main thm}, for~which the corresponding matrix $P(n)$ satisfies the condition~\eqref{condition 2222} for~sufficiently large $n$. Now, consider the subclass $\mathfrak{C}_0 \subset \mathfrak{C}$, which includes symbol pairs $(\phi,d \phi)$ for which~$ \mathcal{E}(n)$ satisfies inequality (\ref{Enneq0}), or rather its generalization~\eqref{Enneq00}
for sufficiently large~$n$ (see also Section~\ref{characterization}). For the purposes of this paper, it is worthwhile to prove that the subclass~$\mathfrak{C}_0$ is not empty by providing an explicit example. To this end, let us consider the symbol pairs $(\phi, d\phi)$, where
\begin{gather}\label{phi example}
\phi(z) = \left(\frac{z-b}{z-a}\right)^{\alpha},
\end{gather}
and
\begin{gather}\label{d example}
d(z) =\left(\frac{z-b_1}{z-a_1}\right)^{\alpha_1}\left(\frac{a_1z-1}{b_1z-1}\right)^{\alpha_1}.
\end{gather}
We assume that $0< a < b < a_1< b_1 < 1,$ and $|\Re \alpha_1| < 1.$ Then, for $\mathcal{E}(n)$ defined in \eqref{En000}, by~standard Watson lemma type asymptotic analysis of integrals we will arrive at
\begin{gather}\label{En asymp}
\mathcal{E}(n) = \kappa b_1^{n-\alpha_1}n^{\alpha_1 -1}\bigg(1 + \mathcal{O}\bigg(\frac{1}{n}\bigg)\bigg), \qquad
n \to \infty,
\end{gather}
where
\begin{gather}
\kappa = -\frac{{\rm i}}{\pi}{\rm e}^{-{\rm i}\pi \alpha_1}\Gamma(1-\alpha_1)(b_1-a_1)^{\alpha_1}\bigg(\frac{1-bb_1}{1-ab_1}\bigg)^{\alpha}
\bigg(\frac{a_1}{b_1}\bigg)^{-\alpha_1} \nonumber
\\ \label{2}\phantom{\kappa =}
\times \bigg[{\rm e}^{-{\rm i}\pi\alpha_1}\frac{{\rm i}}{\pi}\int_{1/b_1}^{1/a_1} \bigg(\frac{z-b_1}{z-a_1}\bigg)^{\alpha_1}
\bigg(\frac{1/a_1-z}{z-1/b_1}\bigg)^{-\alpha_1}\frac{z+b_1}{z-b_1}\frac{dz}{z} - 1\bigg],
\end{gather}
and
\begin{gather*}
\arg\bigg(\frac{z-b_1}{z-a_1}\bigg) = \arg\bigg(\frac{1/a_1-z}{z-1/b_1}\bigg) = 0, \qquad
z \in [1/b_1, 1/a_1].
\end{gather*}
We argue, that for generic values of $a$, $b$, $a_1$, $b_1$, $\alpha$ and $\alpha_1$ the number $\kappa \neq 0$. Indeed,
let us take, for instance, $\alpha_1$ real, i.e.,
\begin{gather*}
-1 < \alpha_1 <1.
\end{gather*}
Then, the integrand in \eqref{2} is strictly positive and hence, except for $\alpha_1 = \pm 1/2$, the first term in the brackets of (\ref{2}) has nonzero imaginary part and thus can not cancel the second term. It~is only needed now to notice that in the case under consideration $r_0 = b_1$ and hence inequality $\kappa \neq 0$ implies that
$\mathcal{E}(n)$ corresponding to the weights (\ref{phi example}) and (\ref{d example}) satisfies the condition~(\ref{Enneq00}) indicated in Remark \ref{Remark slight generalization} with $\varkappa(n)=n^{\al_1-1}$.

We conclude the discussion of this example by indicating the asymptotic behavior of the corresponding
norm parameter $h_{n-1}$ and the determinant $D_n$. To this end we notice that, in~fact, the estimate (\ref{En asymp}) can be extended to the full asymptotic series,
\begin{gather*}
\mathcal{E}(n) \sim \kappa b_1^{n-\alpha_1}n^{\alpha_1 -1}\Bigg(1 + \sum_{k=1}^{\infty} d_kn^{-k}\Bigg), \qquad n
\to \infty.
\end{gather*}
Therefore, from (\ref{T+H h_n asymp 2}) we should have that
\begin{gather*}
h_{n-1} = -b_1\bigg(\frac{n}{n-1}\bigg)^{\alpha_1 -1} \bigg(1 + O\bigg(\frac{1}{n^2}\bigg)\bigg), \qquad n \to \infty,
\end{gather*}
and hence
\begin{gather*}
D_n = C (-b_1)^{n } n^{\alpha_1 -1}\bigg(1 + O\bigg(\frac{1}{n}\bigg)\bigg), \qquad
n \to \infty,
\end{gather*} for some constant $C$ which is yet to be determined. It is interesting that this formula
indicates an oscillatory behavior of the determinant $D_n$ for large $n$. This fact is confirmed numerically.

\begin{Remark} From \eqref{OP Rep of Solution}, \eqref{X naught to Y}, \eqref{24-to-44}, and \eqref{mathfrakR} we have the following representation for the
orthogonal polynomials $\mathcal{P}_n(z)$ in terms of the solution $\mathcal{X}(z;n)$ of the
$\mathcal{X}$-RHP,
\begin{gather*}
\mathcal{P}_n(z) = \mathcal{X}_{11}(z;n) + C_1(n)\mathcal{X}_{21}(z;n) + C_3(n)\mathcal{X}_{41}(z;n).
\end{gather*}
The asymptotic results concerning the function $\mathcal{X}(z;n)$ obtained in this section can be translated to the
large $n$ asymptotic formulae for the polynomials $\mathcal{P}_n(z)$. Indeed, skipping the rather tedious
though straightforward calculations, we arrive at the following asymptotics for $\mathcal{P}_n(z)$
on the unit circle:
\begin{gather}
\mathcal{P}_n(z) = \frac{\be_+(z) (2\al(0)C_{\rho,+}(z) - C_{\rho}(0)\al(0))}{\mathcal{E}(n)} \nonumber
\\
\label{asymp-P_n-|z|=1}\phantom{\mathcal{P}_n(z)=}{}
+ \al_-(z)z^n \bigg( 1 + \frac{C_{\rho}(0)\al(0)R_{1,21}(z;n)-2R_{1,41}(z;n)}{\mathcal{E}(n)} \bigg) + \mathcal{O}\big({\rm e}^{-c_1n}\big),
\end{gather}
as $n \to \infty$. While in the interior and exterior of the unit circle we have the following asymptotic formulae for $\mathcal{P}_n(z)$:
\begin{gather*}
\mathcal{P}_n(z)= \frac{\be(z) ( 2\al(0)C_{\rho}(z) - C_{\rho}(0)\al(0) )}{\mathcal{E}(n)}+ \mathcal{O}\big({\rm e}^{-c_1n}\big), \qquad |z|<1,
\\
\mathcal{P}_n(z)= \al(z)z^n \bigg(1 + \frac{C_{\rho}(0)\al(0)R_{1,21}(z;n)-2R_{1,41}(z;n)}{\mathcal{E}(n)} + \mathcal{O}\big({\rm e}^{-2c_1n}\big) \bigg) , \qquad |z|>1,
\end{gather*}
as $n\to \infty$ (compare with \eqref{asymp-P_n-|z|=1}).
\end{Remark}

\begin{Remark}\label{Remark result for line}
In a similar fashion, we could have obtained the large $n$ asymptotic expression for the norm $h_n$ of orthogonal polynomials \eqref{T+H2 Orthogonality}. In that case, the methods of this section would work analogously if the function
\begin{gather*}
d:= -\phi^{-1}(z)\tilde{u}(z) \equiv \phi^{-1}(z) \int_{a}^{b} \frac{t^{s-1}w(t)}{1-tz} \, \mathrm{d} t,
\end{gather*} satisfies $d\tilde{d}\equiv1$ on the unit circle, or \begin{gather}\label{w on the line Szego-type phi}
\phi(z)\tilde{\phi}(z) = \int_{a}^{b} \frac{t^{s-1}w(t)}{1-tz} \mathrm{d} t \int_{a}^{b} \frac{t^{s-1}w(t)}{1-tz^{-1}} \, \mathrm{d} t, \qquad z\in \T.
\end{gather}
\end{Remark}

\section{Remaining open questions}\label{Section: suggestions for future work}
We consider this work as a starting point of a long term research project. There are many chal\-len\-ging technical as well as conceptual open questions related to the Riemann--Hilbert formalism we are suggesting in this paper. Here we highlight some of them that we consider the most pressing.

\subsection[Derivation of the relevant Christoffel--Darboux formulae and differential identities]
{Derivation of the relevant Christoffel--Darboux formulae\\ and differential identities}

Our main objective in this paper has been to develop a $4 \times 4$ steepest descent analysis for~certain Toeplitz+Hankel determinants and we have achieved that. However, to obtain the asymptotics of $D_n(\phi,d\phi,1,1)$ one has to derive suitable differential identities. We propose that the differential identity has to be with respect to the parameters $\al_i$ in the function $d$ given by~\eqref{T+H Def d}. Thus, one has to perform $m$ integrations in the parameters $\al_i$, $1 \leq i \leq m$. Note that for $\al_1 = \al_2 = \cdots = \al_m = 0$, we have $d \equiv 1$ and hence $\phi=w$. Hence the starting point of the integration in $\al_1$ could be taken from the results of~\cite{BE}.\footnote{Although the authors in~\cite{BE} do not particularly study the asymptotics of $D_n(\phi,\phi,1,1)$, this asymptotics is achievable by their methods.} Integration of the differential identity in $\al_1$ will provide us with an asymptotic expression for $D_n(\phi, d_1\phi, 1 , 1)$, which also serves as the starting point of~integration in $\al_2$. Thus we can find asymptotics of $D_n(\phi,d_1d_2\phi, 1 , 1)$ which also serves as the starting point of integration in $\al_3$, and so on. Repeating this procedure will finally lead us to the asymptotics of $D_n(\phi,d\phi, 1 , 1)$.

In order to derive the differential identities mentioned above, one has to find recurrence relations and prove a Christoffel--Darboux formula for the polynomials \eqref{T+H OP} and follow a path similar to that introduced by I.~Krasovsky in~\cite{Krasovsky}. Note that the recurrence relations can be found by analyzing the function $\mathcal{M}(z;n):= \mathcal{X}(z;n+1)\mathcal{X}^{-1}(z;n)$, which is holomorphic in $\C \setminus \{0\}$ and can be globally determined by its singular parts at zero and infinity.

\subsection [Extension of the analysis to general offset values r,s no= 1]{Extension of the analysis to general offset values $\boldsymbol{r,s \neq 1}$}\label{section-general-r,s}
If we lift the restriction $r=s=1$ then in the jump matrix of the $\mathcal{X}$-RHP
the functions~$\phi(z)$ and~$w(z)$ should be replaced by $z^{1-r}\phi(z)$
and $z^{1-s}w(z)$, respectively. This, in turn would raise a serious question about solvabilty
of the $\Lambda$-RHP. Indeed, for instance, we would not be able to define the function $\al$ in (\ref{alpha}) and
hence to factorize the jump matrix (\ref{T+H model jump reduced}). The way out of this difficulty could be the
use of the relation between the determinants $D_n(\phi,w;r,s)$ with different values of $r$, $s$ or both. Such relations
are well known in the pure Toeplitz case (for example see \cite[Lemma~2.4]{DIK}). However, for general Toeplitz+Hankel determinants they are yet to be found.
Another way to approach the problem could be to develop the so-called B\"{a}klund--Schlesinger transformations
of the $\mathcal{X}$-RHP itself. That is, the transformations of the form,
\begin{gather*}
\mathcal{X}(z) \mapsto R(z)\mathcal{X}(z),
\end{gather*}
where $R(z)$ is a properly chosen rational function for which the above transformation results in~the desired shifting of the parameters $r$ and $s$. Also, one can try
to allow the matrix $P(n)$ in~the setting of the $\mathcal{X}$-RHP to depend on $z$.

It is worth to point out that the problem with the extension of our scheme to the general values of $r$ and $s$ is not
actually a problem of the setting of the relevant Riemann--Hilbert problem. Indeed, it is rather the question of
the correct way to approach to its asymptotic solution. Let us demonstrate this point by considering the pure Toeplitz case.

Assume that $D_n(\phi,0;r,0) \neq 0$ for all $n$, so that $\mathcal{P}_n(z)$ and, correspondingly, the solution $\mathcal{Y}(z;n)$ of the $\mathcal{Y}$-RHP exist for all $n$. Put,
\begin{gather*}
\mathcal{X}(z;n):= P^{-1}_{\infty}(n)\begin{pmatrix}
\mathcal{Y}_{11}(z;n)&0&0&\widetilde{\mathcal{Y}}_{12}(z;n)\\
0&\widetilde{\mathcal{Y}}_{11}(z;n)&\mathcal{Y}_{12}(z;n)&0\\
0&\widetilde{\mathcal{Y}}_{21}(z;n)&\mathcal{Y}_{22}(z;n)&0\\
\mathcal{Y}_{21}(z;n)&0&0&\widetilde{\mathcal{Y}}_{22}(z;n)
\end{pmatrix}\!,
\end{gather*}
where the normalizing matrix $P_{\infty}(n)$ is given by
\begin{gather*}
P_{\infty}(n)= \begin{pmatrix}1&0&0&\mathcal{Y}_{12}(0;n)\\
0&\mathcal{Y}_{11}(0;n)&0&0\\
0&\mathcal{Y}_{21}(0;n)&1&0\\
0&0&0&\mathcal{Y}_{22}(0;n)
\end{pmatrix}\!,
\end{gather*}
which \looseness=1 is invertible for generic $\phi$. It is straightforward to check that the so defined $4\times 4$ matrix-valued function $\mathcal{X}$ solves the $\mathcal{X}$-Riemann--Hilbert problem, \textbf{RH-$\mathcal{X}$1}--\textbf{RH-$\mathcal{X}$4} with $w \equiv 0$, and~$\phi(z)$ replaced by~$z^{1-r}\phi(z)$. Take now $r=0$. That is, let us consider the standard orthogonal polynomials
on the circle with the weight~$\phi(z)$ having zero winding number. Then, from the standard $2\times2$
Riemann--Hilbert formalism \cite{BDJ}, we know everything about the asymptotic behavior of the corresponding orthogonal polynomials $\mathcal{P}_n(z)$ and hence we know asymptotic solution of the $\mathcal{X}$-Riemann--Hilbert problem corresponding
to $D_n(\phi,0;0,0)$. And, we know this in spite of the fact that the approach developed in the body of this work can not be satisfactorily used for the case $r=s=0$. We believe that this observation might entail a useful hint on how to~modify our Riemann--Hilbert approach to deal with general values of~$r$ and~$s$.

\subsection[Extension of the Riemann--Hilbert analysis of Section~\ref{Section Hankel on I} for more general choices of $I$]{Extension of the Riemann--Hilbert analysis of Section~\ref{Section Hankel on I}\\ for more general choices of $\boldsymbol{I}$}

We recall that our Riemann--Hilbert analysis of Section~\ref{Section Hankel on I}, with minimal modifications, naturally extends to the following three cases as well: \textit{i})~$-1<a<b<0$, \textit{ii})~$-\infty<a<b<-1$, and~\textit{iii})~$1<a<b<\infty$. A natural first step in generalizing the results of Section~\ref{Section Hankel on I} beyond the~above cases is considering the case where $I$ has $0$ as an end point. This would slightly affect the~analysis as one has to take into account the behavior of $w$ at $0$ in the set up of the $2\times4$ and the subsequent Riemann--Hilbert problems.

The other interesting case to be studied is when $I$ intersects the unit circle. Clearly, in this case one has to perform local analysis in a neighborhood of the intersection point(s) of $I$ and the unit circle. Although these local constructions are reminiscent of what one does near the~Fisher--Hartwig singularities or the endpoints of the~support of the symbol, here even if the possible intersection points $\pm1$, are regular points for non-FH symbols $\phi$ and $w$, one has to still perform local analysis due to collision of the supports of $\phi$ and $w$.

Another possible generalization would be to consider $I$ to be the union of two symmetric intervals with respect to the unit circle, i.e., $I=[a,b] \cup \big[b^{-1},a^{-1}\big]$. This generalization should be accessible by slight modification of our approach explained in Section~\ref{Section Hankel on I}. However, generalization to the case where $I$ is a union of two non-symmetrical intervals with respect to the unit circle needs a more special treatment.

\subsection{Extension to Fisher--Hartwig symbols}\label{T+H Extension to FH}
One can study the large-$n$ asymptotics of determinant $D_n(\phi,d \phi,1,1)$ (and with increasing effort $D_n(\phi,d \phi,r,s)$ for fixed $r, s \in \Z$) assuming that $\phi$ possesses Fisher--Hartwig singularities $\{z_i\}^{m}_{i=1}$ on the unit circle. It is in fact in this level of generality that E.~Basor and T.~Ehrhardt have been able to compute the asymptotics of $D_n(\phi,d \phi,0,1)$, $D_n(-\phi,d \phi,0,1)$, $D_n(\phi,d z \phi,0,1)$, and $D_n(-z\phi,d \phi,0,1)$ via the operator-theoretic methods in~\cite{BE}. However, the authors in~\cite{BE} further require that the Fisher--Hartwig part of $\phi$ be \textit{even}. In fact they used some results of the work~\cite{DIK} to prove their asymptotic formulas for Toeplitz+Hankel determinants, and for this reason they \textit{inherited} the evenness assumption from the work~\cite{DIK} where the authors needed evenness of $\phi$ in their $2\times 2$ setting to relate Hankel and Toeplitz+Hankel determinants to a Toeplitz determinant with symbol~$\phi$.

From a Riemann--Hilbert perspective, in the presence of Fisher--Hartwig singularities, one has to~construct the $4\times 4$ local parametrices near the points $z_i$. Expectedly, these local parametrices must be expressed in terms of confluent hypergeometric functions as suggested by~\cite{DIK}. We~have not yet worked out the details of this construction but we believe that it should be well within reach. It would be methodologically important to achieve the results obtained from operator-theoretic tools via the Riemann--Hilbert approach as well. Moreover, we expect that the evenness of the Fisher--Hartwig part of $\phi$ would not play a role in our $4 \times 4$ setting, and in that sense there are reasonable prospects of generalizing the results of~\cite{BE} to symbols~$\phi$ with non-even Fisher--Hartwig part.

\subsection[Characterization of generic classes of Szego-type symbol pairs (phi,dphi), with d tilde d equiv 1 on the unit circle]{Characterization of generic classes of Szeg{\H o}-type symbol pairs
$\boldsymbol{(\phi,d \phi)}$,\\ with $\boldsymbol{d\tilde{d}\equiv 1}$ on the unit circle}\label{characterization}

Take one of the six conditions of Lemma \ref{Lemma Unique Solvability of C's}. Denote by $\mathfrak{P}(n)$ and $\mathcal{E}(n)$, respectively, the~cor\-res\-ponding nonzero quantity and its leading order term in the large $n$ asymptotic expansion. Consider the class of symbol pairs $(\phi,d\phi)$, where $\phi$ and $d$ satisfy the properties mentioned in~Theorem~\ref{T+H main thm}. Within this class, take the sub\-class $\mathfrak{C}$ of symbol pairs for which the elements of the corresponding matrix $P(n)$ satisfy $\mathfrak{P}(n) \neq 0$ for sufficiently large $n$. Also consider the subclass $\mathfrak{C}_0 \subseteq \mathfrak{C}$ of symbols pairs for which $\mathcal{E}(n)$
satisfies (\ref{Enneq00}) for sufficiently large $n$. One should be able to find the asymptotics of the norm $h_n$ of orthogonal polynomials in a similar fashion as presented in Section~\ref{Section T+H Asymp h_n} in terms of $\mathcal{E}(n)$, assuming that $(\phi,d\phi) \in \mathfrak{C}_0$. It would be very interesting to completely or partially characterize the classes of symbol pairs $\mathfrak{C}$ and $\mathfrak{C}_0$, corresponding to each one of the six conditions of Lemma \ref{Lemma Unique Solvability of C's}.

\looseness=1 Moreover, although we have provided explicit examples for a class of Szeg{\H o}-type functions~$d$ which further satisfy $d\tilde{d}\equiv1$ on the unit circle (see the beginning of Section~\ref{T+H solvable}), a complete characterization of such functions is currently unknown to the authors. Also, regarding what we discussed in Section~\ref{Section Hankel on I}, for a given Szeg{\H o}-type symbol $\phi$, we are very interested to find the associated class of weights $w$, supported on the interval, for which the equality \eqref{w on the line Szego-type phi} holds.

\subsection{Characteristic polynomial of a Hankel matrix}

As mentioned in the Introduction, arguably the most important motivation behind studying the asymptotics of Toeplitz+Hankel determinants is to study the large $n$ asymptotics of the eigenvalues of the matrix $H_n[w]$. We recall that the characteristic polynomial $\det(H_n[w]-\la I)$ is indeed the Toeplitz+Hankel determinant $D_n(-\la,w,0,0)$. In this case the associated $\La$-model Riemann--Hilbert problem needs a special treatment. In a sense it is a simpler problem as the symbol $\phi$ is identically equal to a constant, but more complicated -- compared to the situation in Section~\ref{T+H solvable} -- as it does not enjoy $J_{\La,23}(z)=0$. In any case, the solution to this model problem provides us with the constant term in the asymptotics of $D_n(-\la,w,0,0)$, and in the case of Fisher--Hartwig weight $w$, one can hope to obtain the leading terms of this asymptotic expansion (up to the constant term, viz. the solution of the $\La$-model problem) from the local analysis near the Fisher--Hartwig singularities.{\footnote{ The very recent work~\cite{FedeleGebert} shows that indeed one can explicitly describe the leading terms of the asymptotics of $D_n(-\lambda, w, 0,0)$, in the case of $w$ having jump discontinuites and $\lambda$ sufficiently large.}} This last point is yet another motivation to pursue the goals of Section~\ref{T+H Extension to FH}.

\subsection*{Acknowledgements}
We are very grateful to Estelle Basor, Thomas Bothner, Christophe Charlier, Percy Deift and Igor Krasovsky for their interest in this work and for many very useful comments and suggestions. We also thank the referees for their valuable remarks. R.~Gharakhloo acknowledges support by NSF-grant DMS-1700261. A.~Its acknowledges support by~NSF-grant
DMS-1700261 and by~Russian Science Foundation grant No.~17-11-01126.

\pdfbookmark[1]{References}{ref}
\LastPageEnding

\end{document}